\documentclass[acmsmall,balance=false]{acmart}
\usepackage{pbalance} % balances the last page automatically
\shrinkLastPage{355pt} % manual balance amount for the last page

\usepackage[utf8]{inputenc}

\usepackage{caption}
\usepackage{amsthm}
\usepackage{booktabs}
\usepackage{caption}

\usepackage{tikz}
\usepackage{tabularx}
\usepackage{booktabs}
\usepackage{tabulary}
\usepackage{multirow}
\usepackage{makecell}
\usepackage{subcaption}
\usetikzlibrary{shapes.geometric}

\usepackage{xcolor}
\usepackage{tabularx}
\usepackage{graphicx}
\usepackage{textcomp}
\usepackage{listings}
\usepackage{array}
\usepackage{flushend}
\usepackage{framed}
\usepackage{amsmath,amsfonts,multirow,stmaryrd,xspace,enumitem}
\usepackage{mathtools} % % amsmath 的增强版, bsmallmatrix依赖
\usepackage{amsthm}
\usepackage{caption}
\usepackage{pgfplots}
\pgfplotsset{compat=newest}
\usepackage{ellipsis, ragged2e, marginnote}
\usepackage[]{hyperref}
\usepackage{threeparttable}
\usepackage{marvosym}
\usepackage{balance}
\usepackage[linesnumbered,ruled,vlined]{algorithm2e}
\usepackage{bbm}
\usepackage{lineno}
\SetCommentSty{footnotesize}

\usepackage{booktabs}
\usepackage{multirow}

\newtheorem{theorem}{Theorem}
\newtheorem{lemma}{Lemma}

\newtheorem{definition}{Definition}

\def\boxit#1{\vbox{\hrule\hbox{\vrule\kern4pt
  \vbox{\kern1pt#1\kern1pt}
\kern2pt\vrule}\hrule}}

\usepackage{array}
\newcolumntype{?}{!{\vrule width 1pt}}

\usepackage{amsfonts}
\usepackage{textcomp}
\usepackage{xcolor}
\def\BibTeX{{\rm B\kern-.05em{\sc i\kern-.025em b}\kern-.08em
    T\kern-.1667em\lower.7ex\hbox{E}\kern-.125emX}}

\usepackage{amsopn}
 \usepackage{courier}
 \usepackage{multicol}
 \usepackage{verbatim}
 \usetikzlibrary {arrows.meta}
 \usepackage[dvipsnames]{xcolor}

\newcommand{\partitle}[1]{\medskip\noindent\textbf{#1.}}%

\def\eg{\textit{e.g.}, }
\def\ie{\textit{i.e.}, }

\newcommand{\func}[1]{\textsf{#1}}
\definecolor{darkBlue}{RGB}{0, 0, 115}

\newcommand{\darkBlue}[1]{\textcolor{darkBlue}{#1}}

\colorlet{PURPLE}{purple}
\colorlet{BROWN}{brown}
\colorlet{ORANGE}{orange}
\colorlet{BLUE}{blue}

\usepackage[noend]{algpseudocode}

\renewcommand\footnotetextcopyrightpermission[1]{} % removes footnote with conference information in first column
\AtBeginDocument{%
  \providecommand\BibTeX{{%
    Bib\TeX}}}

\setcopyright{cc}
\setcctype{by --or -- by-nc-nd}
\acmJournal{PACMMOD}
\acmYear{2026} \acmVolume{4}
\acmNumber{6} \acmArticle{367}
\acmMonth{12} \acmPrice{}
\acmDOI{10.1145/3856362}
\begin{document}

%% Include the 3 required article and appropriate history dates in your .tex file after after \begin{document} command and before the \maketitle
\received{17 April 2026}
\received[revised]{19 August 2026}
\received[accepted]{12 September 2026}

\title{Product Costing in Data Pipelines}


% \affiliation{%

% \author{Anonymous author(s)}

\author{Li Peng}
\affiliation{
  \institution{State Key Laboratory of Blockchain and Data Security, Zhejiang University; Alibaba Group}
  \country{China}
}
\email{jerry.pl@alibaba-inc.com}

\author{Jie You}
\affiliation{
  \institution{Alibaba Group}
  \country{China}
}
\email{youjie.yj@alibaba-inc.com}

\author{Feng Shi}
\affiliation{
  \institution{Central South University}
  \country{China}
}
\email{fengshi@csu.edu.cn}

\author{Yihang Wu}
\affiliation{
  \institution{State Key Laboratory of Blockchain and Data Security, Zhejiang University}
  \country{China}
}
\email{yhwu\_is@zju.edu.cn}

\author{Jianxin Wang}
\affiliation{
  \institution{Central South University}
  \country{China}
}
\email{jxwang@mail.csu.edu.cn}

\author{Weiran Liu}
\affiliation{
  \institution{Alibaba Group}
  \country{China}
}
\email{weiran.lwr@alibaba-inc.com}

\author{Xinyu Peng}
\authornote{Corresponding author.}
\affiliation{
  \institution{Zhejiang University; Alibaba Group}
  \country{China}
}
\email{steven.pengxy@alibaba-inc.com}

\author{Junhua Wang}
\affiliation{
  \institution{Alibaba Group}
  \country{China}
}
\email{junhua.wjh@alibaba-inc.com}

\author{Lei Zhang}
\affiliation{
  \institution{Alibaba Group}
  \country{China}
}
\email{zongchao.zl@taobao.com}

\author{Lin Qu}
\affiliation{
  \institution{Alibaba Group}
  \country{China}
}
\email{xide.ql@taobao.com}

\author{Xiaohui Yu}
\affiliation{
  \institution{York University}
  \country{Canada}
}
\email{xhyu@yorku.ca}

\author{Kui Ren}
\affiliation{
  \institution{State Key Laboratory of Blockchain and Data Security, Zhejiang University}
  \country{China}
}
\email{kuiren@zju.edu.cn}

\author{Jinfei Liu}
\affiliation{
  \institution{State Key Laboratory of Blockchain and Data Security, Zhejiang University; Hangzhou High-Tech Zone (Binjiang) Blockchain and Data Security Research Institute}
  \country{China}
}
\email{jinfeiliu@zju.edu.cn}

\renewcommand{\shortauthors}{Li Peng et al.}

\begin{abstract}
% In data-intensive enterprises, data pipelines are essential for transforming raw data into valuable data products. To provide a sound basis for data pricing, enterprises usually perform \textit{product costing}, which apportions the total cost of all resources (\eg CPU, Memory) consumed within pipelines to data products. However, product costing in data pipelines faces significant challenges not encountered in traditional physical pipelines, including the non-rivalrous nature of data that complicates the costing formulation, the prevalence of cyclic dependencies that distort the accuracy of results, and the massive scale of pipelines that leads to computational intractability.

% Large data-intensive enterprises spend millions of dollars monthly running data pipelines. Apportioning this cost back to individual products, known as \textit{product costing}, is a prerequisite for data pricing and pipeline optimization, yet it is fundamentally more difficult than costing in physical manufacturing. These changelles arise from the unique characteristics of data pipelines, including the non-rivalrous nature of data that complicates the costing formulation, the prevalence of cyclic dependencies that distort the accuracy of results, and the massive scale of pipelines that leads to computational intractability.

Large data-intensive enterprises spend millions of dollars monthly running data pipelines. Apportioning this cost back to individual products, known as \textit{product costing}, is a prerequisite for data pricing and cost recovery. However, this process is fundamentally more challenging than physical manufacturing costing due to the unique characteristics of data pipelines, including the non-rivalrous nature of data that complicates the costing formulation, the prevalence of cyclic dependencies that distort the accuracy of results, and the massive scale of pipelines that leads to computational intractability.

In this paper, for the first time, we formalize the problem of product costing in data pipelines as a \textit{manufacturing cost of data products (MCDP) problem} to overcome the non-rivalrous nature of data. To solve MCDP, we propose \func{CostApp}, a hybrid algorithm that achieves both exactness and scalability, even for cyclic and massive-scale pipelines. \func{CostApp} performs a dimensionality reduction by isolating cyclic dependencies via a feedback arc set (FAS). It aggregates costs onto a compact set of key vertices using a scalable iterative approach, transforming the intractable computational problem into a reduced linear system that can be solved exactly and efficiently by a basic matrix-based method. Experiments on both real-world data and synthetic data show that \func{CostApp} scales to process a production pipeline with over 28 million entities within an hour while preserving exactness, supporting pipelines up to $10^2\times$ larger than the direct matrix-based baseline.
\end{abstract}

\begin{CCSXML}
<ccs2012>
   <concept>
       <concept_id>10002951.10002952.10003219.10003215</concept_id>
       <concept_desc>Information systems~Extraction, transformation and loading</concept_desc>
       <concept_significance>500</concept_significance>
       </concept>
 </ccs2012>
\end{CCSXML}

\ccsdesc[100]{Information systems~Extraction, transformation and loading}


\keywords{Data Pipelines, Data Costing, Cost Apportionment, Data Pricing}

\maketitle
% \linenumbers

\section{Introduction}
\label{sec:intro}

Over the past decades, the pervasive adoption of data-driven decision-making has established data as a critical asset \cite{jones2024ai,dailydata2024}. With both data volume and demand rising sharply, deriving value from data increasingly depends on the ability to process and refine it at scale. To transform large amounts of raw data into valuable data products, large data-intensive enterprises (\eg Google, Meta, and Alibaba) have built extensive data pipelines \cite{sun2023presto,airflow,alipipeline,jacques2025unified,szlang2025workload}. These pipelines ingest diverse data sources and coordinate complex production dependencies to create a wide array of valuable data products that support business-critical scenarios, from operational decision-making to machine learning applications \cite{sun2023presto,szlang2025workload,zhao2022understanding,zhao2024cedar,liu2023optimizing}.

\begin{figure}
    \centering
    \includegraphics[width=0.6\linewidth]{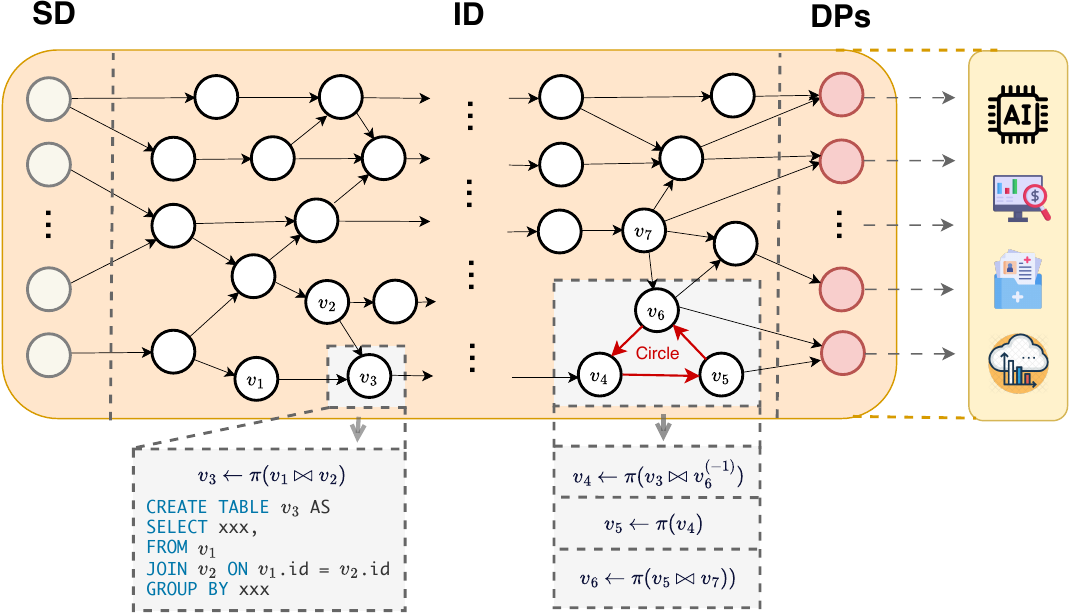}
    \vspace{-0.1in}
    \caption{An illustration of data pipelines. $\Join, \pi$ denotes a join  and project operation in SQL, respectively.
    % that integrates multiple input data entities into a target entity. 
    $v_i^{(-1)}$ represents the version of data entity $v_i$ from the previous day, while entities without superscripts denote the current day's version.}
    
    \label{fig:data_pipeline}
    \vspace{-0.2in}
\end{figure}

Data pipelines in data-intensive enterprises are characterized by \textit{consolidation} and \textit{periodicity}, incurring high resource cost and presenting requirements for costing. Data is \textit{non-rivalrous in consumption} \cite{jones2020nonrivalry,pei2020survey,zhang2024survey}, meaning it can be consumed repeatedly without depletion. To leverage this, enterprises integrate processing from numerous business domains into a \textit{consolidated} digital supply chain, allowing frequently accessed data to be shared across multiple downstream products to maximize reuse value. Consolidated pipelines weave thousands of lineages together, creating a complex, many-to-many dependency graph.
As shown in Figure \ref{fig:data_pipeline}, source data (SD) is progressively refined through transformations into intermediate data (ID) and, finally, into valuable data products (DPs) ready for use across various businesses. These transformations include SQL operations (\eg $v_3 \leftarrow \pi (v_1 \Join v_2)$ in Figure \ref{fig:data_pipeline}), Map-reduce jobs or machine learning tasks (\eg feature extraction). These consolidated pipelines consume massive resources (\eg, CPU, Memory), presenting requirements for cost accounting.
Furthermore, these pipelines operate with workload \textit{periodicity}, typically executing recurrently (\eg daily extract-transform-load (ETL) jobs) to keep data up-to-date. This continuous and cumulative resource consumption means that even modest unit costs aggregate into substantial financial expenditures over time. For example, at Alibaba, data pipelines incur resource costs equivalent to millions of dollars monthly, a financial burden that continues to escalate as illustrated in Figure~\ref{fig:graph_growth}. Consequently, the interplay of complex consolidated dependencies and long-term periodic consumption makes cost visibility a necessity for operating data pipelines.

\begin{figure}[t!]
    \centering
    \begin{minipage}{0.4\textwidth}
        \centering
        \includegraphics[width=1.0\linewidth]{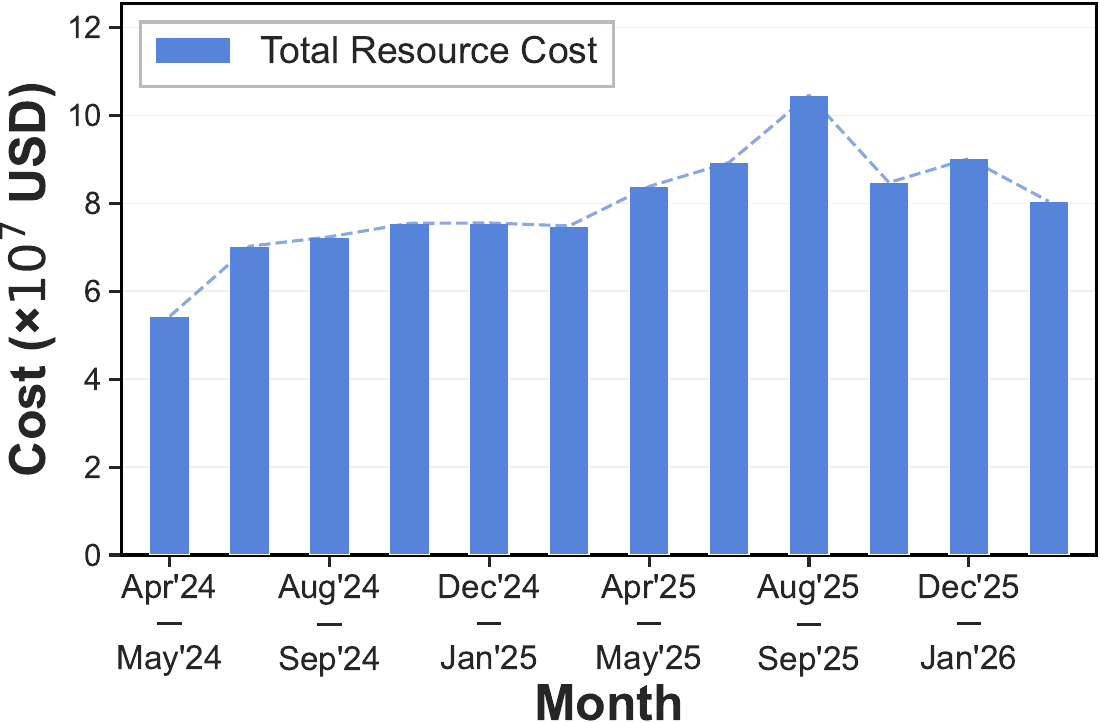}
        
        \vspace{-0.1in}
        \caption{Resource cost of data pipelines in Alibaba.} 
        \label{fig:graph_growth}
        \vspace{-0.2in}
    \end{minipage}
    \hspace{0.2in}
    \begin{minipage}{0.42\textwidth}
        \centering
        \includegraphics[width=1.0\linewidth]{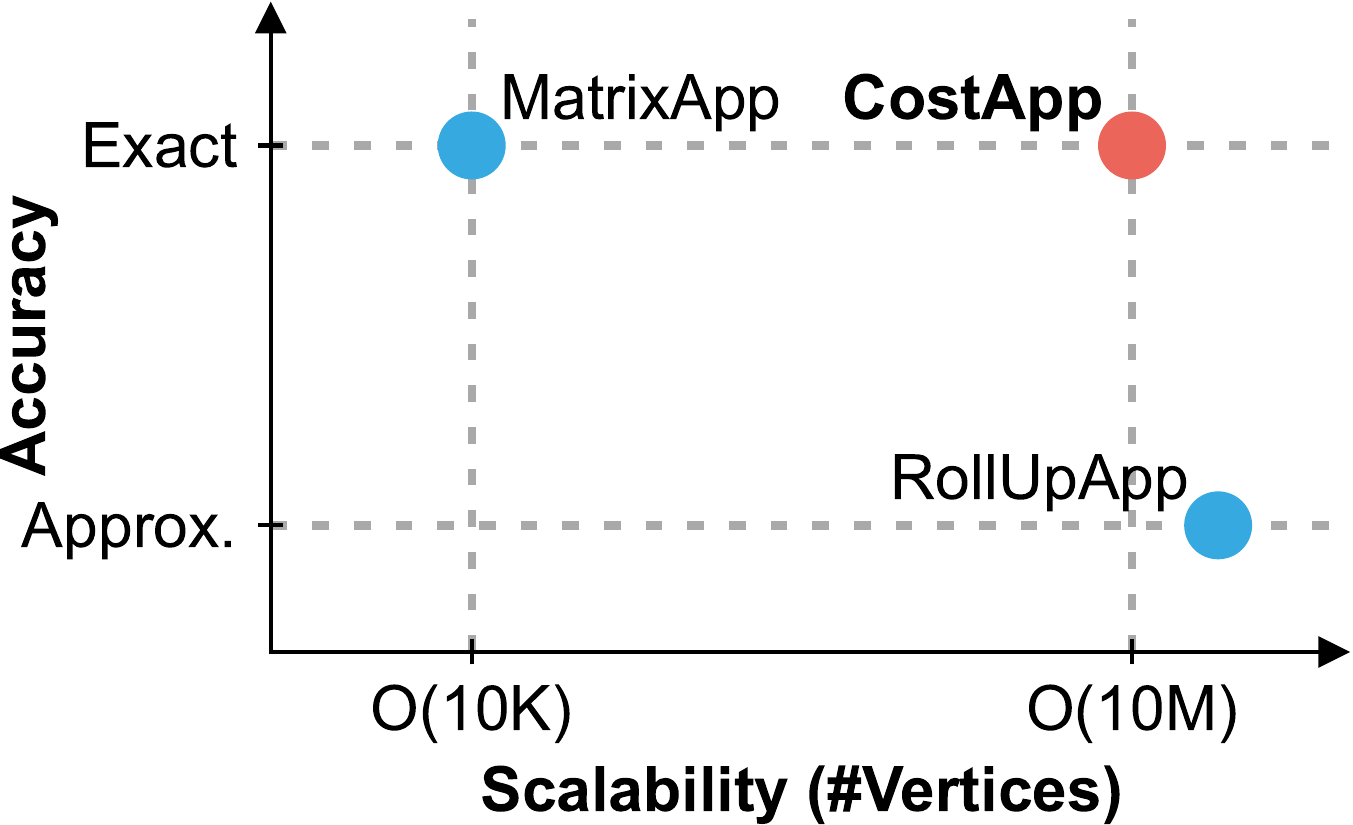}
        \vspace{-0.25in}
        \caption{Comparison of approaches.}
        \label{fig:general_comparison}
        
        \vspace{-0.05in}
    \end{minipage}
\end{figure}

% \begin{figure}[t!]
%     \centering
%     \begin{minipage}{0.23\textwidth}
%         \centering
%         \includegraphics[width=1.0\linewidth]{figs/graph_growth.pdf}
        
%         \vspace{-0.15in}
%         \caption{Growth of resource cost of data pipelines in Alibaba.} 
%         \label{fig:graph_growth}
%         \vspace{-0.2in}
%     \end{minipage}
%     \hfill
%     \begin{minipage}{0.23\textwidth}
%         \centering
%         \includegraphics[width=1.0\linewidth]{figs/general_comparison.pdf}
%         \vspace{-0.25in}
%         \caption{Approaches for product costing in data pipelines.}
%         \label{fig:general_comparison}
%         \vspace{-0.2in}
%     \end{minipage}
% \end{figure}

While resource consumption is distributed across all data entities (SD, ID, DPs), business value is delivered \textit{exclusively} through the final DPs (via monetization or usage). Thus, DPs constitute the natural unit for financial accountability. We term this \textit{product costing in data pipelines}, defined as apportioning the total costs of all resources consumed within pipelines to final DPs. This costing method provides the necessary basis for data pricing \cite{cong2022data,pei2020survey} and cost recovery, where the exactness of the results is critical. Particularly under cost-based pricing models \cite{pei2020survey}, inaccuracies in cost attribution at the DP level can propagate to pricing decisions, potentially leading to systematic mispricing. When aggregated across millions of entities and millions of dollars in costs, such deviations may translate into non-trivial financial losses. Despite its importance to modern data-intensive enterprises, this problem remains largely unexplored in the existing literature.

\partitle{Challenges} Product costing in data pipelines presents unique challenges not encountered in physical pipelines, as summarized in Table \ref{tab:challenges}.

First, the \textit{non-rivalrous nature of data consumption} \cite{jones2020nonrivalry,pei2020survey,zhang2024survey} fundamentally challenges cost modeling. In physical manufacturing pipelines, intermediate outputs are typically rivalrous physical goods. The costs embedded in a unit of a physical good are fixed (typically as its marginal cost) and fully transferred downstream upon its consumption. For example, when a car part is used in assembling a vehicle, its cost is fully attributed to that vehicle. Thus, consumption-based costing in physical pipelines naturally upholds the principle of \textit{completeness} \cite{completeness1} (or \textit{break-even} \cite{young1985cost}), meaning \textit{all} consumed resources are fully accounted for within the costing \textit{without} duplication. However, in data pipelines, data is not depleted when used, resulting in (nearly) zero marginal cost for each additional consumption \cite{pei2020survey}. Applying traditional costing creates a dilemma: allocating cost based on marginal cost \cite{shapiro1998versioning} leads to severe under-costing, while charging the full cost to every consumer results in massive over-costing, both of which violate completeness. This difficulty challenges the formulation of a cost model that guarantees the completeness of costing.

Second, the workload periodicity of data pipelines leads to a time-versioned nature of data entities, inevitably introducing \textit{cyclic dependencies} in the lineage graph. For instance, in Figure \ref{fig:data_pipeline},
the generation of $v_4$ relies on joining $v_3$ with $v_6^{(-1)}$, the version of $v_6$ from the previous day ($v_4 \gets \pi(v_3 \Join v_6^{(-1)})$). This structural dependency among  $v_4,v_5,v_6$ forms a feedback loop or cycle in the lineage graph. Such cycles are pervasive in production pipelines. They cause basic costing methods based on cost roll-up strategy \cite{costrollup2024} (see \S\ref{sec:simple_app}) to fail, as costs circulate indefinitely within loops rather than being fully attributed to final data products. Without proper handling, this leads to significant calculation errors that may mislead pricing strategies and incur revenue losses.

% Second, the workload periodicity of data pipelines leads to a time-versioned nature of data entities, inevitably introducing \textit{cyclic dependencies} in the lineage graph. For instance, a daily user profile update often requires joining new behavioral data with its own previous snapshot, structurally forming a self-dependency or cycle in the lineage graph. Beyond self-loops, complex multi-hop cycles also arise, such as the loop formed by vertices $v_4$, $v_5$, and $v_6$ in Figure \ref{fig:data_pipeline}, where $v_4$ relies on $v_6^{(-1)}$, the version of $v_6$ from the previous day ($v_4 \gets \pi(v_3 \Join v_6^{(-1)})$). Such cycles are pervasive in production pipelines. They cause basic costing methods based on cost roll-up strategy \cite{costrollup2024} (see \S\ref{sec:simple_app}) to fail, as costs circulate indefinitely within loops rather than being fully attributed to final data products. Without proper handling, this leads to significant calculation errors that may mislead pricing strategies and incurse revenue losses.

Third, the consolidation of data pipelines results in a \textit{massive scale} of data entities. While traditional physical supply pipelines typically manage a limited number of entities (\eg thousands of intermediate units), a consolidated data pipeline at large enterprises often comprises tens of millions of data entities with complex dependencies \eg 28 million vertices in Alibaba. At this scale, exact costing solutions may become computationally intractable. Specifically, direct algebraic solutions (see \S\ref{sec:matrix_app}) typically require matrix inversion with $O(|V|^3)$ time complexity. For graphs with millions of vertices, such cubic complexity is prohibitive. Therefore, a practical cost calculation method must be inherently scalable and designed to operate efficiently in a distributed environment.

\vspace{-0in}
\begin{table}[t!]
    \centering
    \vspace{-0.1in}
    \caption{Comparison of physical and data pipelines.}
    \vspace{-0.1in} 
    \resizebox{0.8\columnwidth}{!}{
        \begin{tabular}{p{3cm}|p{3cm}|p{3cm}|p{5cm}}
            \toprule
             & \textbf{Physical pipelines}  & \textbf{Data pipelines} &  \textbf{Unique challenges}\\
            \toprule

           % , and the massive scale of pipelines that leads to computational intractability.

            \textbf{Consumption feature of outputs} & rivalrous & non-rivalrous & consumption-based costing not applicable, complicates cost modeling \\
            \midrule
            \textbf{Type of dependencies} & acyclic & cyclic & result accuracy of cost roll-up strategy \cite{costrollup2024} is distorted \\
            \midrule
            \textbf{Scale} & small & large & exact costing solutions may be computationally intractable. \\ 

            \bottomrule
        \end{tabular}
    }
    \vspace{-0.2in}
    \label{tab:challenges} 
\end{table}

% \begin{figure}
%     \centering
%     \includegraphics[width=0.6\linewidth]{figs/general_comparison.pdf}
%     \caption{Approaches for product costing in data pipelines.}
%     \label{fig:general_comparison}
%     \vspace{-0.2in}
% \end{figure}

% \begin{figure}
%     \centering
%     \includegraphics[width=0.6\linewidth]{figs/graph_growth.pdf}
%     \caption{Growth of number of entites and system costs.}
%     \label{fig:graph_growth}
%     \vspace{-0in}
% \end{figure}

\partitle{Contributions} To address these challenges, this paper makes the following contributions:

First, to resolve the dilemma caused by non-rivalrous data consumption, we establish a cost model for data pipelines and formalize product costing as the MCDP problem. Unlike the consumption-based approach, we adopt a ratio-based cost model governed by the cost conservation property (Definition \ref{def:cost_conservation} in \S\ref{sec:problem_formulation}). This formulation guarantees completeness of product costing, ensuring precise cost apportionment to final data products without the under-costing or duplication pitfalls of traditional approaches. To initially solve the MCDP problem, we propose two basic solutions: \func{RollUpApp}, a distributed, cost roll-up-based algorithm that scales to large pipelines but only yields approximate results, and \func{MatrixApp}, which provides an exact solution via matrix inversion but is limited to small pipelines. These solutions highlight the fundamental trade-off between exactness and scalability, serving as building blocks for our hybrid approach.

Second, we propose \func{CostApp}, a hybrid algorithm that achieves both exactness and scalability by strategically integrating our basic methods. The core idea is to decouple the computational complexity by isolating cyclic dependencies. Specifically, \func{CostApp} first identifies a feedback arc set (FAS) and strategically inserts auxiliary \textit{sink} and \textit{source} vertices on these arcs to isolate these cycles. It then employs the scalable \func{RollUpApp} on the resulting acyclic graph to efficiently aggregate costs onto these auxiliary sink vertices and the final data products. This aggregation transforms the problem into a reduced linear system while preserving the final manufacturing cost of data products, which is then solved exactly using \func{MatrixApp}. \func{CostApp} effectively performs a \textit{dimensionality reduction}, shifting the bottleneck from inverting a massive $|V| \times |V|$ matrix (where $|V|$ is millions) to inverting a tiny matrix sized proportional to the feedback arcs ($<0.05\%$ of $|V|$, see \S\ref{sec:exp_realworld}). It allows us to solve the MCDP problem with exact results and high scalability. A comparison that highlights its strengths is outlined in Figure \ref{fig:general_comparison}.

Third, we implement \func{CostApp} and conduct comprehensive experiments on both real-world data from Alibaba and synthetic data. Results demonstrate that \func{CostApp} successfully reduces a computationally prohibitive problem to a solvable one, achieving both result exactness and high scalability that outperform baselines. It successfully conducts product costing for production pipeline graphs with over 28 million entities in an hour, and has been deployed in the production environments of Alibaba for costing and pricing data for years.

To the best of our knowledge, we are the first to formalize the problem of product costing in data pipelines and to propose a scalable, exact solution. Our contributions are summarized as follows:
\begin{itemize}
    \item We establish a data cost model for data pipelines and formulate product costing as the MCDP problem. Then, we propose two basic solutions for this problem: one exact for small pipelines and one approximate for large ones.
    \item We propose \func{CostApp}, a novel hybrid algorithm that achieves both exactness and scalability for the MCDP problem.
    \item We validate the superiority of \func{CostApp} through extensive experiments on both large-scale industrial and synthetic data. Our code is released at \url{https://github.com/jerryp1/CostApp}.
    % We release both the industrial dataset and our implementation as open-source resources (\url{https://anonymous.4open.science/r/CostApp-E50C/}) to facilitate further research.

\end{itemize}

% \partitle{Recovering data cost is important}
% % Why we need to recover data cost? Recovering data cost is crucial for organizations to understand the value of their data assets, optimize resource allocation, and ensure the sustainability of their data pipelines. By accurately assessing and recovering data costs, organizations can make informed decisions about data investments, improve operational efficiency, and enhance the overall value derived from their data products.

\section{Preliminaries}
\label{sec:preliminary}

In this section, we introduce the background knowledge of data pipelines and Pregel-based distributed computation.

% todo 说明我们的方法独立于pipeline中的具体计算框架
\subsection{Data Pipelines}

In this paper, we focus on \textit{consolidated} data pipelines that consolidate numerous data transformation tasks (\eg SQL queries, MapReduce jobs) across diverse business domains for centralized management. Such pipelines have been extensively deployed by major data-intensive enterprises \cite{sun2023presto,jacques2025unified,alipipeline,szlang2025workload}. Their primary objective is to generate high-quality DPs through a sequence of data transformation operations. Typically, these jobs execute in batch mode: reading source data from storage, performing transformations (\eg cleaning, joining, aggregation, and feature extraction), and persisting the results back to storage. Formally, we denote a processing operation as $o$, which consumes a set of input data entities to produce a single output data entity. For instance, regarding entity $v_3$ in Figure \ref{fig:data_pipeline}, $o$ corresponds to a SQL join-projection operation $v_3 \gets \pi(v_1\bowtie v_2)$.
Within these pipelines, three types of data entities are involved:

\begin{itemize}
    \item Source data (SD). The starting point of data flow. Typically, they represent raw data collected from users or migrated from other data stores.
    \item Intermediate data (ID). The entities that are the output of the computation process but are not yet ready to be used as a product. They are materialized to enable data reuse and reduce redundant computation.
    \item Data products (DPs). The final data entities that are ready to be used as products and will not be processed further. 
\end{itemize}

% 原版
% Each data entity $v_i$ is output by an operation $o$ that performs on a set of input data entities $V_i^{in}$, \ie $v_i \gets o(V_i^{in})$. $o$ denotes operations for generating $v_i$, $V^{in}_i \gets U_G^1(v_i)$. For example, for ID entity $v_3$ in Figure \ref{fig:data_pipeline}, $o$ represents a join-projection SQL operation $v_3 \gets \pi(v_1\bowtie v_2)$, where input set $V_3^{in} \gets \{v_1, v_2\}$. For DS entity $v_7$, $o$ represents a data collection operation $v_7 \gets \textsf{Collect()}$ where $\textsf{Collect()}$ collects raw data from outside sources, \eg a transactional database that records user daily behaviors \cite{}. In this case, $V_{7}^{in} \gets \emptyset $ since $v_7$ is a starting point of the pipeline.

\partitle{Model data pipelines as a weighted directed cyclic graph (DCG)}
To model the evolution of data in a structured and traceable manner, we represent the data pipelines using a data lineage \cite{ikeda2009data,backes2015data,cui2003lineage,jacques2025unified}, a graph-based model that captures how data progresses from its origin to its final form.
Specifically, the pipelines are modeled as an edge-weighted directed graph $G = (V, A, w)$, where

\begin{itemize}
    \item $V$ is a vertex set that denotes the set of data entities, such as SD, ID and DPs. The vertices with out-degree 0 (called \textit{leaves}) represent DPs. The ones with in-degree 0 (called \textit{roots}) represent SD, and the ones with both in-degree and out-degree at least 1 (called \textit{internal vertices}) represent ID. 
    \item $A$ is a directed arc set that represents the transformation flows that produce new data entities from existing ones.
    % \item $C: V \to \mathbb{R}^+$ is a function that assigns a cost to each vertex. The cost of vertex $v$ is initialized as the monetary resource usage during the execution of the operation that produces $v$, which we call the operational cost (see \S\ref{sec:problem_formulation}). $C(v)$ may change during cost propagation processes (see \S\ref{sec:simple_app}).
    \item $w: A \to [0, 1]$ is a function that assigns a static cost allocation ratio to each arc, as detailed in \S\ref{sec:problem_formulation}.
\end{itemize}

Costing is performed on a billing-period snapshot whose vertices represent \emph{logical data entities}, not individual runs. All versions of an entity produced within the period are represented by one vertex, and their measured resource costs are summed into that vertex's direct cost (see $\mathcal{C}_{dc}(v)$ in \S\ref{sec:problem_formulation}). For example, under monthly billing, the 30 daily versions of table $v$ contribute $\mathcal{C}_{dc}(v)=\sum_{d=1}^{30}\mathcal{C}_{dc}(v^{(d)})$ to the same vertex. Dependencies observed during the period form the snapshot's arcs. This snapshot therefore preserves the billed cost and dependency structure needed for allocation. As discussed in \S\ref{sec:intro}, a dependency on an earlier version appears as a feedback arc, so $G$ is modeled as a directed \textit{cyclic} graph (DCG). Other commonly used notations are summarized in Table \ref{tab:notations}.

% Crucially, $G$ is a directed \textit{cyclic} graph (DCG).  This necessitates the use of a DCG abstraction rather than a directed acyclic graph (DAG).  ; a dependency on a prior version therefore appears as a feedback arc.

\vspace{-0in}
\begin{table}[ht]

    \centering
    \vspace{-0in}
    \caption{Notation table.}
    \vspace{-1em}
    \resizebox{0.8\columnwidth}{!} {
    \begin{tabular}{c|c}
        \toprule
        \textbf{Notations}  & \textbf{Description}  \\
        \midrule
        $L(G),R(G)$ and $I(G)$ & leaves, roots, and internal vertices \\
        % $U_G(v)$ (resp., $D_G(v)$) & set of all upstream (resp., downstream) vertices of $v$ \\
        $U_G^1(v)$ (resp., $D_G^1(v)$) & set of all 1-hop upstream (resp., 1-hop downstream) vertices of $v$ \\
        $d^+_G(v)$ (resp., $d^-_G(v)$) & in-degree (resp., out-degree) of $v$  \\
        \bottomrule
    \end{tabular}
    }

    \vspace{-0.2in}
    \label{tab:notations}
\end{table}

% 这里要体现，pregel的实现也能很容易转换成其他实现方式，我们的平台选择不影响方案的通用性。
% 以下引用自Relationship Queries on Large graphs using Pregel
\subsection{Pregel-based Distributed Graph Computation}
% We process large-scale data graphs with millions of vertices and edges and therefore 

Our solutions adopt the vertex-centric, distributed Pregel model~\cite{malewicz2010pregel}. The graph is partitioned across workers, while each vertex stores only local state, such as an accumulated cost or a label, and exchanges messages with adjacent vertices. Execution proceeds in synchronized iterations called \emph{supersteps}. In each superstep, a user-defined \func{compute()} runs in parallel on every active vertex. It consumes messages sent in the preceding superstep, updates the vertex state, and sends messages for the next superstep. A vertex may call \func{voteToHalt}() when it has no work and is automatically reactivated by a later message. The job terminates when all vertices are inactive and no messages are in transit.

This model matches our computation directly. In \func{RollUpApp} (see \S\ref{sec:simple_app}), a message from $u$ to $v$ carries $w(u,v)C(u)$. Vertex $v$ sums all incoming messages and forwards the resulting cost. For a chain $u\!\to\!v\!\to\!p$ stored on different workers, $u$ sends its share in one superstep, $v$ aggregates it in the next, and the DP $p$ receives the final share in the following superstep. \func{FindFAS} follows the same message-driven process with labels instead of costs (see \S\ref{sec:find_fas}), while \func{GetMatrix} propagates keyed cost maps (see \S\ref{sec:get_matrix}). Thus, graph partitioning determines where each local computation runs but does not change the algorithm or its result.

\section{Problem Formulation and Basic Solutions}
\label{sec:problem_formulation}
In this section, we establish the data cost model in data pipelines (\S\ref{sec:data_cost}) and formalize product costing as the MCDP problem (\S\ref{sec:mcp_problem}). Following this, we introduce two basic solutions for the MCDP problem with different limitations (\S\ref{sec:basic_solutions}).
\subsection{Data Cost in Data Pipelines}
\label{sec:data_cost}
To account for the non-rivalrous nature of data, we establish a structured cost model for any data entity $v$, defining its manufacturing cost as the sum of its own direct cost and allocated upstream costs.

% 需要解释一下，实际中我们的方案本身就是一个周期性计算任务
\partitle{Direct cost (DC)}
As stated in \S\ref{sec:intro}, the manufacturing of data entities in pipelines relies on continuous processes involving periodic computations and ongoing storage. Consequently, the production and maintenance of any data entity $v$ incur direct resource expenditures, including CPU usage, memory usage, and storage usage \cite{lyu2022fine,zhang2024survey,srivastava2024saving,wu2019cloud}. We define the \textit{direct cost}, denoted as $\mathcal{C}_{dc}(v)$, as the quantified resource cost incurred \textit{solely} by producing and maintaining $v$:

% it incurs direct resource fees of CPU usage, memory usage, disk I/O, and storage usage for any data entity $v$ \cite{lyu2022fine,zhang2024survey,srivastava2024saving,wu2019cloud}.

% producing and maintaining each data entity $v$ in data pipelines involves periodic re-computation and storage, it incurs direct resource fees of CPU usage, memory usage, I/O, and storage usage \cite{lyu2022fine,zhang2024survey,srivastava2024saving,wu2019cloud}. 

% We define the direct cost $\mathcal{C}_{dc}(v)$ as the quantified resource consumption incurred \textit{solely} by the operation $o$ that produces $v$. This represents the fundamental cost injected into the system.
\begin{equation}
\label{equ:dc}
    \mathcal{C}_{dc}(v) = \mathcal{C}_{CPU}(v) + \mathcal{C}_{Mem}(v) + \mathcal{C}_{Sto}(v)
\end{equation}
where $\mathcal{C}_{CPU}$, $\mathcal{C}_{Mem}$, $\mathcal{C}_{Sto}(v)$ denote the respective monetary costs of utilized resources. This value represents the fundamental cost injected into the system. Accordingly, the total cost injected into the pipelines is thus $\mathcal{C}_{total} = \sum_{v \in V(G)} \mathcal{C}_{dc}(v)$. 

Importantly, the specific composition of $\mathcal{C}_{dc}(v)$ is \textit{orthogonal} to our target product costing problem and solutions. Our model is agnostic to the underlying billing mechanisms and can be flexibly adapted to diverse real-world billing policies \cite{wu2019cloud}.

% (\eg \$/core-hour, \$/GB transferred).

\partitle{Manufacturing cost (MC)} 
The \textit{manufacturing cost} of an entity $v$, denoted as $\mathcal{C}_{mc}(v)$, is defined as the total cost of \textit{all} resources consumed in manufacturing entity $v$ across its entire production lineage. It comprises the entity's own direct cost and the apportioned costs inherited from its upstream inputs. Formally, for an entity $v$ with upstream inputs $U_G^1(v)$, the manufacturing cost is defined recursively as:
\begin{equation}
\label{equ:ma}
    \mathcal{C}_{mc}(v) = \mathcal{C}_{dc}(v) + \sum_{u \in U_G^1(v)} w(u,v) \cdot \mathcal{C}_{mc}(u)
\end{equation}

Here, $w(u,v)$ is the \textit{cost allocation ratio}, representing the fraction of $u$'s manufacturing cost allocated to $v$. The second term, $ \sum_{u \in U_G^1(v)} w(u,v) \cdot \mathcal{C}_{mc}(u)$, represents the \textit{input cost} $\mathcal{C}_{ic}(v)$, capturing the cumulative costs inherited from direct upstream dependencies. The following example demonstrates this recursive computation.

\begin{example}
Consider $v_1$, $v_2$, and $v_3$ in Figure \ref{fig:data_pipeline}. Let $w(v_1,v_3) = 1$, $w(v_2,v_3) = 0.6$, $\mathcal{C}_{mc}(v_1)=3, \mathcal{C}_{mc}(v_2)=5, \mathcal{C}_{dc}(v_3)=2$. Then, we have $\mathcal{C}_{ic}(v_3) = 1\times3+0.6\times5=6$, $\mathcal{C}_{mc}(v_3)= \mathcal{C}_{ic}(v_3) + \mathcal{C}_{dc}(v_3) =8$.
\end{example}

\partitle{Cost allocation ratio $w(u,v)$}
As discussed in \S\ref{sec:intro}, consumption-based costing fails to guarantee \textit{completeness} of costing for non-rivalrous data.
% , which means the total cost of $u$ may not be fully accounted for without duplication. 
To overcome this limitation, a cost allocation ratio $w(u,v)$ is introduced to apportion the manufacturing cost of $u$ to its downstream consumer $v$. Rather than relying on physical consumption, this ratio-based strategy distributes costs according to relative contribution. To ensure completeness (see Theorem \ref{thm:cost_conservation}), we enforce the following conservation property.

% This ratio-based strategy allows for distributing costs in a way that reflects the relative contribution of each upstream entity to the downstream entities, rather than relying on the assumption of rivalrous consumption. 

% Unlike traditional consumption-based costing, our ratio-based strategy is consumption-independent, accommodating the non-rivalrous nature of data. Crucially, to ensure that the total cost of $u$ is fully and precisely distributed across its lineage without duplication (\textit{completeness}, see Theorem \ref{thm:cost_conservation}), we enforce the following conservation property.

\begin{definition}[Cost Conservation]
\label{def:cost_conservation}
A data pipeline graph $G=(V,A,w)$ satisfies the \textit{cost conservation} property if the sum of cost allocation ratios for any non-leaf node $u \in V(G) \setminus L(G)$ equals unity:
$$
\sum_{v \in D_G^1(u)} w(u,v) = 1
$$
\end{definition}

The ratio definition that satisfies the above property is flexible and can be adapted to various allocation strategies, including equal-split, usage-based \cite{wu2019cloud,srivastava2024saving}, or value-based strategies \cite{aadithya2010efficient,jia2019towards,wei2020efficient}. For instance, under a usage-based strategy, the weight can be derived from the normalized data volume read by downstream vertices \cite{srivastava2024saving}:
$
w(u,v) = \frac{\text{read\_volume}(v \leftarrow u)}{\sum_{k \in D_G^1(u)} \text{read\_volume}(k \leftarrow u)}
$
where $\text{read\_volume}(v \leftarrow u)$ denotes the volume of data read from $u$ to produce $v$.

The choice of $w$ determines the resulting allocation, but it does not affect the correctness or asymptotic complexity of our algorithms. This holds for any fixed, nonnegative weights satisfying Definition~\ref{def:cost_conservation}. We justify this claim in \S\ref{sec:analysis}.

\subsection{Product Costing in Data Pipelines as an MCDP Problem}
\label{sec:mcp_problem}
% By the definition of MC, the MC of a single DP inherently captures the cumulative expenditure from all upstream operations and dependencies required for its creation.
Based on the definition of MC, we demonstrate that product costing can be reduced to solving the MCDP problem.

\partitle{MC of DPs fully captures the total cost of the pipelines}
The goal of product costing is to apportion the total direct cost of the pipeline to its final DPs so that every cost is counted exactly once. We call this requirement \emph{completeness}.
A simple example illustrates why completeness is necessary. Suppose an intermediate table $A$ costs 100 and is reused by products $B$ and $C$. Charging only the near-zero marginal reuse cost allocates none of the original 100, whereas charging the full cost to each product allocates 200 in total. The former under-allocates and the latter duplicates cost. Ratios satisfying $w(A,B)+w(A,C)=1$ instead distribute exactly 100, regardless of how the ratio is chosen.

Because the manufacturing cost $\mathcal{C}_{mc}(p)$ of each DP includes both its direct cost and all inherited upstream costs, computing these costs for all DPs achieves the product-costing goal above. The following theorem formalizes this connection.
% The following theorem formally guarantees that the sum of the DPs' MC equals the total direct cost injected into the system.

% where the cost allocation ratios for any non-leaf node $u \in V(G) \setminus L(G)$ satisfy the conservation property $\sum_{v \in D_G^1(u)} w(u,v) = 1$,

\begin{theorem}[Completeness of Product Costing]
\label{thm:cost_conservation}
Given a data pipeline graph $G=(V,A,w)$ that satisfy the cost conservation property, the total manufacturing cost of data products (leaf nodes) is equal to the total direct cost injected into the system. Formally:
$$
\sum_{p \in L(G)} \mathcal{C}_{mc}(p) = \sum_{v \in V(G)} \mathcal{C}_{dc}(v)
$$
\end{theorem}

The proof is demonstrated in Appendix \ref{appendix:proof_cost_conservation}. Theorem \ref{thm:cost_conservation} shifts product costing to computing the MC of all DPs, which we term the \textit{manufacturing cost of data products (MCDP) problem}.

\begin{definition}[The MCDP Problem]
    Given an edge-weighted DCG $G=(V,A,w)$, a DC vector $\mathbf{d}$ where $\mathbf{d}_i = \mathcal{C}_{dc}(v_i)$ for all $v_i \in V$, the MCDP problem is to compute the MC vector $\mathbf{c}^{L} \in \mathbb{R}^{|L(G)|}$ such that each $\mathbf{c}^L_i$ represents the MC of a leaf node (DP) $p_i \in L(G)$.
\end{definition}

\partitle{Desired goals} We define the desired goals for an algorithm that solves the MCDP problem as follows.
\begin{itemize}
    \item \textbf{Exactness}. The algorithm must compute the MC $\mathcal{C}_{mc}(p)$ for each data product $p \in L(G)$ exactly as defined by the cost model (Eq. \ref{equ:ma}).
    %  satisfying the completeness principle of product costing (Theorem~\ref{thm:cost_conservation}).
    \item \textbf{Scalability} The algorithm must scale to large data pipelines (\eg millions of vertices) and complete the computation within a practical time window (\eg for daily billing cycles). 
\end{itemize}

% 这里的 daily billing cycles在后面实验中也应体现

% Since manufacturing costs serve as the baseline for data pricing and pipeline optimization \cite{wu2019cloud}, approximate results can lead to under-pricing (financial loss) or misallocation of budget during cost optimization. 

The necessity of these goals stems from both economic and technical imperatives. 
First, the exactness of the MCDP solution serves as the foundation for the completeness of product costing. It is fundamental to financial accountability, particularly for data pricing and pipeline optimization. Given that cost-based pricing is prevalent in practice \cite{wu2019cloud,amaral2019factors,mattos2021pricing} and resource optimization is critical, any inaccuracy risks distorting pricing strategies, eroding revenue, and misallocating budgets. Second, scalability is essential for handling enterprise-level data pipelines, which often comprise millions of entities and require timely cost reporting.

\subsection{Basic Solutions}
\label{sec:basic_solutions}
We present two basic solutions for the MCDP problem, \func{RollUpApp} and \func{MatrixApp}. While neither fully satisfies both exactness and scalability goals simultaneously, they serve as fundamental building blocks for our proposed hybrid approach (see \S\ref{sec:the_improved_algorithm}).

\subsubsection{The \func{RollUpApp} Algorithm}
\label{sec:simple_app}

A straightforward idea to solve the MCDP problem is to employ the cost roll-up strategy \cite{costrollup2024,engh18system,wentker2019bottom,cunnama2016using}, which means aggregating costs bottom-up along the dependency hierarchy. Specifically, in data pipeline lineage graphs, this involves propagating costs from upstream to downstream vertices until accumulation at the leaves. Based on this, we introduce \func{RollUpApp}, a scalable algorithm for implementing this propagation.

% A basic idea of solve the MCDP problem is to propogate costs from upstream vertices to downstream ones iteratively until all costs reach the leaves. Based on this idea, we introduce \func{RollUpApp}, an iterative, vertex-centric algorithm. The high-level idea is to allow costs to `flow' downstream until they reach the leaves, with the accumulated cost of the leaves representing the manufacturing cost once all costs have been fully transferred. \func{RollUpApp} is designed for scalability on large graphs. However, due to the existence of cycles, it sacrifices the guaranteed exactness of the solution for an approximate, convergent result.

The core of \func{RollUpApp} is an iterative cost propagation mechanism based on a stateful cost variable $C(v)$ and a cost roll-up operation \func{CR}. Here, $C(v)$ represents the cost currently held by $v$, and initialized with its direct cost \ie $C(v) \leftarrow \mathcal{C}_{dc}(v)$. In each iteration, \func{CR} is applied to non-leaf vertices to propagate their costs downstream. $\func{CR}(v)$ is a two-step state update:

% To facilitate cost propagation based on $C(v)$, a cost apportion operation \func{CA} is defined as below to allow cost ``flow'' among vertices.

% \partitle{The CA operation}
% The cost apportion operation \func{CA} transfers a vertex's accumulated cost to its immediate downstream neighbors. Formally, for any non-leaf vertex $v$, \func{CA}$(v)$ is a two-step state update:

\begin{enumerate}
    \item \textbf{Distribute Cost:} For each direct downstream neighbor $v' \in D_G^1(v)$, its cost is updated by adding a proportional share of $v$'s cost:
    $C(v') \leftarrow C(v') + w(v,v') \cdot C(v)$.
    \item \textbf{Reset Cost:} The cost at $v$ is set to zero, signifying that its contribution has been fully passed on: $C(v) \leftarrow 0$.
\end{enumerate}

An example of a \func{CR} operation is shown in Figure \ref{fig:ca_example}.

\begin{figure}[t!]
    \centering
    \begin{subfigure}[b]{0.2\textwidth} 
        \centering
        \includegraphics[width=\linewidth]{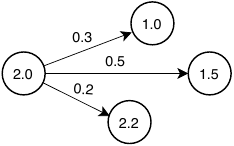} 
        \caption{Before \func{CR}.}
        \label{fig:ca_before}
    \end{subfigure}
    \hspace{0.2in}
    % \hfill % Add spacing between subfigures
    \begin{subfigure}[b]{0.2\textwidth} % Corrected: Use [b] for bottom alignment
        \centering
        \includegraphics[width=\linewidth]{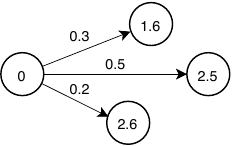} 
        \caption{After \func{CR}.}
        \label{fig:ca_after}
    \end{subfigure}
    \vspace{-0.1in}
    \caption{An example of a cost roll-up operation \func{CR}.}
    \vspace{-0.1in}
    \label{fig:ca_example}
\end{figure}

\partitle{Apportion cost via \func{CR} Operations based on Pregel}
Building upon \func{CR}, we introduce \func{RollUpApp} that repeatedly applies the \func{CR} to all non-leaf vertices to solve the MCDP problem. We adopt a Pregel-based implementation \cite{malewicz2010pregel} since \func{CR} is \textit{vertex centric} and naturally fits distributed graph processing frameworks. \func{RollUpApp} is shown in Alg. \ref{alg:simpleapp-pregel}. It takes as input the graph $G=(V,A,w)$, the stateful cost variable $C(v)$. The output is the final cost state $C^{out}(v)$ for all vertices $v \in V$. The algorithm proceeds in supersteps.

\begin{algorithm}[b]\small
    \SetKwBlock{CAOp}{ }{end}
    \SetKwBlock{compute}{\func{compute}(msgs)}{end}
    \caption{\func{RollUpApp} (Pregel-style pseudocode)}
    \label{alg:simpleapp-pregel}
    \KwIn{Graph $G=(V, A, w)$, direct costs $\mathcal{C}_{dc}(v)$ for each $v \in V$}
    \KwOut{$C^{out}(v)$}

    Each vertex $v \in V$ sets $C^{out}(v) \gets \mathcal{C}_{dc}(v)$ \tcp*{\darkBlue{Initialization}}

    \While{any non-leaf vertex $v$ is active}{
    
        \textbf{foreach} active vertex $v$ \textbf{do} in parallel \CAOp {
            \tcc{\darkBlue{\func{CR} operation in Pregel style}}
            \compute {
            % \Indp

                update $C^{out}(v) \gets C^{out}(v) + \sum_{m_i \in msgs} m_i$\;
                
                % \If {$v \in V^{op}$} {
                %     sends $w(v, v') \cdot C_1(v)$ to each $v' \in D_G^1(v)$\;
                % }
                  
                sends $w(v, v') \cdot C^{out}(v)$ to each $v' \in D_G^1(v)$\;
                
                updates $C^{out}(v) \gets 0$ \;
            
                \func{voteToHalt}()\;

                \tcp{\darkBlue{Reactivated upon receiving a message}}
            }
        
        }

    }
    \Return $C^{out}(v)$
\end{algorithm}

Initially, each vertex $v$ initiates a new state $C^{out}(v)$ (line 1). In each superstep, every active non-leaf vertex performs \func{CR} operation (line 3). Each vertex collects all incoming costs from the messages received from upstream neighbors and accumulates these into its cost state $C^{out}(v)$ (line 5). \func{RollUpApp} then pushes its current cost to its immediate downstream neighbors according to the edge weights (line 6), and resets its cost to zero (line 7), indicating that its cost is ``flowed-away''. The process iterates until all non-leaf vertices become inactive.
%  or the algorithm reaches the maximum superstep threshold $t$ (line 2)
Finally, the algorithm returns the final cost state $C^{out}(v)$ for all $v \in V$ (line 10). The final cost state of DPs $\mathbf{c}^L$ can be extracted from $C^{out}(v)$ as $\mathbf{c}^L_i \gets C^{out}(p_i)$ for $p \in L(G)$.

% 实验部分关于P的大小要进行讨论
\func{RollUpApp} is scalable to large graphs due to its Pregel-based distributed computing paradigm \cite{malewicz2010pregel}. For DAGs, \func{RollUpApp} is guaranteed to terminate and yield exact results.

\begin{theorem}
\label{thm:simple_ca_cycle}
Given a data pipeline DAG $G=(V,A,w)$ that satisfies the cost conservation property and $C(v)$ initialized as $\mathcal{C}_{dc}(v)$, \func{RollUpApp(G,C)} terminates and exactly computes the MC vector $\mathbf{c}^L$ where $\mathbf{c}^L_i = \mathcal{C}_{mc}(p_i)$ for all leaves $p_i \in L(G)$ (thereby solving the MCDP problem) in at most $P$ supersteps, where $P$ is the length of the longest path in $G$.
\end{theorem} 
\emph{Proof intuition.} Every \func{CR} operation conserves cost because its outgoing weights sum to one. In a DAG, cost can advance at least one edge per superstep and can never return to a processed vertex. Hence, after at most the longest-path length $P$, all cost has reached the leaves. Cost conservation then ensures that each leaf holds exactly its recursively defined MC. The full proof is in Appendix~\ref{sec:proof_simple_ca_cycle}.

% At the same time, when $G$ is DAG and $C(v)$ is initialized as $\mathcal{C}_{dc}(v)$, it outputs exact MC results for all DPs, as stated below.
% \begin{theorem}
% \label{thm:simple_ca_cycle}
% For any DAG $G=(V,A,w)$ and $C(v)$ initialized as $\mathcal{C}_{dc}(v)$, \func{RollUpApp(G,C,V)} terminates and correctly computes the MC vector $\mathbf{c}^L$ where $\mathbf{c}^L_i = \mathcal{C}_{mc}(p_i)$ for all leaves $p_i \in L(G)$ (solve the MCDP problem) in at most $P$ supersteps, where $P$ is the length of the longest path in $G$.
% \end{theorem} 
% The proof is demonstrated in Appendix \ref{sec:proof_simple_ca_cycle}.

\partitle{Limitations} 
In graphs with cycles (DCGs), costs circulate indefinitely, preventing the algorithm from terminating naturally. While enforcing a superstep threshold ensures termination, it inevitably leads to approximate results, failing the exactness goal.

% \func{RollUpApp} fails to satisfy exactness and termination for a DCG because costs circulate indefinitely within cycles. To ensure termination, we can set superstep threshold $t$ to limit the number of supersteps. This variant is denoted as \func{RollUpApp-T} and detailed in Appendix \ref{appendix:simpleapp_threshold}. However, the exactness goal remains unresolved.

\subsubsection{The \func{MatrixApp} Algorithm}
\label{sec:matrix_app}

\begin{algorithm}[b]\small
    \caption{\func{MatrixApp}}
    \label{alg:matrix_app}
    \KwIn{Cost allocation matrix $\mathbf{W}$, DC vector $\mathbf{d}$}
    \KwOut{MC vector of DPs $\mathbf{c}^{L}$}
    Compute $\mathbf{c} \gets (\mathbf{I} - \mathbf{W}^\mathsf{T})^{-1} \mathbf{d}$ \;
    Extract the sub-vector $\mathbf{c}^{L}$ from $\mathbf{c}$ \;
    \Return $\mathbf{c}^{L}$
\end{algorithm}

To address the cyclic dependencies where \func{RollUpApp} fails, we formulate the cost flow as a system of linear equations.

% As noted in \S\ref{sec:preliminary}, data pipelines contain cycles (\eg, $v_4-v_5-v_6$ in Figure \ref{fig:data_pipeline}) and forms a DCG. In such cases, the manufacturing costs of vertices within a cycle are mutually dependent, forming a system of linear equations. 
\partitle{Linear system formulation}
As defined in Eq. \ref{equ:ma}, the manufacturing costs across a cyclic graph are mutually dependent. For any vertex $v$, its cost consists of its direct cost and the apportioned costs from its inputs:
\begin{equation}
\label{equ:linear_system_node}
    \mathcal{C}_{mc}(v) - \sum_{u \in U_G^1(v)} w(u,v) \cdot \mathcal{C}_{mc}(u) = \mathcal{C}_{dc}(v)
\end{equation}
Considering all vertices in $V(G)$ indexed from $1$ to $n$, this forms a global system of $|V(G)|$ linear equations. Let $\mathbf{c}$ be the vector of manufacturing costs $[\mathcal{C}_{mc}(v_1), \dots, \mathcal{C}_{mc}(v_n)]^T$, $\mathbf{d}$ be the vector of direct costs $[\mathcal{C}_{dc}(v_1), \dots, \mathcal{C}_{dc}(v_n)]^T$,
$\mathbf{W}$ is a \textit{cost allocation matrix} where $\mathbf{W}[i,j] = w(v_i,v_j)$ if there is an arc from $v_i$ to $v_j$ and 0 otherwise, and $\mathbf{I}$ is an identity matrix. The system can be expressed in matrix form (see Appendix \ref{sec:matrix_form_derivation} for detailed derivation) as:
% \begin{equation}
% \label{equ:matrix_form}
%     (\mathbf{I} - \mathbf{W}^\mathsf{T})\mathbf{c} = \mathbf{d}
% \end{equation}

\begin{equation}
\label{equ:matrix_form}
    \mathbf{c} = \mathbf{d} + \mathbf{W}^\mathsf{T} \mathbf{c}
\end{equation}

% closed-form : 用有限步骤的代数运算得到的解析解。
% where $\mathbf{W}$ is a \textit{cost allocation matrix} where $W_{ij} = w(v_i,v_j)$ if there is an arc from $v_i$ to $v_j$ and 0 otherwise (detailed derivation is provided in Appendix \ref{sec:matrix_form_derivation}). 

Solving this system for $\mathbf{c}$ yields the exact and unique solution to the MCDP problem, formalized as \func{MatrixApp} in Alg.~\ref{alg:matrix_app}.
% We remark that one must first compute the complete cost vector $\mathbf{c}$ for all nodes (line 1) before extracting the MC of leaf costs (line 2), as the cost of any leaf node depends on the costs of its non-leaf ancestors.

% due to the weight structure不是很清楚，需要用更精确的一句话说明
% The invertibility of $(\mathbf{I} - \mathbf{W}^\mathsf{T})$ is ensured due to the weight structure of $G$ (see the proof in Appendix \ref{appendix:proof_matrix_inver}), ensuring that a unique solution always exists. 

\begin{theorem}[Exactness of \func{MatrixApp}]
\label{thm:matrixapp_exact}
Given a data pipeline graph $G=(V,A,w)$ that satisfies the cost conservation property, \func{MatrixApp} computes the exact and unique solution to the MCDP problem.
\end{theorem}
\emph{Proof intuition.} Eq.~\ref{equ:matrix_form} writes all recursive cost equations simultaneously. The weights are nonnegative, and every vertex has a positive-weight path to a leaf. These properties make the spectral radius of $\mathbf{W}$ smaller than one, so $\mathbf{I}-\mathbf{W}^{\mathsf T}$ is invertible. Its unique solution is precisely the MC vector defined by Eq.~\ref{equ:ma}. The full proof is in Appendix~\ref{appendix:proof_matrixapp_exact}.

% 关于可逆性的说明
% The invertibility of $(\mathbf{I} - \mathbf{W}^\mathsf{T})$ is guaranteed by the spectral structure of $\mathbf{W}$: all entries are non-negative, each row sum is at most~1, and at least one row sum is strictly less than~1 (the leaves). A formal proof is given in Appendix~\ref{appendix:proof_matrix_inver}.

% all entries are non-negative, each row sum is at most~1, and at least one row sum is strictly less than~1 (the leaves). These properties imply that the spectral radius satisfies $\rho(\mathbf{W}) < 1$, and hence $(\mathbf{I} - \mathbf{W}^\mathsf{T})$ is nonsingular. A formal proof is given in Appendix~\ref{appendix:proof_matrix_inver}.

\partitle{Limitations}
While \func{MatrixApp} guarantees exactness, it is computationally prohibitive for large-scale pipelines. Matrix inversion generally incurs cubic time complexity $O(|V|^3)$ and quadratic space complexity $O(|V|^2)$. Even with optimizations for sparse matrices, this approach becomes computationally and memory-prohibitive for pipelines with millions of entities, failing the scalability goal.

% 原版，同时描述了两个threshold版本的simpleApp
% However, \func{RollUpApp} fails in satisfy exactness for a DCG because costs circulates indefinitely within cycles, which also leads to non-termination of execution. To ensure termination, one can set stop conditions to achieve approximate accuracy. For example, we can set a ratio threshold $\epsilon_r$ to stop the algorithm when the total cost remaining on all non-leaf vertices is above $\epsilon_r$, or set a step threshold $\epsilon_s$ to limit the number of supersteps. We denote these two variants as \func{Simple-RT} and \func{Simple-ST}, respectively, and details them in Appendix \ref{appendix:alternative_simpleapp}. These variants form a trade-off between computational efficiency and accuracy of result. Even so, the exactness goal of these variants of \func{RollUpApp} remains unresolved.

Our proposed \func{RollUpApp} and \func{MatrixApp} for the MCDP problem highlight a fundamental trade-off between exactness and scalability, motivating the need for a hybrid solution.

% \begin{figure}[t]
%     \centering
%     \begin{tikzpicture}[x=1cm,y=1cm,>=Latex,every node/.style={font=\small}]
%         \node[draw,rounded corners,align=center,text=brown] (rollup) at (0,0) {\func{RollUpApp}\\scalable and exact on DAGs};
%         \node[draw,rounded corners,align=center,text=brown] (matrix) at (4.6,0) {\func{MatrixApp}\\exact but not scalable};
%         \node[draw,rounded corners,thick,align=center,text=brown] (costapp) at (2.3,1.65) {\func{CostApp}\\exact and scalable};
%         \draw[->,brown,thick] (rollup) -- (costapp);
%         \draw[->,brown,thick] (matrix) -- (costapp);
%     \end{tikzpicture}
%     \vspace{-0.08in}
%     \caption{Relationship among the two basic methods and \func{CostApp}.}
%     \label{fig:relation}
%     \vspace{-0.15in}
% \end{figure}

% 存在以下几个问题：
% 4.2和4.3应该放在一起，因为都是进行“成本汇聚”
% 4.2.1 首先说为了”进行成本汇聚“，需要先定位并破坏环，但是没有说清晰表达“为什么要进行成本汇聚”。同时，也没有说明为什么一定要找FAS，直接介绍了FAS的定义和算法。另外，对于说明FindFAS算法正确性的定理，应该用一句话来直观解释。
% 4.2.2 在介绍过程前，没有说明graphEdit的目的：是为了让所有的环路成本都能被新创建的source节点吸收掉，从而实现成本汇聚。
% 4.3中，没有清晰地陈述这一步的结果对于“主线”的影响：新的成本状态使得所有成本都集中在新创建的source节点和DP节点上，从而大大简化了成本分配关系。
% 4.4中，要清晰地首先表达两点1）线性系统得到降维 2）剩下的挑战是寻找矩阵W。 寻找矩阵算法，要优先描述向量化版本的逻辑。 此外，在求解时，要直接介绍优化版本，而不是只是将它作为一个可选项。
% 汇聚cost时，对于simpleApp的第三个参数，应该说明排除这些节点可以使得它们充当汇聚节点的角色，只累计流入的成本而不将其传出。

\section{The Improved Algorithm: CostApp}
\label{sec:the_improved_algorithm}

We introduce \func{CostApp}, a novel hybrid algorithm that strategically integrates the strengths of both basic methods to deliver an exact and highly scalable solution for the MCDP problem on cyclic data pipelines. We first overview our approach (\S\ref{sec:overview}), followed by a step-by-step detailing of its execution phases (\S\ref{sec:phase1}, \S\ref{sec:phase2}), and conclude with analyses of its correctness and efficiency (\S\ref{sec:analysis}).

\subsection{Overview of \func{CostApp}}
\label{sec:overview}

\begin{figure}[t]
    \centering
    \includegraphics[width=0.6\linewidth]{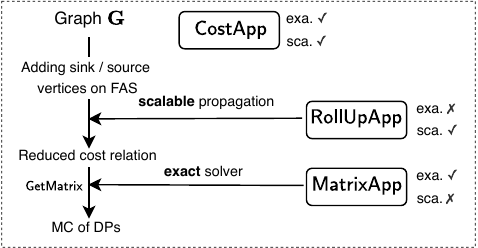}
    \vspace{-0.08in}
    \caption{The key design of \func{CostApp}.}
    \label{fig:relation}
    \vspace{-0.15in}
\end{figure}

The high-level idea of \func{CostApp} is decoupling and reduction: it transforms an intractable, large-scale problem into an equivalent but easily solvable, compact one. The algorithm \func{CostApp} is illustrated in 
Alg. \ref{alg:costapp}. \func{CostApp} operates in two main phases:

\partitle{Phase 1: Acyclic cost aggregation} 
This phase uses scalable, distributed graph algorithms to propagate all direct costs throughout the acyclic portions of the pipelines and consolidate them onto a small, well-defined set of terminal vertices. The process begins by identifying the cyclic dependencies using our scalable \func{FindFAS} algorithm (line 1) to obtain an FAS. Instead of simply removing these arcs, we isolate the cycles using \func{GraphEdit} (line 2), which replaces each feedback arc with a special path through auxiliary \textit{sink} and \textit{source} vertices. They effectively isolate the cyclic dependency and render the majority of the graph acyclic, enabling us to use the highly scalable \func{RollUpApp} algorithm to efficiently propagate all direct costs (line 5). Upon completion, all costs are precisely aggregated onto a small set of terminal vertices \ie the original data products and the newly created sink nodes. This sets the stage for an exact solution.

\partitle{Phase 2: Cyclic cost resolution} 
With all costs consolidated on the terminal vertices, this phase constructs and solves a compact linear system to resolve the remaining cyclic dependencies. First, we extract the aggregated costs from Phase 1 to form a new, reduced direct cost vector $\mathbf{d}_{\text{term}}$ (line 6). We then compute the reduced cost allocation matrix, $\mathbf{W}_{\text{term}}$, using \func{GetMatrix} (line 7), which captures the cost-flow relationships \textit{only among} these terminal vertices. This transforms the original large-scale problem into a compact and equivalent linear system, which we solve exactly and efficiently using \func{MatrixApp$^*$} (a variant of \func{MatrixApp}, see \S\ref{sec:solving_reduced_system}) (line 8). This final step yields the exact manufacturing costs results for all DPs.

\begin{algorithm}[b]\small
\caption{\func{CostApp}}
\label{alg:costapp}
\KwIn{Graph $G=(V,A,w)$, direct costs $\mathcal{C}_{dc}(v)$ for all $v \in V$}
\KwOut{The MC vector of DPs $\mathbf{c}^{L}$}

\tcc{\darkBlue{Phase 1: Acyclic Cost Aggregation}}
$A_{\text{FAS}} \gets \func{FindFAS}(G)$ \label{line:findfas}\tcp*{\darkBlue{(1.a) Identifing cycles}}
$(G_1, C_1,V_{sk}, V_{sr}) \gets \func{GraphEdit}(G, \mathcal{C}_{dc}, A_{\text{FAS}})$ \label{line:graphedit}\tcp*{\darkBlue{(1.b) Isolating cycles}}

Let $G_{DAG}$ be the subgraph of $G_1$ without sink-to-source arcs\; 
Let $V_{\text{term}}$ be the set of terminal vertices, comprising the leaves of $G_{DAG}$, \ie $L(G) \cup V_{sk}$\;
$C_{DAG} \gets \func{RollUpApp}(G_{DAG}, C_1)$ \label{line:aggregate}\tcp*{\darkBlue{(1.c) Aggregating costs}}

\tcc{\darkBlue{Phase 2: Cyclic Cost Resolution}}
$\mathbf{d}_{\text{term}} \gets \func{ExtractCosts}(C_{DAG},V_{\text{term}})$ \label{line:extract_o}\;
$\mathbf{W}_{\text{term}} \gets \func{GetMatrix}(G_{DAG}, V_{sk}, L(G)) $ \label{line:getmatrix}\tcp*{\darkBlue{(2.a) Getting the cost allocation matrix}}
$\mathbf{c}^{L} \gets \func{MatrixApp$^*$}(\mathbf{W}_{\text{term}}, \mathbf{d}_{\text{term}})$ \label{line:solve_reduced}\tcp*{\darkBlue{(2.b) Solving a compact system for final costs}}

\Return{$\mathbf{c}^{L}$}
\end{algorithm}

\partitle{Running example}
Figure \ref{fig:costapp_workflow_example} illustrates the execution of \func{CostApp} step by step. For simplicity, we initialize the cost of each vertex to $1$. In Phase 1, the ordering labels (orange) generated by \func{FindFAS} identify the feedback arc (blue) in step \darkBlue{(1.a)}. Next, auxiliary sink and source vertices (blue) are strategically inserted to isolate the cyclic dependency \darkBlue{(1.b)}. All costs are aggregated onto the terminal vertices (red, including all sink vertices and leaf vertices) by invoking \func{RollUpApp} \darkBlue{(1.c)}, zeroing out the remaining vertices (green). Moving to Phase 2, unit costs are injected at the source vertices to construct the reduced cost allocation matrix $\mathbf{W}_{\text{term}}$ \darkBlue{(2.a)}. Finally, \func{MatrixApp$^*$} is invoked to solve this compact linear system to yield the exact MCs \darkBlue{(2.b)}.

\begin{figure*}[t]
    \centering
    \includegraphics[width=1\linewidth]{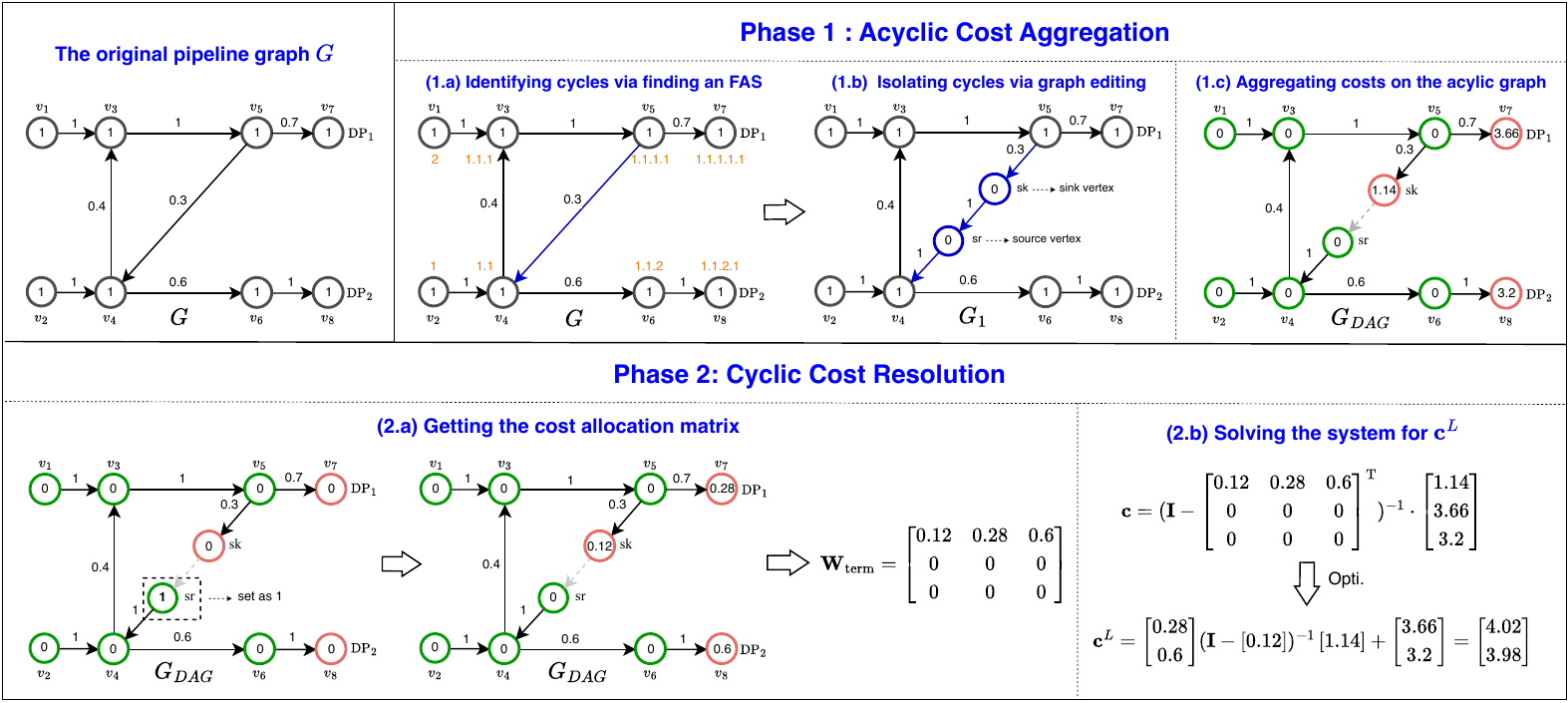}
    \vspace{-0.2in}
    \caption{An illustrative workflow of \func{CostApp}.}
    \vspace{-0.1in}
    \label{fig:costapp_workflow_example}
\end{figure*}

\partitle{Key design} The core idea of \func{CostApp} is the injection of auxiliary sink and source vertices on the feedback arcs to 
overcome both the \textit{cyclic} and \textit{massive scale} challenges (\S\ref{sec:intro}). In Phase~1, sink vertices act as \textit{collectors} to intercept all costs that would flow over feedback arcs. This simultaneously \textit{isolates cycles} and \textit{simplifies cost dependencies}, paving the way for formulating a compact equivalent linear system. To resolve this compact system in Phase 2, the source vertices serve as \textit{re-injection points}. They capture how the intercepted costs at the sink vertices re-enter the original graph, allowing us to obtain a compact cost allocation matrix and efficiently solve this system.

% . To overcome the \textit{cyclic} and \textit{massive scale} challenges (stated in \S\ref{sec:intro}), \func{CostApp} introduces sink vertices that act as \textit{collectors} to intercept all costs that would flow over feedback arcs. This enables us to simultaneously \textit{isolate cycles} and \textit{simplify the cost flow dependencies} in Phase 1. This sets up the stage for formulating a compact equivalent linear system. To resolve this compact system in Phase 2, the source vertices serve as \textit{re-injection points}. They capture how the intercepted costs at the sink vertices re-enter the original graph, based on which we can obtain a compact cost allocation matrix and efficiently solve the reduced system.

Figure~\ref{fig:relation} summarizes the key design of \func{CostApp}.
By strategically injecting auxiliary vertices and integrating \func{RollUpApp} and \func{MatrixApp}, \func{CostApp} shifts the computational bottleneck from a massive $|V| \times |V|$ matrix inversion to a tiny one (often less than 0.05\% of $|V|$ in size, see \S\ref{sec:exp_realworld}). Our design is simple yet effective, achieving a dramatic performance gain while preserving the exactness of results (see Theorem~\ref{thm:costapp_correctness}). We now detail each phase.

\subsection{Phase 1: Acyclic Cost Aggregation}
\label{sec:phase1}

The goal of phase 1 is to aggregate all direct costs $\mathcal{C}_{dc}(v)$ from millions of internal vertices onto a small set of terminal vertices, while preserving the manufacturing cost of every DP. To achieve this, directly using \func{RollUpApp} on the original graph to push cost would fail, as costs would circulate indefinitely in cycles. Our solution, therefore, involves three steps: (1) identifying cycles via finding an FAS, (2) isolating cycles via graph editing, and (3) aggregating costs on the resulting acyclic graph.

% The key innovation lies in the first two steps: first, transforming the graph to break cycles, and second, using \func{RollUpApp} to efficiently propagate costs on the modified graph.

% To prevent the infinite circulation of costs, we must first identify and isolate the cycles.

\subsubsection{\darkBlue{(1.a)} Identifying Cycles via Finding an FAS}

% \subsection{Identifing the Feedback Arc Set (FAS)}
\label{sec:find_fas}

To make cost aggregation possible, we must first break all cycles while preserving the majority of the original graph structure. A standard way in graph theory is to find a FAS, \ie a subset of arcs whose removal makes the graph acyclic. While finding a minimum FAS is an NP-hard problem~\cite{karp1972reducibility,crescenzi1995compendium,even1995joseph}, for our purpose \emph{any} valid FAS suffices, as long as it can be computed efficiently on large graphs.

We therefore design a scalable Pregel-style algorithm \func{FindFAS} with the following high-level strategy. First, we assign each vertex a hierarchical label. Given these labels, any arc $\langle u,v\rangle$ for which the destination label is a prefix of the source label is classified as a feedback arc. The set of all such arcs forms a valid FAS (Lemma~\ref{lem:FAS}). Arcs that do not satisfy this prefix condition need not be included in the FAS.

% Our strategy is to first establish a linear ordering of all vertices in the graph $G$. Given such an ordering, any arc $\langle u, v \rangle$ where $u$ appears after $v$ is considered a \textit{backward} or inconsistent arc. The set of all such inconsistent arcs constitutes a valid FAS. 

\partitle{Distributed labeling for vertex ordering}
To establish the ordering, we use a synchronous labeling algorithm \func{Label}. Each vertex $v$ maintains a label $Label(v)$, represented as a dot-separated sequence of positive integers (e.g., `1.2.1'). We use the conventional lexicographic order $<_{lo}$ defined below.

\begin{definition}[Lexicographical Label Ordering]
\label{def:lexicographic_order}
Given two distinct labels $Label_A = a_1.a_2.\dots.a_p$ and $Label_B = b_1.b_2.\dots.b_q$, we say $Label_A$ is lexicographically smaller than $Label_B$, denoted $Label_A <_{lo} Label_B$, if either:
% A>B： (1).A是B的前缀，如A=1.2, B=1.2.1 (2).存在一个位置k，使得在k之前的所有位置相等，而在k位置上A的值小于B的值，如A=1.2.1, B=1.2.5
% 对于1.3.1和1.2， 1.2 > 1.3.1
% 对于1.2.1和1.3，1.2.1 > 1.3 (1.2.1>1.2, 1.2>1.2.1)'
% 即：逐位比较大小，以第一个不相等的位置为比较结果
\begin{enumerate}
    \item $Label_A$ is a prefix of $Label_B$ (\ie $p < q$ and $a_i = b_i$ for all $1 \le i \le p$).
    \item There exists an index $k \le \min(p, q)$ such that $a_i = b_i$ for all $i < k$ and $a_k < b_k$.
\end{enumerate}

\end{definition}

The Pregel-based \func{Label} algorithm assigns distinct integer labels to roots and initializes every other vertex with $\top$, a value larger than any valid label. In each superstep, every active vertex $v$:

% 传播时越来越大，所以正常都是小指大。当遇到大指小时（且小的是大的前缀），就是feedback。
% 每个节点会选择收到最小的Label。

\begin{enumerate}
    \item collects candidate labels from its upstream neighbors,
    \item adopts the lexicographically \textit{smallest} candidate if it is smaller than $Label(v)$, and
    \item propagates an extended version of its new label, obtained by appending a local outgoing-edge suffix, to its downstream neighbors.
\end{enumerate}

This process repeats until no label changes. In Figure~\ref{fig:costapp_workflow_example}(1.a), a root with label `1' proposes `1.1' along its first outgoing branch. The downstream vertex keeps that label if it is smaller than all other incoming label, and then extends it for the next hop. Hence labels grow from an ancestor to its descendants, while every vertex consistently selects one smallest incoming label. The full pseudocode is in Appendix~\ref{appendix:label}.

% The \func{Label} algorithm is guaranteed to terminate. Let $L_p$ be the length of the longest simple path in any strongly connected component of $G$. The algorithm converges in a number of supersteps proportional to $L_p$.

\begin{lemma}
\label{thm:label}
\func{Label} terminates and assigns a stable label to every vertex in at most $L_p+1$ supersteps, where $L_p$ is the length of the longest simple path in $G$.
\end{lemma}

The proof is similar to the proof of Theorem \ref{thm:simple_ca_cycle} and we omit it here for brevity.

% \begin{theorem}
% \label{thm:label}
% \func{Label} terminates and assigns a final, stable label to every vertex in $G$ within at most $L_p+1$ supersteps, where $L_p$ is the length of the longest simple path in $G$.
% \end{theorem}

\partitle{Finding an FAS}
Once \func{Label} terminates, every arc $\langle u,v\rangle$ falls into exactly one of the following types:

\begin{enumerate}
    \item direct forward arc: $Label(u)$ is a prefix of $Label(v)$ (\eg $Label(v_4)=1.1$, $Label(v_6)=1.1.2$ in Figure~\ref{fig:costapp_workflow_example}).
    \item cross forward arc: $Label(u) <_{lo} Label(v)$ but $Label(u)$ is not a prefix of $Label(v)$.
    \item feedback arc: $Label(v)$ is a prefix of $Label(u)$ (\eg $Label(v_5)=1.1.1.1$, $Label(v_4)=1.1$ in Figure~\ref{fig:costapp_workflow_example}).
    \item cross backward arc: $Label(v) <_{lo} Label(u)$ but $Label(v)$ is not a prefix of $Label(u)$ (\eg $Label(v_1)=2$, $Label(v_3)=1.1.1$ in Figure~\ref{fig:costapp_workflow_example}).
\end{enumerate}

% \begin{lemma}
% \label{lem:types of arcs}
% For any arc $\langle u, v \rangle \in A(G)$, the final labels $L(u)$ and $L(v)$ satisfy exactly one of the following conditions:
% \begin{enumerate}
%     \item Direct Forward Arc: $L(v)$ is a direct extension of $L(u)$ (\eg $L(u)=1.2$, $L(v)=1.2.3$).
%     \item Cross Forward Arc: $L(v) >_{lo} L(u)$, but $L(u)$ is not a prefix of $L(v)$ (\eg $L(u)=1.2$, $L(v)=1.3$).
%     \item Feedback Arc: $L(u) >_{lo} L(v)$. (\eg $L(u)=1.3$, $L(v)=1.2$).
% \end{enumerate}
% \end{lemma}

Only type~(3) arcs constitute $A_{\text{FAS}}$. A type~(4) arc may lie on a cyclic path, but its destination label is not a prefix of its source label, so it need not be selected.
\begin{lemma}
\label{lem:FAS}
The set

\[
A_{\text{FAS}}=\{\langle u,v\rangle\in A(G)\mid Label(v)\text{ is a prefix of }Label(u)\}
\]
is a valid feedback arc set for $G$. In other words, $G'=(V,A\setminus A_{\text{FAS}})$ is a DAG.
\end{lemma}
\emph{Proof sketch.} Consider any directed cycle and let $x$ be a vertex on it with the smallest label. Because every vertex keeps the smallest label proposed by its predecessors, traversing the cycle from $x$ forces $Label(x)$ to remain a prefix of the labels on that traversal. If $y$ is the predecessor of $x$ on the cycle, then $Label(x)$ is a prefix of $Label(y)$, so the closing arc $\langle y,x\rangle$ belongs to $A_{\text{FAS}}$. Therefore, removing $A_{\text{FAS}}$ removes at least one arc from every directed cycle, so no directed cycle remains. Appendix~\ref{appendix:proof_fas} gives the full proof.

Therefore, \func{FindFAS} first runs \func{Label} and then, in one superstep, marks each incoming arc $\langle u,v\rangle$ for which $Label(v)$ is a prefix of $Label(u)$. Its superstep complexity is $O(L_p)$, dominated by \func{Label}.

% (Theorem~\ref{thm:Alg-FAS}).
% \begin{theorem}
% \label{thm:Alg-FAS} 
% The \func{FindFAS} algorithm correctly computes a feedback arc set of $G$ with a total complexity dominated by the \func{Label} algorithm, requiring $O(L_p)$ supersteps.
% \end{theorem} 

Figure~\ref{fig:costapp_workflow_example}(1.a) illustrates the result. Each vertex receives a hierarchical label (orange), and $\langle v_5,v_4\rangle$ is selected (blue) because $Label(v_4)=1.1$ is a prefix of $Label(v_5)=1.1.1.1$.

% Our approach to identifying a valid (though not necessarily minimum) FAS is to first establish a linear ordering over all vertices in the graph. Given such an ordering, any arc that travels \textit{backward} from a higher-ordered vertex to a lower-ordered one is inconsistent with a directed acyclic structure. The collection of all such backward arcs constitutes a valid FAS, as their removal eliminates all cycles. We have developed a scalable, distributed algorithm, \func{FindFAS}, to compute this ordering and identify the FAS. The detailed implementation is deferred to \S\ref{sec:find_fas} to maintain focus on the main workflow.

% For instance, consider the cyclic dependency graph in Figure~\ref{fig:costapp_workflow_example} involving vertices $\{v_3,v_4, v_5, \}$. Our \func{FindFAS} algorithm might establish a linear ordering where $v_4 < v_3 < v_5$. Under this ordering, the arcs $(v_4, v_3)$ and $(v_3, v_5)$ are forward arcs. However, the arc $(v_5, v_4)$ is a backward or feedback arc because it points from a higher-ordered vertex ($v_5$) to a lower-ordered one ($v_4$). This arc is therefore identified as a member of the FAS, as highlighted in blue in Figure~\ref{fig:costapp_workflow_example}.

\subsubsection{\darkBlue{(1.b)} Isolating Cycles via Graph Editing}
Simply deleting the FAS arcs would break the original cost flows. To preserve them while removing cycles from Phase~1, we replace each FAS arc with a path through a \textit{sink} vertex and a \textit{source} vertex. The sink retains the cost during Phase~1, and the source marks where that cost resumes its downstream flow in Phase~2.

To this end, we introduce a graph transformation \func{GraphEdit}. For each feedback arc $\langle u,v\rangle \in A_{\text{FAS}}$, \func{GraphEdit} performs the following local rewrite on $G$:
\begin{enumerate}
    \item Remove the original arc $\langle u,v\rangle$.
    \item Create two new auxiliary vertices: a \emph{sink} vertex $v_{sk}$ and a \emph{source} vertex $v_{sr}$, and add them to $V_{sk}$ and $V_{sr}$, respectively.
    \item Insert three arcs forming the path $u \to v_{sk} \to v_{sr} \to v$ with $w(u, v_{sk}) = w(u,v), w(v_{sk}, v_{sr}) = 1, w(v_{sr}, v) = 1$.
    % $(u, v_{sk})$ with weight $w(u,v)$, and $(v_{sk}, v_{sr})$, $(v_{sr}, v)$ both with weight $1$.
\end{enumerate}

% The key insight here is that this transformation precisely isolates the original cyclic dependency into a single, controlled link, $v_{sk} \to v_{sr}$. The new sink vertex $v_{sk}$ that can capture all costs that would have entered the cycle at that point, and the new source vertex $v_{sr}$ can be exploited to computing the reduced allocation matrix in Phase 2 (will be illustrated in \S\ref{sec:phase2}).
Figure~\ref{fig:costapp_workflow_example}(1.b) shows this transformation. The feedback arc $\langle v_5,v_4\rangle$ is replaced by $v_5\to v_{sk}\to v_{sr}\to v_4$. Phase~1 accumulates the cost at $v_{sk}$, and Phase~2 uses $v_{sr}$ to compute how that cost reaches the terminal vertices. Appendix~\ref{appendix:graphedit} gives the formal illustration of \func{GraphEdit}.

Crucially, this transformation does not change the manufacturing cost of any original vertex. The path $u \to v$ in $G$ is replaced by an equivalent path $u \to v_{sk} \to v_{sr} \to v$ in $G_1$ that carries exactly the same cost (due to $w(u,v_{sk})\cdot w(v_{sk},v_{sr})\cdot w(v_{sr},v) = w(u,v)$), and all other edges remain untouched. As a result, from the perspective of the original vertices, the total cost transmitted along every path is preserved. This is formalized below.

\begin{theorem}[Equivalence of MC of DPs after Graph Editing]
\label{thm:graph_edit}
    Let $G_1$ and $C_1$ be the transformed graph and cost state produced by $\func{GraphEdit}(G, \mathcal{C}_{dc}, A_{\text{FAS}})$. The manufacturing cost $\mathcal{C}_{mc}(p)$ for any original data product $p \in L(G)$ is identical in both the original system defined on $(G, \mathcal{C}_{dc})$ and the transformed system defined on $(G_1, C_1)$.
\end{theorem}

The proof is provided in Appendix \ref{appendix:proof_graph_edit}.

\subsubsection{\darkBlue{(1.c)} Aggregating Costs on the Acyclic Graph}
% \subsubsection{Propagating Costs to Terminal Vertices}
\label{sec:aggregating_costs}

After graph editing, we obtain a modified graph $G_1$ with additional sink/source vertices and arcs. Let $A_{S2S}$ denote the set of sink-to-source arcs $\{\langle v_{sk},v_{sr}\rangle\}$. By construction, removing $A_{S2S}$ yields an acyclic subgraph $G_{DAG} = G_1 \setminus A_{S2S}$. The leaves of $G_{DAG}$, which we call \emph{terminal vertices}, are
\[
V_{\text{term}} = L(G_{DAG}) = L(G) \cup V_{sk},
\]
that is, all original DPs (leaves of $G$) and all newly created sink vertices.

The acyclic structure of $G_{DAG}$ allows us to aggregate costs onto the much smaller set $V_{\text{term}}$ without disrupting the MC of DPs. Specifically, we run \func{RollUpApp} on $G_{DAG}$ (Alg.~\ref{alg:costapp}, line 3), initializing each vertex $v$ with its post-\func{GraphEdit} cost state $C_1(v)$. Because \func{RollUpApp} correctly transfers all costs to the leaves of an acyclic graph (Theorem~\ref{thm:simple_ca_cycle}), it propagates all costs through $G_{DAG}$ and consolidates them onto $V_{\text{term}}$. The following theorem states that this aggregation preserves the MC of every DP.

\begin{theorem}[Equivalence of MC of DPs after Aggregating Costs]
    \label{thm:aggregation_cost_preservation}
    Let $C_{DAG}$ be the cost state output by \func{RollUpApp}($G_{DAG}, C_1$). The manufacturing cost $\mathcal{C}_{mc}(p)$ for any original data product $p \in L(G)$ is identical in both the system defined on $(G_1, C_1)$ and the system defined on $(G_1, C_{DAG})$.
\end{theorem}

\emph{Proof intuition.} The internal vertices of $G_{DAG}$ have no outgoing path except through a terminal vertex. Because \func{RollUpApp} transfers cost without loss or duplication, aggregating the internal costs at the terminal vertices does not change the final total cost received by any terminal vertex. Restoring the isolated sink-to-source dependencies therefore yields the same MC for every DP. The full proof is in Appendix~\ref{appendix:proof_aggregation_cost_preservation}.

% By Theorem~\ref{thm:simple_ca_cycle}, $\func{RollUpApp(G_{DAG},C)}$ terminates in at most $P$ supersteps, where $P$ is the length of the longest path in $G_{DAG}$, and correctly transfers all costs to the leaves of $G_{DAG}$, \ie to $V_{\text{term}}$. The MC of DPs is preserved because the $G_{DAG}$ is a subgraph of $G_1$.

% Let $C_1$ denote the resulting cost state. It satisfies a crucial property:
% \begin{itemize}
%     \item For every terminal vertex $t \in V_{\text{term}}$, $C_1(t) > 0$ in general;
%     \item For every non-terminal vertex $v \in V(G_1) \setminus V_{\text{term}}$, $C_1(v) = 0$.
% \end{itemize}

As illustrated in Figure~\ref{fig:costapp_workflow_example}(1.c), the entire initial direct cost of the pipeline is consolidated onto the much smaller terminal set $V_{\text{term}}= \{v_7, v_8, v_{sk}\}$ (highlighted in red). The costs of all other vertices (highlighted in green) are zero, reducing the scale of the remaining cost-flow dependencies.

Phase~1 has used a scalable graph algorithm to transform a large, distributed MCP instance into a compact representation over a small number of terminal vertices, thereby setting up Phase~2 to resolve the remaining cyclic dependencies via a tiny linear system.

% yuan ban
% Upon its completion, we obtain an intermediate cost state, which we denote $C_1$, with a crucial property: all costs are now concentrated on $V_{\text{term}}$.
% \begin{itemize}
%     \item For any terminal vertex $t \in V_{\text{term}}$, $C_1(t) > 0$.
%     \item For any other vertex $v \in V(G_1) \setminus V_{\text{term}}$, $C_1(v) = 0$.
% \end{itemize}

% As illustrated in Figure \ref{fig:costapp_workflow_example}(d), the entire initial operational cost of the pipeline is now precisely and losslessly consolidated onto the much smaller terminal set $\{v_7, v_8, v_{sk}\}$ (highlighted in red). Thus, this step effectively uses a scalable algorithm to reduce a complex, distributed MCDP problem into a compact representation, opening up possibility to exactly resolve the remaining, smaller cyclic dependencies in Phase 2.

\subsection{Phase 2: Cyclic Cost Resolution}
\label{sec:phase2}
Phase~2 uses this compact representation to determine how the aggregated costs at sink vertices $V_{sk}$ are redistributed among themselves and ultimately to the DPs $L(G)$. Specifically, $C_{DAG}(t)$ contains every direct cost that can reach terminal $t$ without crossing an isolated feedback dependency. These values are sufficient because all non-terminal states are zero, and every remaining cyclic flow must leave a sink $v_{sk}$, re-enter through its paired source $v_{sr}$, and reach another member of $V_{\text{term}}$. The terminal costs and sink-to-terminal transfer ratios therefore form a closed system. Solving it accounts for repeated traversals of the original cycles and recovers the final DP costs without explicitly including any internal vertex with \(C_{DAG}(v)=0\) in the reduced system.

% This redistribution is governed by the feedback relationships isolated in the $v_{sk} \to v_{sr}$ paths, which we now resolve as a compact linear system.

\subsubsection{Formulating the Reduced Linear System}
\label{sec:formulating_reduced_system}

Our insight is that we can treat the aggregated costs $C_{DAG}$ on $V_{\text{term}}$ as \emph{new} direct costs for this reduced system. We denote this reduced cost vector by $\mathbf{d}_{\text{term}}$, where each entry corresponds to $C_{DAG}(v)$ for a terminal vertex $v$ in $V_{\text{term}}$. Consequently, this system follows our original cost definition (Eq.~\ref{equ:ma}) for the final MC of $V_{\text{term}}$, denoted by $\mathbf{c}_{\text{term}}$:

% Specifically, after Phase~1, all costs are concentrated on $V_{\text{term}} = L(G) \cup V_{sk}$.

% We treat the aggregated costs $C_1$ on terminal vertices as \emph{new} direct costs for this reduced system. We denote this reduced cost vector by $\mathbf{d}_{\text{term}}$, where each entry corresponds to $C_1(v)$ for a terminal vertex $v$ in $V_{\text{term}}$.

% This reduced system, which involves only the terminal vertices $V_{\text{term}}$, follows our original cost definition for their final manufacturing costs, $\mathbf{c}_{\text{term}}$ (Eq.~\ref{equ:ma}):

% 是否需要一个基于矩阵/无限几何序列收敛的证明
\begin{equation}
    \mathbf{c}_{\text{term}} = \mathbf{d}_{\text{term}} + \mathbf{W}_{\text{term}}^T \mathbf{c}_{\text{term}}
\end{equation}

Here, $\mathbf{W}_{\text{term}}$ is a reduced cost allocation matrix that captures how cost flows \emph{among} these terminal vertices via the isolated cyclic dependencies, where an entry $\mathbf{W}_{\text{term}}[i,j]$ represents the fraction of cost that flows from terminal vertex $t_i$ to $t_j$.

% and satisfies cost conservation (each row-sum equals to $1$, Definition~\ref{def:cost_conservation}).

This equation has the same form as the system in \S\ref{sec:matrix_app}. Because $\mathbf{W}_{\text{term}}$ retains the required nonnegativity and row-sum properties, the reduced system has a unique solution (Appendix~\ref{appendix:proof_matrix_inver}). We can therefore solve it exactly with \func{MatrixApp}. With $\mathbf{d}_{\text{term}}$ already computed in Phase~1, the remaining task is to construct $\mathbf{W}_{\text{term}}$.
%  We describe how to do this efficiently by leveraging source vertices in the next step.

\subsubsection{\darkBlue{(2.a)} Getting the Cost Allocation Matrix} 
\label{sec:get_matrix}
% The remaining task in Phase~2 is to construct the reduced cost allocation matrix $\mathbf{W}_{\text{term}}$. 
% This matrix captures the cost-flow relationships \textit{among} terminal vertices, \ie an entry $\mathbf{W}_{\text{term}}[i,j]$ represents the fraction of cost that flows from terminal vertex $t_i$ to $t_j$.
To compute $\mathbf{W}_{\text{term}}$, our core idea is to inject unit costs $1.0$ into the source vertices and track their propagation through the acyclic graph $G_{DAG}$. By observing how these costs are eventually distributed among the terminal vertices, we can construct the cost allocation matrix.

% Recall that $V_{\text{term}} = L(G) \cup V_{sk}$, and every sink $v_{sk} \in V_{sk}$ is uniquely paired with a source vertex $v_{sr} \in V_{sr}$ via a single sink-to-source arc. A crucial observation is that for a sink vertex $t_i$, $\mathbf{W}_{\text{term}}[i,j]$ is equivalent to the cost received by terminal vertex $t_j$ when a \textit{unit cost} is injected at the corresponding source vertex $v_{sr}$ and allowed to traverse the acyclic graph $G_{DAG}$. This is because the only way for cost to flow from $t_i$ to other terminal vertices in $G_1$ is through the path $t_i \to v_{sr} \to \dots \to t_j$. This observation allows us to compute $\mathbf{W}_{\text{term}}$ efficiently using \func{RollUpApp} by injecting unit cost to each source vertex.

Recall that $V_{\text{term}} = L(G) \cup V_{sk}$, and every sink vertex $t_j \in V_{sk}$ is uniquely paired with a source vertex $v^{(j)}_{sr} \in V_{sr}$ via a single sink-to-source arc. A crucial observation is that the $(j, i)$-th entry of the matrix, $\mathbf{W}_{\text{term}}[j,i]$, precisely equals the cost fraction received by the terminal vertex $t_i$ when a \textit{unit cost} ($1.0$) is injected at the $j$-th source vertex $v^{(j)}_{sr}$ and propagated through $G_{DAG}$. This holds because any cost flowing from $t_j$ to other terminal vertices must traverse the path starting with $t_j \to v^{(j)}_{sr} \to \dots \to t_i$. Leveraging this property, we can efficiently compute $\mathbf{W}_{\text{term}}$ by injecting unit costs at all source vertices and executing \func{RollUpApp}.

Based on this idea, we design \func{GetMatrix} to obtain $\mathbf{W}_{\text{term}}$ based on Pregel, as shown in Alg.~\ref{alg:get_matrix}. To enable efficient computation within a single Pregel job, we adopt a \emph{map-based} cost propagation approach. Specifically, instead of a scalar cost, we assign a cost state map to each vertex (line 3) to record \{source identifier: cost\} pairs. This representation allows the unit costs from all source vertices to be propagated in parallel within a \textit{single} Pregel job rather than through one job per source, while the identifiers keep their contributions distinguishable when paths merge. We inject unit costs at all source vertices (lines 4-8) and propagate them through $G_{DAG}$ using \func{RollUpApp$^*$} (line 9). Due to the cost conservation property (Definition~\ref{def:cost_conservation}), the final cost map at any terminal vertex $t_i$ constitutes the $i$-th column of $\mathbf{W}_{\text{term}}$ (lines 10-16). Thus, all columns of $\mathbf{W}_{\text{term}}$ are obtained in a single Pregel job, setting the stage for solving the reduced system in the next step.

Map propagation applies each arc weight independently to every value, and values with the same key are summed when paths merge. For example, suppose two sources contribute $0.6$ and $0.2$ to an intermediate vertex $x$, giving the map $\{1{:}0.6,2{:}0.2\}$. If $x$ forwards half of its state to terminal $t$, then $t$ receives $\{1{:}0.3,2{:}0.1\}$. Thus, the keys keep the contributions from different sources separate throughout the propagation.

Note that this map-based approach relies on a variant of \func{RollUpApp} that aggregates cost maps by summing values of the same keys. We term this variant \func{RollUpApp$^*$} and detail it in Appendix \ref{appendix:simpleapp_star}. Noting that \func{GetMatrix} (line 7 in Alg.\ref{alg:costapp}) and \func{RollUpApp} (line 5 in Alg.~\ref{alg:costapp}) both operate with Pregel on $G_{DAG}$, they can be merged into a single Pregel workload to avoid duplicated overhead. We denote this merged workload as \func{S\&G} in the following parts.

% 优化前算法
% \begin{algorithm}\small
% \caption{\func{GetMatrix}}
% \label{alg:get_matrix}
% \KwIn{
%     Graph $G_{DAG} = (V, A, w)$,
%     set of sink vertices $V_{sk}$, 
%     data products $L(G)$.
% }
% \KwOut{Cost allocation matrix $\mathbf{W}_{\text{term}} \in \mathbb{R}^{k \times k}$}

% \tcc{\darkBlue{1. Setup and Initialization}}
% Let $V_{term} = L(G) \cup V_{sk}$, $k = |V_{term}|$, $m = |V_{sr}|$\;
% % \tcp*{\darkBlue{The indexed terminal vertices}}
% Initialize $\mathbf{W}_{\text{term}}$ as a $k \times k$ zero matrix\;
% % Let $V^{op} = V \setminus V_{term}$ \;
% $C(v) \gets \text{empty map } \emptyset$ for each $v$ \tcp*{\darkBlue{Initiates a cost state map}}
% \For{$i \gets 1 $ to $ m$}{
%     Let $v_{sr,i}$ be the $i$-th source vertex\;
%     $C(v_{sr,i}) \gets \{i : 1.0\}$ \tcp*{\darkBlue{Inject a unit cost}}
% }

% \tcc{\darkBlue{2. \func{RollUpApp$^*$} aggregates maps by summing values of the same keys}}
% $C^{out} \gets \func{RollUpApp$^*$}(G_{DAG}, C)$

% \tcc{\darkBlue{3. Fill the matrix with the results}}

% \For{$i \gets 1$ to $k$}{
%     Let $t_i$ be the $i$-th terminal vertex\;

%     \For{$j \gets 1$ to $m$}{
%         $\mathbf{W}_{\text{term}}[j,i] \gets C^{out}(t_i).\text{getValue}(j, 0)$ \tcp*{\darkBlue{Get value for source index $i$, default to 0 if missing}}
%     }
% }
% \Return{$\mathbf{W}_{\text{term}}$}
% \end{algorithm}

\begin{algorithm}\small
\caption{\func{GetMatrix}}
\label{alg:get_matrix}
\KwIn{
    Graph $G_{DAG} = (V, A, w)$,
    set of sink vertices $V_{sk}$, 
    data products $L(G)$,
    set of source vertices $V_{sr}$.
}
\KwOut{Cost allocation matrix $\mathbf{W}_{\text{term}} \in \mathbb{R}^{k \times k}$}

\tcc{\darkBlue{1. Setup and Initialization}}
Let $V_{term} = L(G) \cup V_{sk}$, $k = |V_{term}|$\;
Initialize $\mathbf{W}_{\text{term}}$ as a $k \times k$ zero matrix\;
$C(v) \gets \text{empty map } \emptyset$ for each $v \in V$ \tcp*{\darkBlue{Initiates a cost state map}}

\For{$j \gets 1 $ to $k$}{
    Let $t_j$ be the $j$-th terminal vertex in $V_{term}$\;
    \If{$t_j \in V_{sk}$}{
        Let $v^{(j)}_{sr}$ be the source vertex paired with sink $t_j$\;
        $C(v^{(j)}_{sr}) \gets \{j : 1.0\}$ \tcp*{\darkBlue{Inject unit cost}}
    }
}

\tcc{\darkBlue{2. \func{RollUpApp$^*$} aggregates maps by summing values of the same keys}}
$C^{out} \gets \func{RollUpApp$^*$}(G_{DAG}, C)$

\tcc{\darkBlue{3. Fill the matrix with the results}}

\For{$i \gets 1$ to $k$}{
    Let $t_i$ be the $i$-th terminal vertex\;
    \For{$j \gets 1$ to $k$}{
        % 这里的if是因为只有sink vertex才会注入cost，其他terminal vertex不会注入cost，所以如果t_j不是sink vertex，那么它就不可能把cost传递给t_i了，所以直接填0
        \If{$t_j \in V_{sk}$}{
            % 这里的getValue(j, 0)是因为不是所有节点j都会触达节点i。如果t_i没有收到来自t_j的cost，那么它在cost map里就没有j这个key了，所以默认返回0。
            $\mathbf{W}_{\text{term}}[j,i] \gets C^{out}(t_i).\text{getValue}(j, 0)$ \tcp*{\darkBlue{Get value injected from $t_j$, default to 0}}
        }
        \Else{
            $\mathbf{W}_{\text{term}}[j,i] \gets 0$ \tcp*{\darkBlue{Non-sink terminals do not inject costs}}
        }
    }
}
\Return{$\mathbf{W}_{\text{term}}$}
\end{algorithm}

\subsubsection{\darkBlue{(2.b)} Solving the Reduced Linear System}
\label{sec:solving_reduced_system}

With $\mathbf{W}_{\text{term}}$ and $\mathbf{d}_{\text{term}}$ in hand, we can solve the compact system defined in \S\ref{sec:formulating_reduced_system}. Rather than apply \func{MatrixApp} to the entire reduced matrix, we exploit its block structure to further reduce the computational cost.

Recall that $V_{\text{term}} = L(G) \cup V_{sk}$. We index the rows and columns of $\mathbf{W}_{\text{term}}$ by listing $V_{sk}$ first and the DPs second. Since data products are leaves and do not propagate costs, both blocks in the DP row group are zero. Consequently, the reduced allocation matrix $\mathbf{W}_{\text{term}}$ is \textit{block upper triangular}:

$$
\mathbf{W}_{\text{term}} = 
\begin{bmatrix}
    \mathbf{W}_{sk,sk} & \mathbf{W}_{sk,dp} \\
    \mathbf{0} & \mathbf{0}
\end{bmatrix}.
$$

Therefore, only the $|V_{sk}|\times|V_{sk}|$ sink block must be inverted, while the DP costs are obtained by matrix-vector multiplication:

\begin{equation}
\label{eq:optimized_dp_calc}
    \mathbf{c}_{dp} = \mathbf{d}_{dp} + \mathbf{W}_{sk,dp}^T (\mathbf{I}_{sk} - \mathbf{W}_{sk,sk}^T)^{-1} \mathbf{d}_{sk}
\end{equation}

where $\mathbf{d}_{dp}$ and $\mathbf{d}_{sk}$ are the subvectors of $\mathbf{d}_{\text{term}}$ corresponding to data products and sink vertices, respectively. We term this variant as \func{MatrixApp$^*$}. This optimization further reduces the dominant inversion cost to $O(|V_{sk}|^3)$, followed by a cheap matrix-vector multiplication. In production data pipelines, $|V_{sk}|$ is vastly smaller than $|V(G)|$ due to the sparsity of feedback relationships ($|V_{sk}| = |A_{\text{FAS}}| < 0.05\%|V|$, see \S\ref{sec:exp_realworld}).

\subsection{Analysis}
\label{sec:analysis}
We now focus on the correctness and complexity of \func{CostApp}.

\partitle{Correctness}
% Recall that Theorem~\ref{thm:matrixapp_exact} establishes that solving the global linear system (\func{MatrixApp}) yields the exact ground truth for the MCDP problem.
To formally prove the correctness of \func{CostApp}, we demonstrate that the dimensionality reduction achieved through Phase~1 and Phase~2 preserves this exactness. 
% The input $G$ is the billing-period snapshot defined in \S\ref{sec:preliminary}. Consolidating versions defines the direct cost of each logical entity as the sum of its version costs, so it preserves the total billed cost. It does not change the cost-conservation requirement on the outgoing weights. The following correctness result therefore applies directly to the snapshot.

% In other words, we prove that the final costs of data products computed by the reduced system in \func{CostApp} are mathematically identical to those computed by \func{MatrixApp} on the full original graph.

\begin{theorem}[Correctness of \func{CostApp}]
\label{thm:costapp_correctness}
For any data pipeline graph $G = (V, A, w)$ that satisfies the cost conservation property and $C(v)$ initialized as $\mathcal{C}_{dc}(v)$, \func{CostApp} calculates the manufacturing cost $\mathcal{C}_{mc}(p)$ for all data products $p \in L(G)$ exactly, matching the exact results of the global \func{MatrixApp} algorithm.
\end{theorem}

% 以下为正常简化版
\begin{proof}[Proof Sketch]
The proof proceeds by tracing the state transitions across the execution of \func{CostApp}. Let $\mathbf{c}^*$ be the exact cost vector solving the global linear system (Eq.~\ref{equ:matrix_form}), and $\mathbf{c}^*$ is the ground truth of MC of DPs (Theorem~\ref{thm:matrixapp_exact}). First, the transformations in \func{GraphEdit} strictly preserve the accumulated costs along any path directed to the leaves, so the costs of all original vertices are preserved in the transformed graph (Theorem~\ref{thm:graph_edit}). Second, \func{RollUpApp} aggregates all costs along the acyclic partition of the graph onto the terminal vertices. Due to Theorem~\ref{thm:aggregation_cost_preservation}, this aggregation preserves the exactness of MC of DPs. Finally, the remaining dependencies exclusively form a reduced linear system defined on the terminal vertices. Solving this reduced system via \func{MatrixApp$^*$} recovers the exact identical solution for the data products as in the global system solved by \func{MatrixApp}. The detailed proof is given in Appendix~\ref{app:proof_costapp_correctness}.
\end{proof}

\partitle{Independence from the choice of $w$}
The guarantees above hold for any fixed, nonnegative weights satisfying Definition~\ref{def:cost_conservation}. \func{FindFAS} depends only on graph topology. The supersteps and messages of \func{RollUpApp} and \func{GetMatrix} depend on paths and arcs, and the dimension of \func{MatrixApp$^*$} depends on $|A_{\mathrm{FAS}}|$. Thus, neither exactness nor asymptotic complexity (see below) depends on the numerical values of $w$. 

\partitle{Complexity}
% todo is rooted in its hybrid design优化下
\func{CostApp} combines distributed graph processing with a compact linear solve. Table~\ref{tab:efficiency_comparison} summarizes the complexity of each component. The total complexity of \func{CostApp} is the sum of its distributed and centralized components: $O(L_p \cdot (|V| + |A|)) + O(|A_{\text{FAS}}|^3 + |A_{\text{FAS}}| \cdot |DPs|)$. The first term, $O(L_p \cdot (|V| + |A|))$, represents the near-linear complexity of the Pregel-based modules (\func{FindFAS}, \func{S\&G}). These steps handle the bulk of the computation and dominate the overall runtime.
% , but they are designed to be highly scalable on distributed clusters. 
The second term, $O(|A_{\text{FAS}}|^3 + |A_{\text{FAS}}| \cdot |DPs|)$, is the complexity of \func{MatrixApp$^*$}, where $O(|A_{\text{FAS}}| \cdot |DPs|)$ represents the complexity of matrix multiplication in {MatrixApp$^*$}.

% \begin{table*}[t]
% \centering
% \caption{Complexity comparison.}
% \label{tab:efficiency_comparison}
% \small
% \resizebox{1.9\columnwidth}{!}{
%     \begin{tabularx}{\textwidth}{l|l|c|c|>{\raggedright\arraybackslash}X}
%     \toprule
%     \textbf{Category} & \textbf{Algorithm / Module} & \textbf{Scalability} & \textbf{Exactness} & \textbf{Time Complexity}  \\
%     \midrule
%     \multirow{2}{*}{\textbf{Baselines}}  & \func{RollUpApp} \cite{costrollup2024}& High & Approximate & $O(S \cdot (|V|+|A|))$ \\
%     & \func{MatrixApp} & Low & Exact & $O(|V|^3)$ \\
%     \midrule
%     \multirow{4}{*}{\textbf{Ours}} & \textbf{\func{CostApp}} & \textbf{High} & \textbf{Exact} & $\mathbf{O(L_p \cdot (|V|+|A|) + |A_{\text{FAS}}|^3 + |A_{\text{FAS}}| \cdot |DPs|)}$ \\
%     \cmidrule(lr){2-5}
%     & \quad \func{FindFAS} & - & - & $O(L_p \cdot (|V|+|A|))$ \\
%     & \quad S\&G (\func{RollUpApp} + \func{GetMatrix}) & - & - & $O(L'_p \cdot (|V|+|A|))$  \\
%     & \quad \func{MatrixApp$^*$} & - & - & $O(|A_{\text{FAS}}|^3) + |A_{\text{FAS}}| \cdot |DPs|$ \\
%     \bottomrule
%     \end{tabularx}
% }
% \vspace{0.03in}
% \begin{minipage}{0.9\textwidth}
% \footnotesize\textit{Notation: $|V|$, $|A|$, $|A_{\text{FAS}}|$, and $|\mathrm{DPs}|$ are the numbers of vertices, arcs, feedback arcs, and DPs, respectively. $S$ is the superstep limit used for the approximate baseline (\S\ref{sec:exp_exactness}). $L_p$ and $L'_p$ are the longest simple-path lengths of $G$ and $G_{DAG}$.}
% \end{minipage}
% \end{table*}

\begin{table}[h!]
\centering
\vspace{-0.1in}
\caption{Complexity comparison.}
\label{tab:efficiency_comparison}
\resizebox{0.9\columnwidth}{!}{%
\begin{tabular}{l|c|c|l}
\toprule
\textbf{Algorithm / Module} & \textbf{Scalability} & \textbf{Exactness} & \textbf{Time Complexity}  \\
\midrule
\func{RollUpApp} \cite{costrollup2024}& High & Approximate & $O(S \cdot (|V|+|A|))$ \\
\func{MatrixApp} & Low & Exact & $O(|V|^3)$ \\
\midrule
\textbf{\func{CostApp}} & \textbf{High} & \textbf{Exact} & $\mathbf{O(L_p \cdot (|V|+|A|) + |A_{\text{FAS}}|^3 + |A_{\text{FAS}}| \cdot |DPs|)}$ \\
\cmidrule(lr){1-4}
\quad \func{FindFAS} & - & - & $O(L_p \cdot (|V|+|A|))$ \\
\quad S\&G (\func{RollUpApp} + \func{GetMatrix}) & - & - & $O(L'_p \cdot (|V|+|A|))$  \\
\quad \func{MatrixApp$^*$} & - & - & $O(|A_{\text{FAS}}|^3 + |A_{\text{FAS}}| \cdot |DPs|)$ \\
\bottomrule
\multicolumn{4}{l}{ \textit{Notation: $|V|$ is the number of vertices, $|A|$ is the number of arcs, $|A_{\text{FAS}}|$ is the size of the feedback arc set, |DPs| is}}
\\
\multicolumn{4}{l}{ \textit{ the number of DPs, $S$ is the number of supersteps, $L_p \approx L_p'$ are the longest simple path length of $G,G_{DAG}$, respectively.}}
\\
% \multicolumn{4}{l}{ \textit{ }}
\end{tabular}%
}

\end{table}

% \partitle{Comparative Advantage}
% \func{CostApp} successfully resolves the fundamental trade-off between exactness and scalability presented by the basic solutions. It achieves the \textit{same exactness} as the intractable \func{MatrixApp} but replaces its $O(|V|^3)$ bottleneck with a fast, localized computation, making an exact solution feasible for large graphs. It is significantly more powerful than the approximate \func{RollUpApp}, as it guarantees a complete and correct cost allocation without the compromises of heuristic termination.

% In summary, by combining scalable distributed graph processing with a small-scale algebraic solution for cyclic dependencies, \func{CostApp} delivers a both exact and scalable solution for industrial-scale data pipelines.

\section{Experiments}
\label{sec:experiments}
In this section, we conduct a comprehensive experimental evaluation of our proposed \func{CostApp} for the MCDP problem.

Our goal is to answer the following research questions:
\begin{itemize}
    \item \textbf{RQ1: Exactness and Scalability.} How \func{CostApp} outperform baselines in terms of exactness and scalability? 
    \item \textbf{RQ2: Modular Evaluation.} What are the performance characteristics of \func{CostApp}'s individual components, and where are the primary computational bottlenecks?
    \item \textbf{RQ3: Real-world Efficiency.} How effective and practical is \func{CostApp} when applied to production data pipelines?
\end{itemize}

\subsection{Experimental Setup}
\label{sec:exp_setup}

Our workloads, including Pregel-based graph computations and SQL queries, are conducted on the MaxCompute platform ~\cite{maxcompute} that is deployed on production clusters at Alibaba (a top-tier global internet enterprise).
% The matrix related modelus (\func{MatrixApp} and \func{MatrixApp$^*$}) are performed on a single machine because distributed matrix inversion suffers from challenging deployment \cite{meng2016mllib}, prohibitive communication overhead and complex synchronization requirements \cite{choi1992scalapack}. Any improved implementation on matrix inversion benefits \func{CostApp} as well. Optimizations for matrix inversion are orthogonal to our contribution and thus beyond the scope of this paper. 
The concrete platform settings are shown in Table \ref{tab:settings}. Unless otherwise specified, all tasks are executed based on the setting in Table \ref{tab:settings}. We denote \func{S\&G} as the combination of \func{RollUpApp} and \func{GetMatrix} modules in \func{CostApp} (lines 5-6 in Alg. \ref{alg:costapp}). \func{MatrixApp$^*$} denotes the linear system solver constituted in \func{CostApp} (line 8 in Alg. \ref{alg:costapp}).

\begin{table}[ht]
    \centering
    \vspace{-0.2in}
    \caption{Execution form of modules of \func{CostApp} and the setting of platforms.}
    \label{tab:settings}

    \begin{subtable}[b]{0.48\columnwidth}
        \centering
        \resizebox{0.7\linewidth}{!}{
            \begin{tabular}{c|c}
                \toprule
                Modules  & Execution form \\
                \midrule
                \textbf{\func{FindFAS} } & Pregel \\
                \textbf{\func{GraphEdit}} & SQL \\
                \textbf{\func{S\&G} } & Pregel \\
                \textbf{\func{MatrixApp$^*$}} & Single Machine \\
                \bottomrule
            \end{tabular}
        }
        \vspace{0.1in}
        \caption{Execution form of modules.} 
        \vspace{-0.4in}
        \label{tab:modules_execution} 
    \end{subtable}
    \hfill 
    \begin{subtable}[b]{0.5\columnwidth}
        \centering
        \resizebox{\linewidth}{!}{
            \begin{tabular}{lr|r}
                \toprule
                 
                \textbf{Pregel/SQL} & |workers| & 100 \\
                 & memory per worker & 30GB  \\
                 & |threads| per worker & 1  \\
                 & |cpu cores| per worker & 8  \\
                %  \midrule
                % \textbf{SQL} &  &    \\
                \midrule
                \textbf{Single Machine} & memory & 100GB \\
                 & |cpu cores| & 8  \\
                \bottomrule
            \end{tabular}
        }
        \caption{Platform settings.} 
        \vspace{-0.4in}
        \label{tab:platform_settings} 
    \end{subtable}
\end{table}

% \noindent\textbf{Baselines.} We compare \func{CostApp} against three baselines operating directly on $G$: \func{RollUpApp}, \func{MatrixApp}-D and \func{MatrixApp}-S. \func{RollUpApp} employs cost roll-up strategy \cite{costrollup2024,engh18system,wentker2019bottom,cunnama2016using} (see Alg. \ref{alg:simpleapp-pregel}) and delivers approximate results for cyclic graphs. \func{MatrixApp}-D employs Alg. \ref{alg:matrix_app} and utilizes dense matrix representation with LU decomposition \cite{lu-decomposition} to solve the global linear system. \func{MatrixApp}-S exploits the sparsity of cost allocation matrix $\mathbf{W}$ by leveraging SuperLU \cite{li2005overview}, a state-of-the-art sparse solver. The submodule \func{MatrixApp$^*$} of \func{CostApp} applies LU decomposition since the inverted matrix in Eq. \ref{eq:optimized_dp_calc} is dense.

\noindent\textbf{Baselines.} We compare \func{CostApp} against three baselines operating directly on $G$: (1) \func{RollUpApp} (Alg.~\ref{alg:simpleapp-pregel}), a cost roll-up approach~\cite{costrollup2024,engh18system,wentker2019bottom,cunnama2016using} that yields approximate results on cyclic graphs, (2) \func{MatrixApp}-D, Alg.~\ref{alg:matrix_app} with an exact global solver using dense LU decomposition~\cite{lu-decomposition}, and (3) \func{MatrixApp}-S, Alg.~\ref{alg:matrix_app} with the state-of-the-art sparse solver SuperLU~\cite{li2005overview} to exploit the sparsity of the global allocation matrix $\mathbf{W}$. \func{MatrixApp}-D and \func{MatrixApp}-S differ only in the matrix representation and solver backend. The submodule \func{MatrixApp$^*$} in \func{CostApp} applies dense LU decomposition since the reduced inverted matrix in Eq.~\ref{eq:optimized_dp_calc} is dense.

% \func{RollUpApp} apply CA operation on the orginal pipeline graph $G$ based on Pregel Job \cite{}, and deliver a approximate solution for cyclic graphs. \func{MatrixApp}-dense and \func{MatrixApp}-sparse directly solve the global linear system on the original graph $G$. Both are exact but are expected to fail on large graphs due to their prohibitive time and space complexity.
% These baselines represent the two extremes of the exactness-scalability trade-off, providing a clear contrast to our proposed \func{CostApp}.

% basic solutions in \S\ref{sec:basic_solutions}, which are directly applying \func{RollUpApp} and \func{MatrixApp} on the original graph $G$.

% \begin{itemize}
%     \item \func{MatrixApp}$(G.\mathbf{W},G.\mathbf{d})$(\S\ref{sec:matrix_app}): This algorithm solves the global system of linear equations directly. It is theoretically exact but is expected to fail on large graphs due to its $\mathcal{O}(|V|^3)$ time and $\mathcal{O}(|V|^2)$ space complexity.
%     \item \func{RollUpApp}$(G,C,V)$(\S\ref{sec:simple_app}): This is a purely iterative, Pregel-based algorithm. It is highly scalable but provides only an approximate solution for cyclic graphs. In the scalability evaluation, we use the variant \func{RollUpApp-RT} with threshold $\epsilon_t$ to ensure termination.

%     % we configured \func{RollUpApp} to terminate after a fixed number of supersteps ($\epsilon_s=100$), representing a practical trade-off between runtime and accuracy.
% \end{itemize}

\noindent
    \begin{minipage}[b!]{0.45\linewidth}
        \centering
        \vspace{-0.2in}
        \captionof{table}{Statistics of G1---G4.}
        \label{tab:graph_stats}
        
        \vspace{0.1em} 
        
        \resizebox{\linewidth}{!}{
            \begin{tabular}{l|rrr}
                \toprule
                & \textbf{\#Vertices} & \textbf{\#Arcs} & \textbf{\#DPs}  \\
                \midrule
                G1 & 28,181,925 & 66,758,378 & 10,003,135   \\
                G2 & 11,061,395  & 36,345,530 & 5,805,784   \\
                G3 &  712,772 & 3,808,320  & 309,852  \\
                G4 & 581,853 & 3,341,002  &  252,812   \\
                \bottomrule
            \end{tabular}
        }
    \end{minipage}%
    \hspace{0.6in}% 
    \begin{minipage}[b!]{0.35\linewidth}
        \centering
        \includegraphics[width=\linewidth]{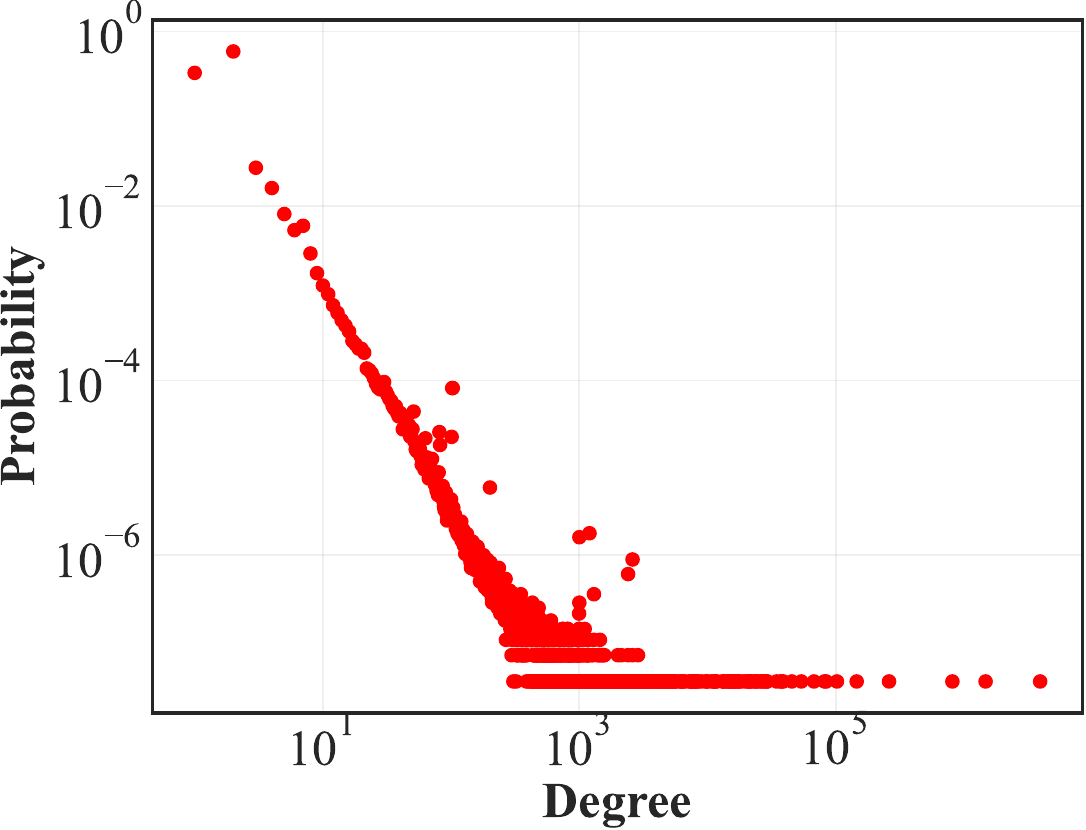}
        \vspace{-0.3in}
        \captionof{figure}{Degree distribution of G1.}
        \vspace{0.1in}
        \label{fig:degree_distribution}
    \end{minipage}

\noindent\textbf{Datasets.} We use both real-world and synthetic datasets.

\noindent\textit{Real-world Datasets}: We use four large-scale production data pipeline graphs from Alibaba, denoted as G1---G4, where G1 is the full view of the whole consolidated data pipelines. Their detailed statistics are presented in Table~\ref{tab:graph_stats}. The datasets exhibit scale-free characteristics and the degree of each node in these graphs follows a power-law distribution, with an average degree of each entity 2.5. A graphic view of the degree distribution is shown in Figure \ref{fig:degree_distribution}. 

% We open-source these datasets, and they represent the first publicly available real-world data pipeline graphs in data-intensive enterprises, to the best of our knowledge.

\noindent\textit{Synthetic Datasets}: To systematically evaluate performance under controlled conditions, we use synthetic graphs with varying parameters, including number of vertices, edge density (average degree), and the number of injected feedback arcs. To simulate the degree distribution of the scale-free real-world datasets (as in Figure \ref{fig:degree_distribution}), the graph was generated using a Barab\'asi--Albert model \cite{albert2002statistical}. To obtain a pipeline topology, we orient each BA edge from the existing vertex to the newly added vertex, producing a DAG. We then add a controlled number of later-to-earlier arcs between randomly selected vertex pairs to introduce cyclic dependencies. Weights of graphs are set based on equal-split strategy. The BA model captures degree skew but not every semantic constraint of an industrial pipeline. We therefore use the synthetic graphs to study sensitivity to graph size, density, and feedback sparsity, while the four production graphs G1--G4 evaluate performance on real pipelines.

\begin{figure}[htbp]
    \centering
    \begin{minipage}{0.4\textwidth}
        \centering
        \includegraphics[width=0.85\linewidth]{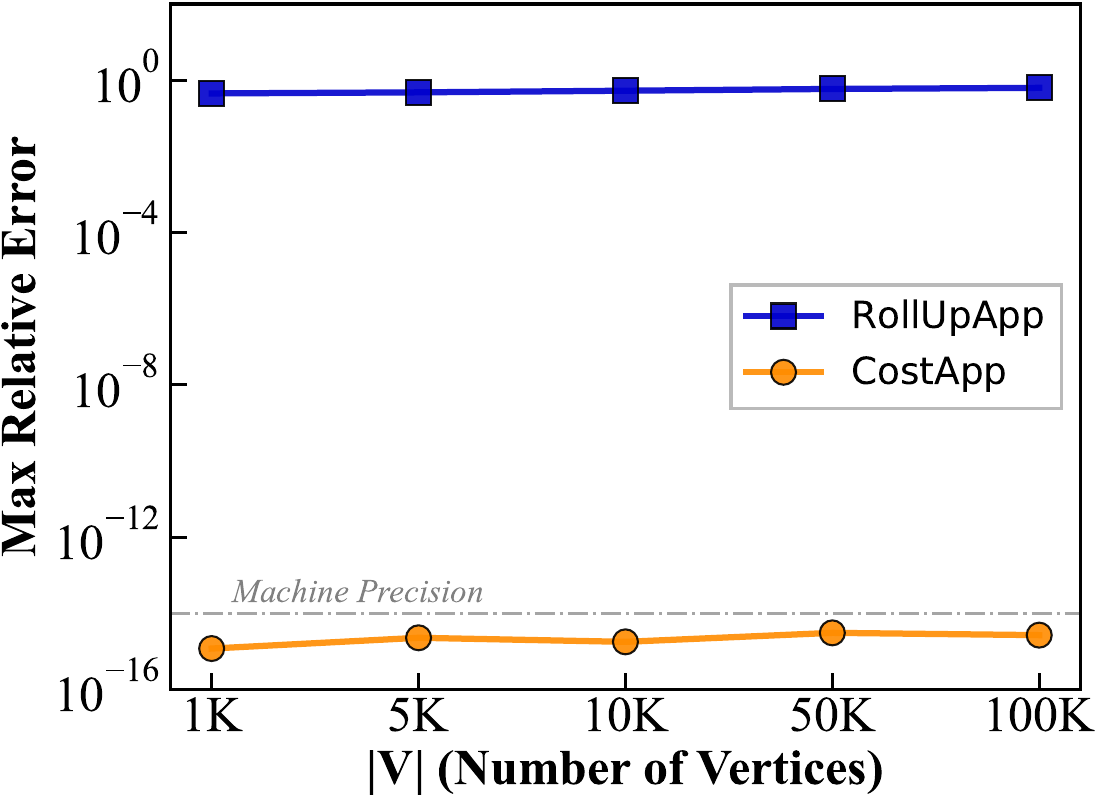}
        \vspace{-0.1in}
        \caption{Max relative error across scales.}
        \vspace{-0.1in}
        \label{fig:max_error_across_scales}
    \end{minipage}
    \hspace{0.3in}
    \begin{minipage}{0.4\textwidth}
        \centering
        \includegraphics[width=0.85\linewidth]{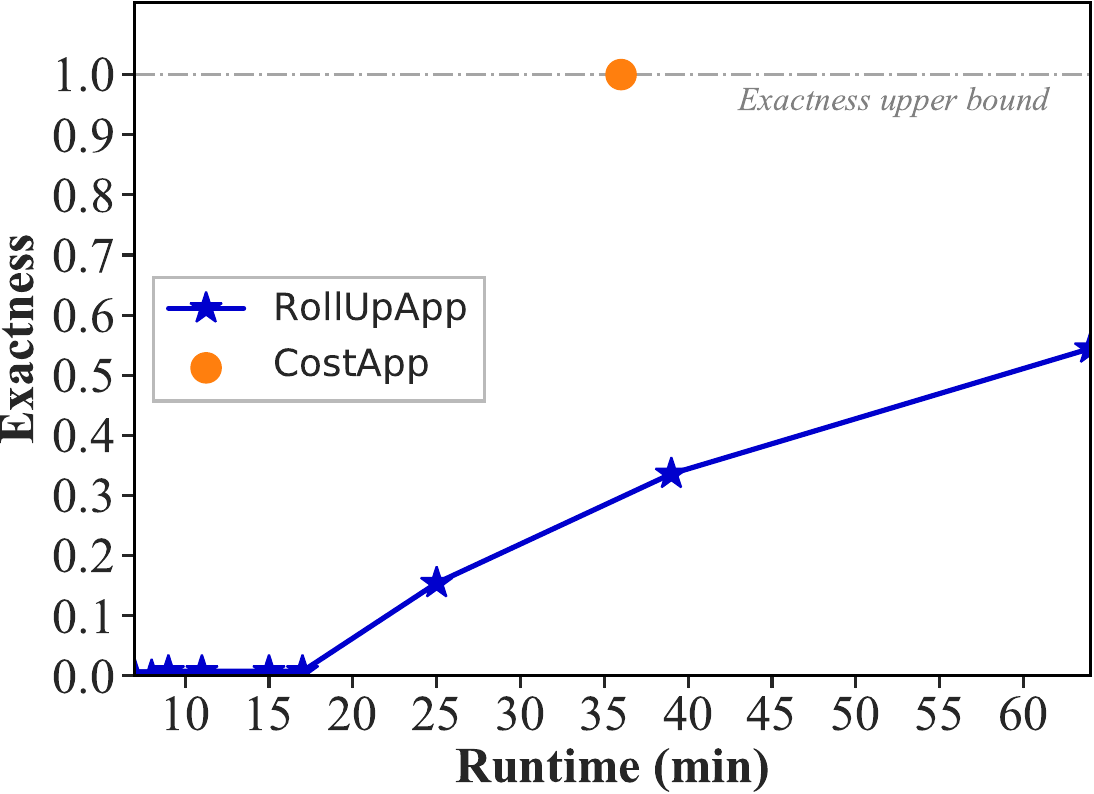} 
        \vspace{-0.1in}
        \caption{Exactness evaluation on G1.}
        \vspace{-0.1in}
        \label{fig:exactness}
    \end{minipage}
\end{figure}

\subsection{Exactness and Scalability}
\label{sec:exp_exactness_scalability}
To answer \textbf{RQ1}, we evaluate the exactness and scalability of \func{CostApp} compared to our baselines. Unless varied explicitly, the experiments use the default synthetic graph with $|V|=10^6$ and $|A|=3|V|$ to approximate the scale and density of real-world pipelines.

\subsubsection{Exactness}
\label{sec:exp_exactness} Given \func{MatrixApp} as the ground truth for exactness (Theorem \ref{thm:matrixapp_exact}), we evaluate the accuracy of \func{CostApp} and \func{RollUpApp} by measuring the Maximum Relative Error (MRE) of the estimated MC of DPs compared to the exact MC obtained from \func{MatrixApp}. \func{RollUpApp} is set to run for a fixed threshold of supersteps (=10) in Figure~\ref{fig:max_error_across_scales} to enforce termination, while in Figure~\ref{fig:exactness}, we track it with various thresholds.
% We track MRE across varying graph scales and over time.

% We evaluate the exactness of \func{CostApp} against the baselines \func{RollUpApp} (approximate) and \func{MatrixApp} (exact).

\partitle{Exactness across data scales} Figure~\ref{fig:max_error_across_scales} reports the MRE across varying pipeline sizes. \func{CostApp} consistently maintains an MRE around $10^{-15}$, consistent with standard floating-point roundoff, and thus matches the exact solution to machine precision regardless of graph scale. Conversely, \func{RollUpApp} suffers from a significant MRE ($>40\%$) due to unresolved cyclic dependencies. These results are consistent with the exactness guarantee in Theorem~\ref{thm:costapp_correctness}.

\partitle{Convergence efficiency} Figure~\ref{fig:exactness} further illustrates the exactness-runtime trade-off on a large cyclic pipeline graph G1. Exactness is measured by 1-MRE. Because unhandled costs circulate indefinitely within loops, \func{RollUpApp} requires extensive execution time merely to approach the exact solution asymptotically and does not reach machine precision within the measured time. In contrast, by effectively isolating cycles, \func{CostApp} inherently eliminates this convergence bottleneck, matching the exact solution to machine precision within a highly practical runtime (36 minutes).

% 在较小规模下，costapp和simpleapp的准确度 随着规模增长的影响。fix阈值增长规模。
% 在较大规模下，costapp和simpleapp的准确度 随着simpleApp的t数量的影响。fix规模增长阈值。 同时第二张图也展示了最差的100个DP的平均准确度和最差的DP的准确度。
% With theoretical guaranteed exactness, \func{CostApp} is expected to outperform the approximate \func{RollUpApp} in terms of accuracy. We evaluate this by measuring the accuracy of the computed MC of DPs on the real-world dataset G1. The accuracy metric is defined as the ratio of the estimated MC to the exact MC obtained from \func{CostApp}. We track accuracy over time.
% Figure~\ref{fig:exactness} presents the results. \func{CostApp} completes in approximately 30 minutes, delivering an exact cost result for all DPs. In contrast, the iterative baseline \func{RollUpApp} exhibits shortfalls. While the average accuracy of the 100 worst-performing data products (\func{RollUpApp}\_avg100) appears to converge, the minimum accuracy (\func{RollUpApp}\_min) remains critically low, reaching only 0.55 even after 60 minutes. This discrepancy highlights the unreliability of approximate methods for financial applications such as product costing.

\subsubsection{Scalability}
\label{sec:exp_scalability}

We evaluate the scalability of \func{CostApp} by analyzing the runtime performance with respect to various graph characteristics, including the number of vertices ($|V|$), the number of arcs ($|A|$), and the number of injected feedback arcs ($|A_{FAS}|$). 

\begin{figure}[t]
    \centering
    \begin{subfigure}[b]{0.3\textwidth} 
        \centering
        \includegraphics[width=\linewidth]{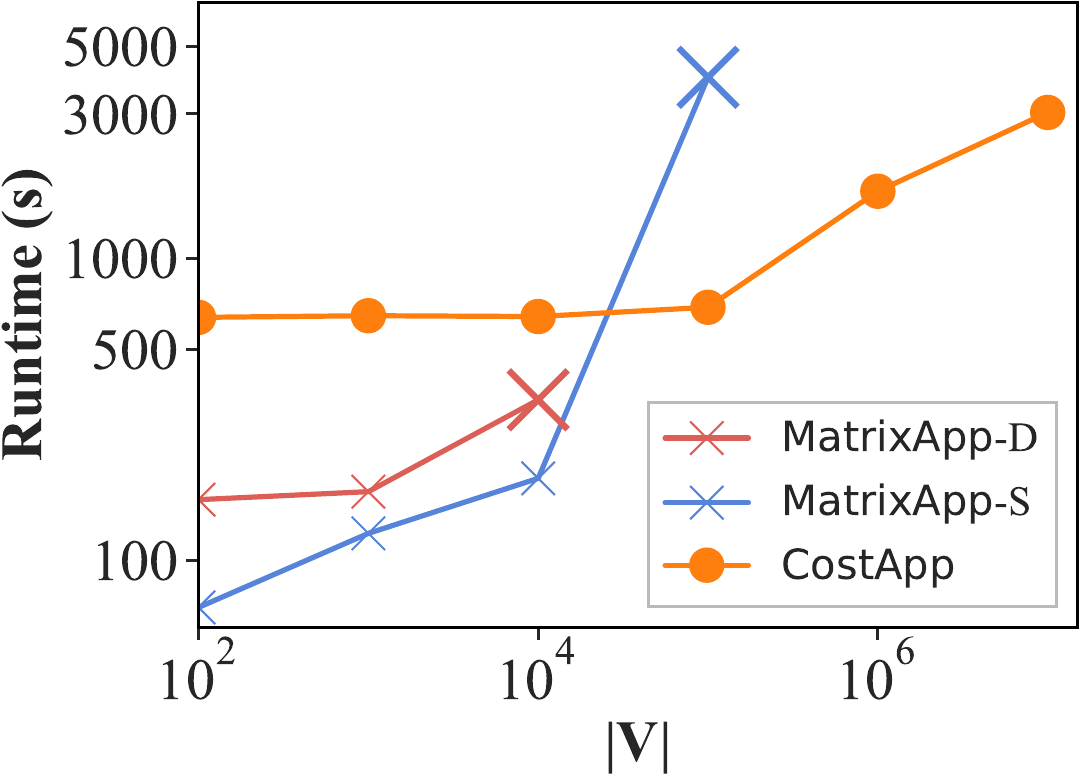} 
        \caption{Time vs $|V|$}
        \label{fig:scalability_v}
    \end{subfigure}
    % \hspace{0.1in}
    \begin{subfigure}[b]{0.3\textwidth} 
        \centering
        \includegraphics[width=\linewidth]{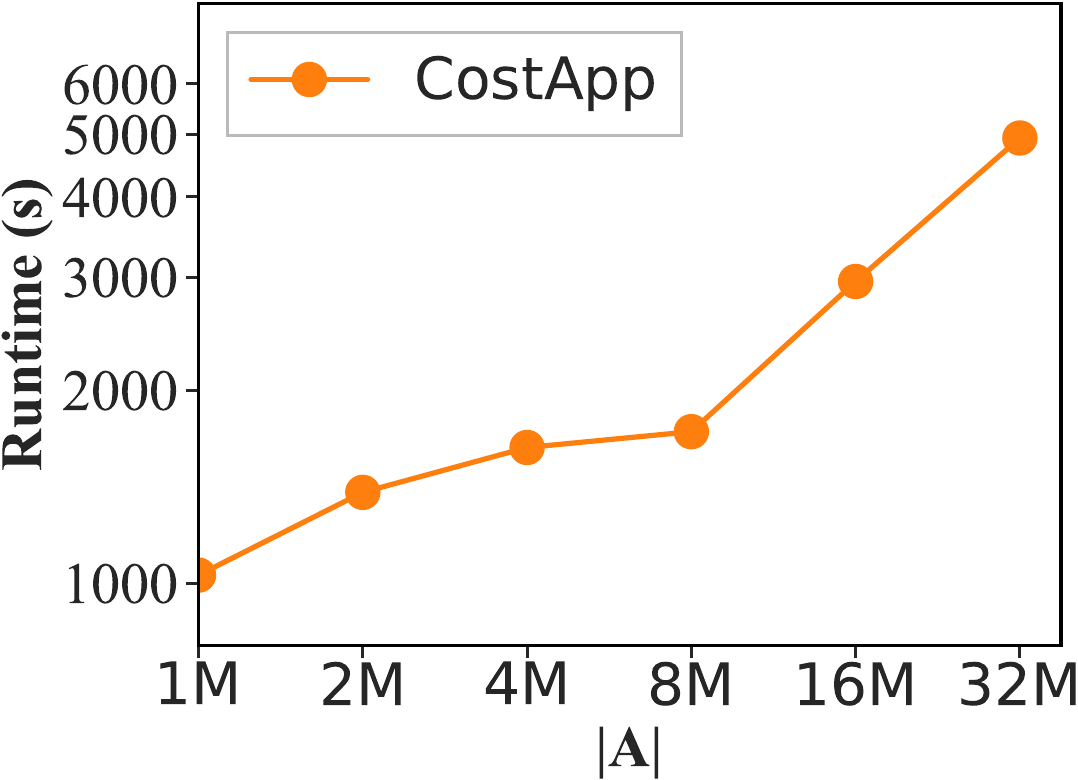} 
        \caption{Time vs $|A|$}
        \label{fig:scalability_a}
    \end{subfigure}
    % \hspace{0.1in}
    \begin{subfigure}[b]{0.3\textwidth} 
    \centering
    \includegraphics[width=\linewidth]{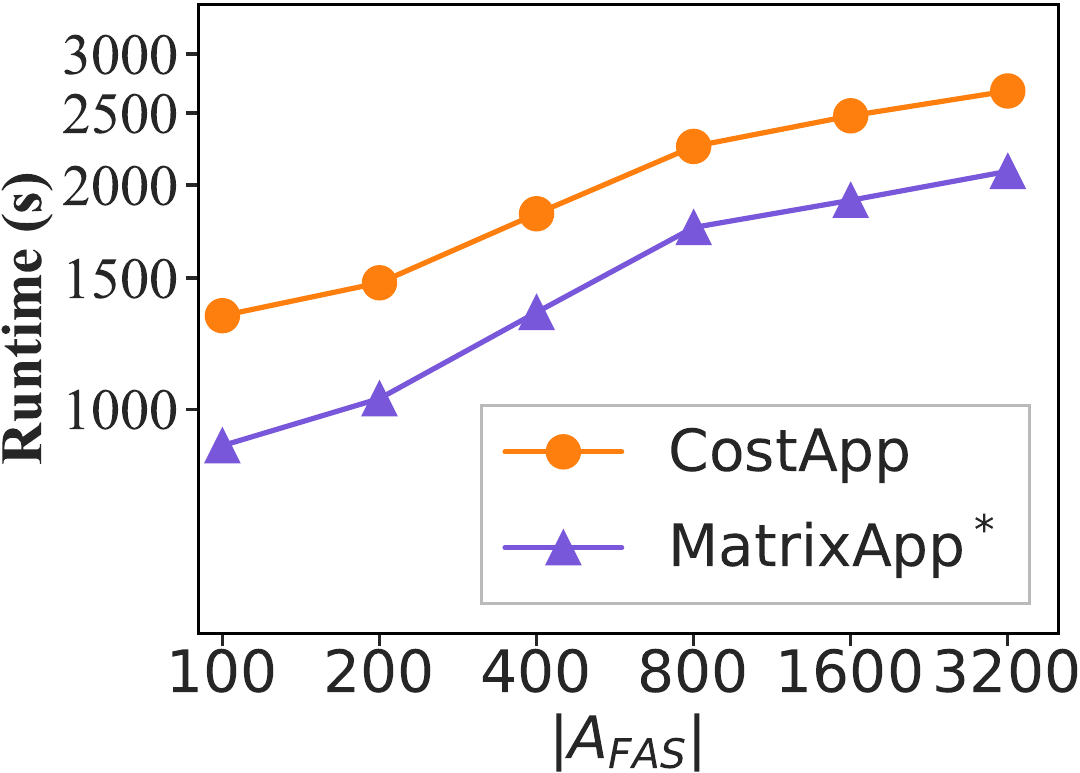} 
    \caption{Time vs $|A_{FAS}|$}
    \label{fig:scalability_f}
    \end{subfigure}

    \vspace{-0.1in}
    \caption{Comparison of scalability of algorithms.}
    \vspace{-0.1in}
    \label{fig:overall_scalability}
\end{figure}

\begin{figure}[t]
    \centering
    \begin{subfigure}[t]{0.4\textwidth} 
        \centering
         \raisebox{0pt}{\includegraphics[width=0.9\linewidth]{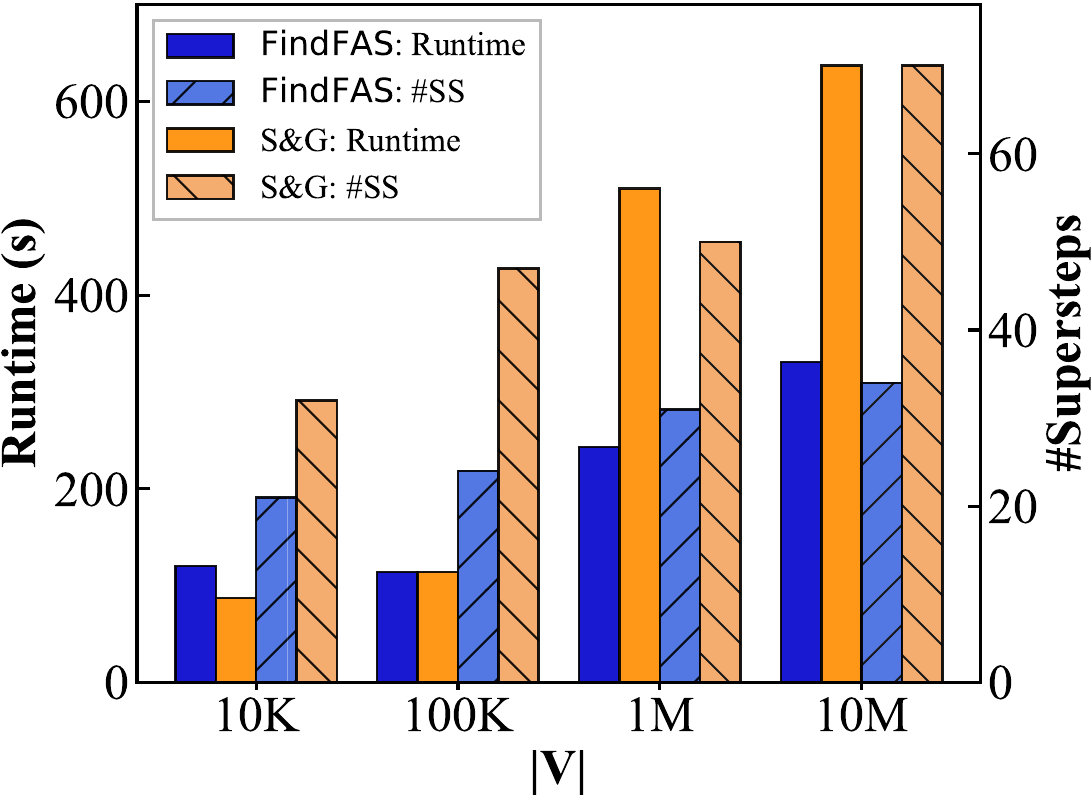} }
        \caption{Runtime and supersteps}
        \label{fig:exp_pregel_v}
    \end{subfigure}
    % \hspace{0.1in}
    \begin{subfigure}[t]{0.425\textwidth} 
        \centering
        \raisebox{1.5pt}{\includegraphics[width=0.9\linewidth]{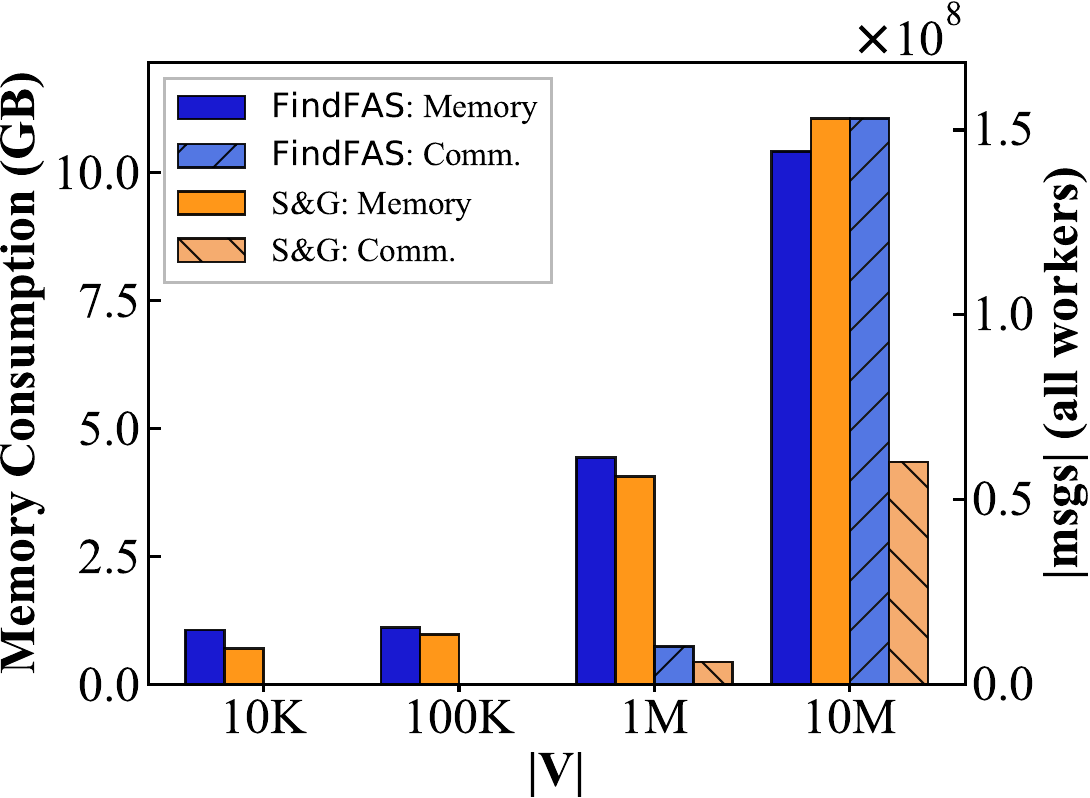}}
        \caption{Memory consumption and communicated messages}
        \label{fig:exp_fas_size}  
    \end{subfigure}
    \vspace{-0in}
    \caption{Evaluation of Pregel-based modules.}
    \vspace{-0.2in}
    \label{fig:pregel_evaluation}
\end{figure}

\partitle{Scalability with respect to $|V|$ and $|A|$} Figure~\ref{fig:overall_scalability}(a) demonstrates the scalability advantage of \func{CostApp}. 
While the baseline \func{MatrixApp}-D suffers from cubic complexity and runs out of memory (OOM) on graphs exceeding $10^4$ vertices, \func{CostApp} exhibits near-linear scalability, successfully processing graphs with $10^7$ vertices. Although \func{MatrixApp}-S is more memory-efficient (see \S\ref{sec:modular_evaluation}) than \func{MatrixApp}-D due to its exploitation of sparsity, it incurs prohibitive runtime (more than 2 days) for pipelines with more than $10^6$ vertices due to the inherently sequential nature of SuperLU. Similarly, Figure~\ref{fig:overall_scalability}(b) shows that the runtime of \func{CostApp} scales linearly with the number of arcs ($|A|$). This behavior is consistent with the communication complexity of our Pregel-based modules, where message passing dominates (\S\ref{sec:modular_evaluation}). For the arc-scaling experiment, we fix $|V|=10^6$. At this scale, \func{MatrixApp}-D runs out of memory, while \func{MatrixApp}-S exceeds the two-day time limit, as shown in Figure~\ref{fig:overall_scalability}(a). We therefore omit them from Figure~\ref{fig:overall_scalability}(b).
% todo 上一段的最后一句话 含义有点不明确

% While the baseline \func{MatrixApp}-D (and \func{MatrixApp}-S) suffers from cubic complexity and runs out of memory (OOM) on graphs exceeding $10^4$ (and $10^5$) vertices, \func{CostApp} exhibits near-linear scalability, successfully processing graphs with $10^7$ vertices. Similarly, Figure~\ref{fig:overall_scalability}(b) shows that the runtime scales linearly with the number of arcs ($|A|$). This behavior is consistent with the communication complexity of our Pregel-based modules, where message passing dominates (\S\ref{sec:modular_evaluation}). It also confirms that the distributed aggregation phase effectively handles massive graph data.

\partitle{Impact of FAS size $|A_{FAS}|$}
Figure~\ref{fig:overall_scalability}(c) isolates the impact of cycles by varying $|A_{FAS}|$ while keeping $|V|$ and $|A|$ constant. The results reveal that the total runtime grows with $|A_{FAS}|$, with a growth curve parallel to that of the \func{MatrixApp$^*$} component. The gap between the two curves represents the overhead of the distributed graph processing stages (\func{S\&G}). Since $|V|$ and $|A|$ are fixed, this distributed overhead remains constant, confirming that the marginal increase in runtime is driven solely by the matrix operations on the reduced graph. This validates that our hybrid design successfully confines the computationally expensive, cycle-dependent workload to a compact algebraic step, preventing the cubic complexity from dominating the entire pipelines.

% The results of exactness and scalability collaboratively demonstrate that \func{CostApp} uniquely provides the guaranteed exactness of a direct algebraic solution while matching the scalability required for massive graphs, overcoming the fundamental limitations of the basic solutions.

\subsection{Modular Evaluation}
\label{sec:modular_evaluation}
We profile the main components of \func{CostApp} to answer \textbf{RQ2}. The following experiments use the same default graph setting as in \S\ref{sec:exp_exactness_scalability}.

\partitle{Pregel-based modules} 
Figure \ref{fig:pregel_evaluation} presents the performance of the core Pregel-based modules of \func{CostApp}: \func{FindFAS} and \func{S\&G}. In Figure~\ref{fig:exp_pregel_v}, \#SS denotes the number of Pregel supersteps. In Figure~\ref{fig:exp_fas_size}, ``Memory'' is peak memory per worker and ``Comm.'' is the total number of messages across all workers.

As shown in Figure \ref{fig:exp_pregel_v}, both modules demonstrate near-linear scalability with respect to the number of vertices ($|V|$). Notably, the \func{S\&G} module requires nearly double the execution time of \func{FindFAS} when $|V| > 10^5$. 
Figure~\ref{fig:exp_fas_size} explains the gap. \func{FindFAS} sends more messages because labels propagate broadly, but each message contains only one scalar label. \func{S\&G} propagates key-value maps (\S\ref{sec:get_matrix}), so serialization, merging, and network payload increase with the number of source entries in each map. Its per-message computation and communication cost is therefore higher, so it takes longer despite sending fewer messages. Memory usage for both modules scales moderately, remaining within a manageable 12GB per worker even for graphs with $10^7$ vertices.

\partitle{Linear system solvers} 
Figure~\ref{fig:matrixapp_evaluation} shows that baselines \func{MatrixApp-D} and \func{MatrixApp-S} exhibit poor scalability. \func{MatrixApp-D} fails beyond $10^4$ vertices due to its $\mathcal{O}(|V|^2)$ space requirement triggering OOM errors (Figure~\ref{fig:matrixapp_mem}). While \func{MatrixApp-S} uses less memory ($<5$GB), its runtime becomes impractical ($>2$ days) for $|V| \ge 10^6$. In contrast, \func{MatrixApp$^*$} scales efficiently to $10^7$ vertices within 360s with a minimal memory footprint. This efficiency is achieved by shifting the computational bottleneck from the global vertex count $|V|$ to the significantly smaller FAS size $|A_{\text{FAS}}|$. By transforming the inversion of a massive $|V| \times |V|$ matrix into a tiny $|A_{\text{FAS}}| \times |A_{\text{FAS}}|$ matrix (\eg, $|A_{\text{FAS}}|=312 \ll |V|$ when $|V|=10^6$), \func{CostApp} successfully eliminates the intractable resource demands of algebraic methods, enabling exact costing for large-scale pipelines.

\begin{figure}[t!]
    \centering
    \begin{subfigure}[t]{0.4\textwidth} 
        \centering
        \includegraphics[width=0.9\linewidth]{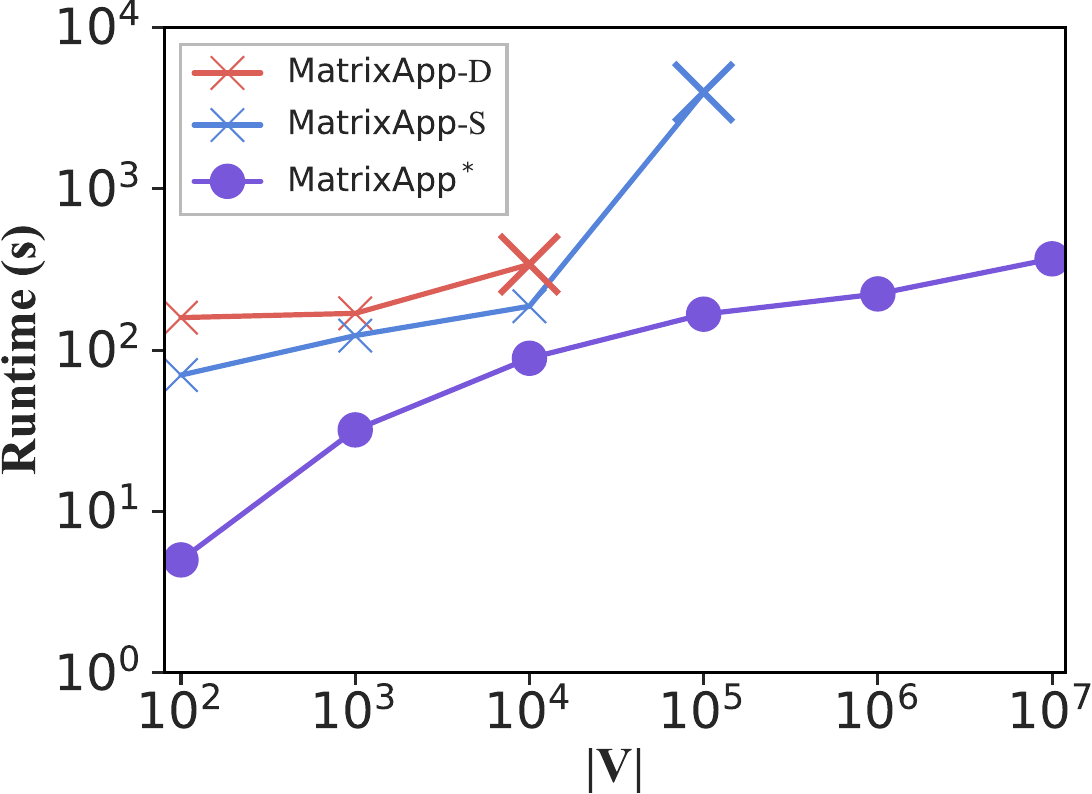}
        \caption{Runtime}
        \label{fig:matrixapp_runtime}
    \end{subfigure}
    \hspace{0.05in}
    \begin{subfigure}[t]{0.4\textwidth} 
        \centering
        \includegraphics[width=0.9\linewidth]{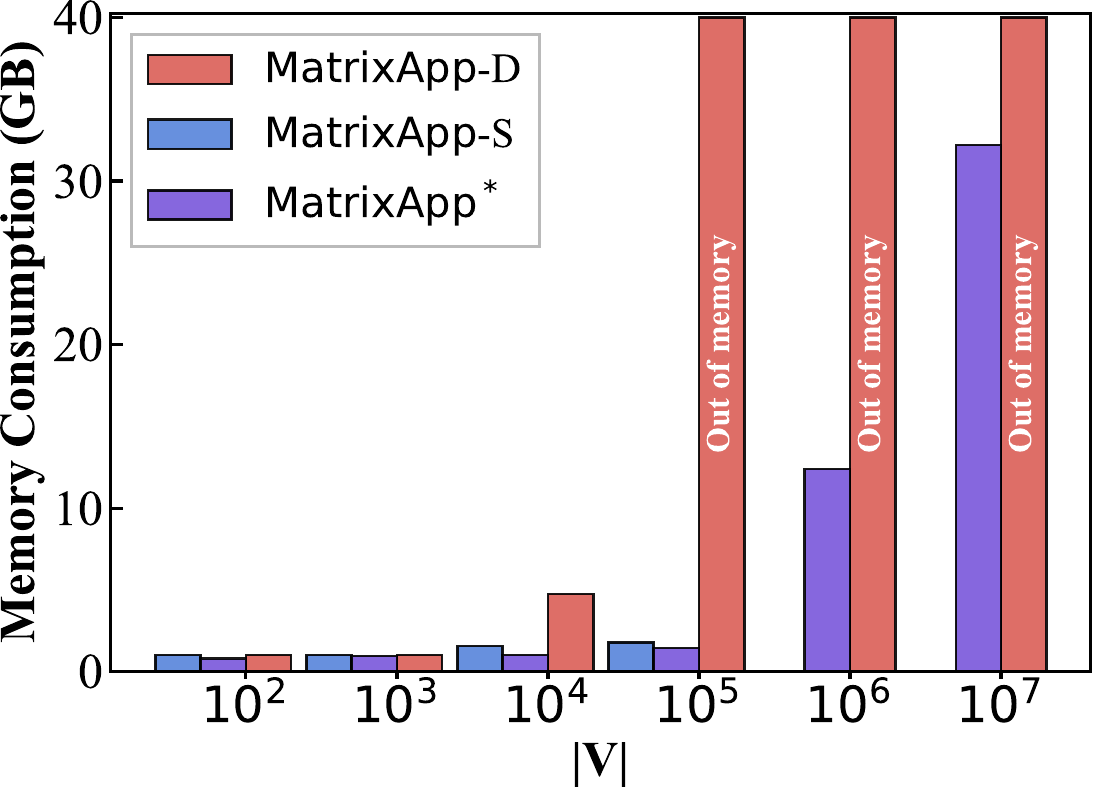}
        \caption{Memory consumption}
        \label{fig:matrixapp_mem}
    \end{subfigure}
    \vspace{-0.1in}
    \caption{Evaluation of \func{MatrixApp}-D/S and \func{MatrixApp$^*$}.}
    \vspace{-0.2in}
    \label{fig:matrixapp_evaluation}
\end{figure}

% \begin{figure}[t!]
%     \centering
%     \begin{minipage}[t]{0.23\textwidth}
%         \centering
%         \includegraphics[width=0.9\linewidth]{figs/matrixapp.pdf}
%         \caption{Performance comparison of \func{MatrixApp} and \func{MatrixApp$^*$}.}
%         \label{fig:matrixapp}
%     \end{minipage}
%     \hfill
%     \begin{minipage}[t]{0.23\textwidth}
%         \centering
%         \includegraphics[width=0.9\linewidth]{figs/matrixapp_mem.pdf}
%         \caption{Memory consumption of \func{MatrixApp} and \func{MatrixApp$^*$}.}
%         \label{fig:matrixapp_mem}
%     \end{minipage}
% \end{figure}

% \begin{figure}
% \centering
%     \includegraphics[width=0.5\linewidth]{figs/fas_size.pdf}
%     \caption{Size of FAS.}
%     \label{fig:fas_size}
% \end{figure}

% 这部分可以放到real-world这章，用来展示环路确实是稀疏的

% 说明，满足daily billing的需求
\subsection{Real-world Pipeline Evaluation}
\label{sec:exp_realworld}

To validate the performance of \func{CostApp} (\textbf{RQ3}) on practical pipeline graphs, we conducted experiments on four large-scale production data pipelines G1-G4 at Alibaba. 

% These graphs represent diverse and complex data pipeline workflows in practice. 

% They span several orders of magnitude in size and structural complexity, as detailed in Table \ref{tab:real_world_performance}, providing a robust testbed for our algorithm.

% G1:0.02%,  G2:0.03%, G3:0.03%, G4:0.03%
\partitle{FAS size in real-world pipelines}
As the size of FAS identified by \func{FindFAS} directly impacts the efficiency of the following steps, we first evaluate $|A_{FAS}|$ of G1---G4 as shown in Figure \ref{fig:fas_size}. Since \func{FindFAS} employs randomized label initialization, the size of the identified FAS ($|A_{FAS}|$) exhibits slight variance. We report the average of 10 runs. Crucially, we observe that $|A_{FAS}|$ is consistently lower than $0.05\%$ of $|A|$. This result confirms that the cycles in real-world data pipeline graphs are indeed \textit{sparse}, justifying the design of \func{CostApp} that leverages the sparsity of cycles. The small size of FAS is crucial for ensuring the efficiency of the \func{MatrixApp$^*$} step.

\partitle{Overall Performance and Breakdown}
We evaluate the performance of \func{CostApp} on G1-G4 and present the results in Figure~\ref{fig:real_eva}.

The production results demonstrate the practical efficiency of \func{CostApp} on industrial-scale graphs. On the largest graph, G1, end-to-end execution takes approximately 2250 seconds, which is within the strict time window for a daily billing cycle. Runtime decreases consistently from G1 to G4 as the graph size decreases.

The breakdown of execution time provides deep insight into why \func{CostApp} scales. 
The distributed graph processing stages (\func{FindFAS} and \func{S\&G}) dominate the workload, accounting for over 80\% of the total runtime on G1. In contrast, the \func{MatrixApp$^*$} step consumes a minor fraction ($<20\%$) of the time. Most work therefore remains in the \textit{near-linear} distributed components, while cubic matrix inversion is confined to a small reduced system.

Overall, the results show that \func{CostApp} preserves exactness and completes within one hour on all four of Alibaba's production graphs, each containing millions of vertices.

\subsection{Discussion}
\label{sec:exp_discussion}

\partitle{Comparison with SCC condensation}
An alternative approach to solving MCDP is to contract each strongly connected component (SCC), process the resulting DAG in topological order, and solve a local linear system for each non-trivial SCC. This approach is exact, but sparse feedback arcs do not imply small SCCs. For example, suppose the graph is one directed cycle through all $|V|$ vertices. Removing any one arc breaks the cycle, so $|A_{\mathrm{FAS}}|=1$. SCC condensation, however, places all $|V|$ vertices in one SCC. It must therefore solve a $|V| \times |V|$ local system, whereas \func{CostApp} solves a $1 \times 1$ residual system. Furthermore, multiple non-trivial SCCs (e.g. 1803 SCCs in G1) also require multiple local solves, while \func{CostApp} solves one residual system for all feedback dependencies.

The two approaches also organize computation differently. SCC processing must alternate between graph propagation and local solving along the condensation DAG because an SCC cannot be solved until its upstream contributions are available, incurring high runtime. \func{CostApp} instead completes scalable graph propagation and then solves one compact residual system, avoiding interleaved execution. The FAS-based reduction is therefore preferable when feedback dependencies are sparse but the SCCs they create may be large. Our experiment on G1 confirms this effect, with interleaved execution making the SCC approach $5.2\times$ slower than \func{CostApp} under the same production configuration.

\section{Related Work}

\begin{figure}[t!]
    \centering
    \begin{minipage}[t]{0.4\textwidth}
        \centering
        \includegraphics[width=0.92\linewidth]{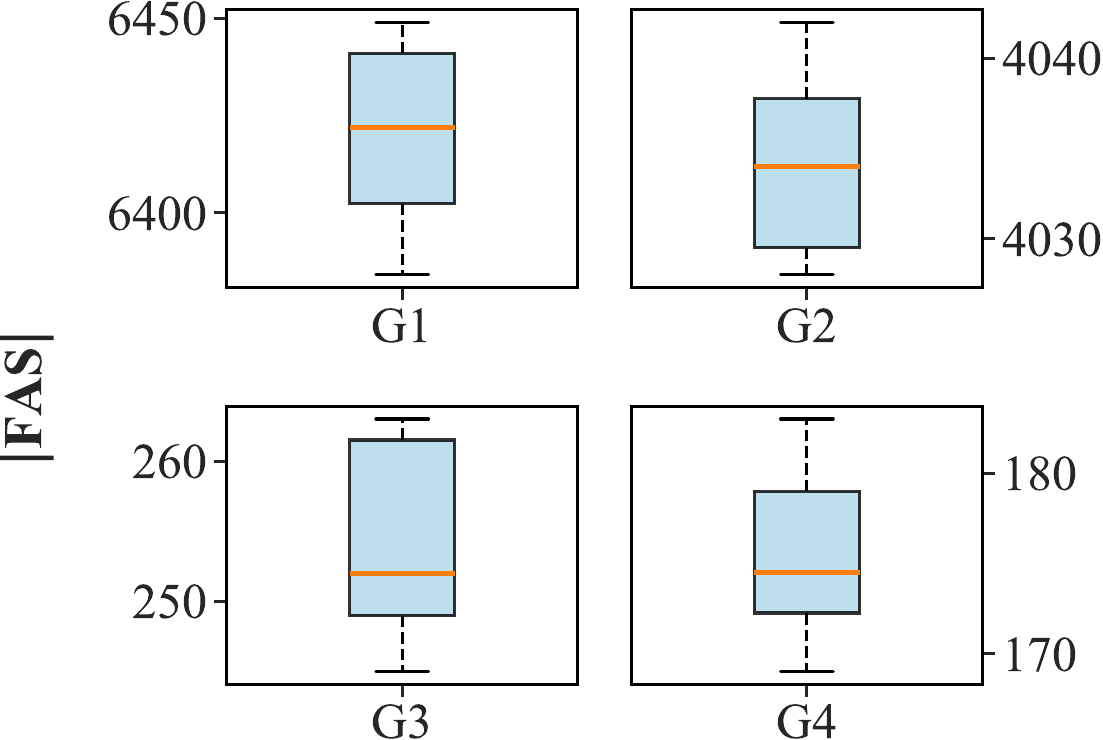}
        \vspace{-0.1in}
        \caption{Sizes of FAS on G1---G4.}
        \vspace{-0.2in}
        \label{fig:fas_size}
    \end{minipage}
    \hspace{0.3in}
    \begin{minipage}[t]{0.4\textwidth}
        \centering
        \includegraphics[width=0.92\linewidth]{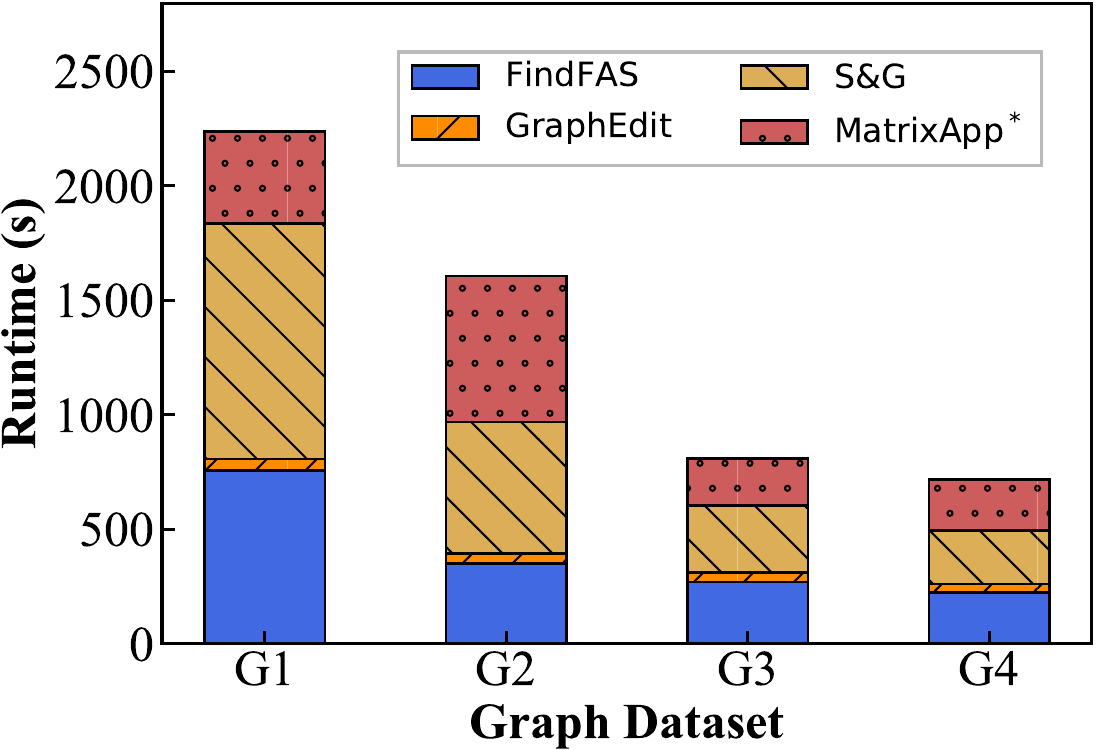}
        \vspace{-0.1in}
        \caption{Evaluation on G1---G4 and time breakdown.}
        \vspace{-0.2in}
    \label{fig:real_eva}
    \end{minipage}
\end{figure}

\subsection{Data Pipelines}
Data pipelines, particularly ETL (Extract-Transform-Load) workflows, have been widely studied as the backbone of data warehousing and analytics \cite{simitsis2023history,etl1,yang2015lenses,zhang2025streaming,khan2024overview,chai2022data,xie2016olap}. 
They enable the automated extraction of data from diverse sources, support complex transformations to ensure data quality, and facilitate loading into target stores for subsequent analysis or selling. Many big companies embrace data pipelines to build data manufacturing lines and provide business data pipeline platforms, 
\eg Meta \cite{sun2023presto}, Alibaba \cite{zhang2025streaming}, and Google \cite{google}. A typical data pipeline is executed in scheduled mode \cite{munappy2020data}, which minimizes the load on source systems by operating during off-peak hours. 
% hich improves efficiency by processing large volumes of data together, which minimizes the load on source systems by operating during off-peak hours. 
% At the same time, a real-time pipeline \cite{zhang2025streaming} aims to improve latency and provide more timely data results.

% This batch-oriented approach improves efficiency by processing large volumes of data together, minimizes the loApache Airflowad on source systems by operating during off-peak hours, and simplifies error handling and monitoring through structured, repeatable workflows. Consequently, batch processing remains a practical and cost-effective solution for large-scale data integration.

% Alkis \etal studies the logical optimization of ETL processes, modeling it as a state-space search problem \cite{simitsis2005optimizing}. 

% \subsection{Data-related Cost}
% Managing data-related costs is a central concern in modern cloud-centric computing. Early studies on cloud infrastructure emphasized the breakdown of data-related costs into server hardware, facility infrastructure, power, and networking, highlighting persistent challenges in resource underutilization and the need for system-level innovations to boost efficiency \cite{greenberg2008cost}. In recent years, the focus has shifted from static cost structures to dynamic and workload-aware cost management. This shift demands fine-grained cost management that aligns workload characteristics with the most cost-effective pricing models \cite{srivastava2024saving}. Notably, the operation cost in our work naturally include these various types of costs. 

\subsection{Cost of Digital Product in Economy}
The cost structure of digital products differs fundamentally from that of physical goods, due to major reductions in search, replication, transportation, tracking, and production costs \cite{pei2020survey}. Search and transportation are made efficient by the Internet, enabling easy discovery and access. Replication costs approach zero, allowing scale, customization, and new models such as bundling and subscriptions \cite{mackie2000pricing, shapiro1999information}. Low tracking costs facilitate personalized pricing, while raising concerns about privacy \cite{odlyzko2003privacy,fudenberg2012price}. Due to free-duplication nature of data, the production cost of data only contains fixed cost that is not related to the sharing amount. Thus, the unit production cost can approach zero through sharing as long as sufficient reuse and sales volume. These facilitate the development of innovative pricing strategies and the evolution of data marketplaces \cite{liu2021dealer,pei2023data,zhang2024survey,shuyue2025efficient,tong2018dynamic}. Different from prior works that focus on the cost structure and pricing strategies of digital products, our work aims to product costing in data pipelines, which is a prerequisite for pricing and orthogonal to cost structure.

\subsection{Value and Cost Allocation}
Prior work also studies how to choose an allocation rule. Shapley-value methods allocate a coalition's value according to each participant's marginal contribution and follow well-established fairness principles~\cite{shapley1953value,aadithya2010efficient}. In data valuation, they infer contribution from repeated utility evaluations, but exact computation generally requires exponentially many coalitions, so scalable methods rely on approximation or task-specific structure~\cite{ghorbani2019data,jia2019towards,wei2020efficient,kwon2022beta,pang2026asss,zhang2023efficient}. Our setting instead starts from measured resource costs and a business-selected allocation policy. A value-based rule, including one derived from Shapley values, may determine $w$. \func{CostApp} is complementary to these methods because it computes the exact pipeline-wide costs implied by those weights.

\section{Conclusion and Future Work}

In this paper, we addressed the critical challenge of product costing in large-scale data pipelines. By formalizing the MCDP problem to account for the non-rivalrous nature of data, we proposed \func{CostApp}, a hybrid algorithm that achieves both high scalability and result exactness. By strategically isolating cyclic dependencies, \func{CostApp} reduces the computational complexity of exact costing from cubic to near-linear relative to the graph size. Extensive experiments show \func{CostApp} scales to process a production pipeline with over 28 million entities within an hour while preserving exactness, supporting pipelines up to $10^2\times$ larger than the direct matrix-based baseline.

% Extensive experiments on real-world production pipelines at Alibaba\footnotemark[1] confirm that our approach delivers dramatic performance gains without compromising the exactness of the results.

% In this paper, we addressed the critical challenge of product costing in large-scale data pipelines. We formalized the MCDP problem to account for the non-rivalrous nature of data. We proposed \func{CostApp}, a hybrid algorithm that achieves both scalability and result exactness. By strategically isolating cyclic dependencies via a feedback arc set, \func{CostApp} reduces the computational complexity of exact costing from cubic to near-linear relative to the graph size. Extensive experiments on real-world production pipelines at Alibaba\footnotemark[1] demonstrate that our approach achieves both result exactness and high scalability.

\partitle{Future work} Our work opens several promising research directions. First, extending \func{CostApp} to support \textit{explainable product costing}. It means the exact cost breakdown of each data product is fully traceable to its upstream origins, and would provide vital insights for fine-grained pipeline optimization. Second, the core capability of \func{CostApp} in efficiently solving cyclic, large-scale sparse linear systems holds vast potential beyond data pipelines. It presents a promising pathway for accelerating numerical solvers in much broader domains, such as exact value function evaluations in Markov Decision Processes (MDPs) and solving the Leontief inverse in macroscopic input-output economic modeling \cite{leontief1986input,leontief1974sructure,ten2009input}. Detailed discussions on these extensions are provided in Appendix~\ref{app:future_work}.

\section*{Acknowledgement}
This work was supported in part by NSFC (U23A20306, 62472378, U25A20430, 62472449), the Zhejiang Provincial Natural Science Foundation for Distinguished Young Scholars (LR25F020001), and the Zhejiang Province Pioneer Plan (2024C01074, 2025C01084).

\bibliographystyle{ACM-Reference-Format}
\bibliography{mybibfile}

\clearpage
\appendix

\section{Proof of Theorems and Lemmas}

\subsection{Proof of Theorem \ref{thm:cost_conservation}}
\label{appendix:proof_cost_conservation}
% We prove Theorem \ref{thm:cost_conservation} as follows.

\begin{proof}
We begin with the total direct cost of the system. By rearranging Eq. \ref{equ:ma}, we can express the direct cost of any single vertex $v$ as $\mathcal{C}_{dc}(v) = \mathcal{C}_{mc}(v) - \mathcal{C}_{ic}(v)$. Summing over all vertices in the graph $V(G)$:
\begin{equation}
\label{eq:proof_start}
\begin{aligned}
    \sum_{v \in V(G)} \mathcal{C}_{dc}(v) &= \sum_{v \in V(G)} (\mathcal{C}_{mc}(v) - \mathcal{C}_{ic}(v)) \\
    &= \sum_{v \in V(G)} \mathcal{C}_{mc}(v) - \sum_{v \in V(G)} \mathcal{C}_{ic}(v)
\end{aligned}
\end{equation}
Let us analyze the total input cost term, $\sum_{v \in V(G)} \mathcal{C}_{ic}(v)$. By the definition of input cost $\mathcal{C}_{ic}(v)$, we get:
$$
\sum_{v \in V(G)} \mathcal{C}_{ic}(v) = \sum_{v \in V(G)} \sum_{u \in U_G^1(v)} w(u,v) \cdot \mathcal{C}_{mc}(u)
$$
This double summation iterates over all directed arcs $\langle u,v\rangle \in A(G)$. We can change the order of summation to first sum over all potential source nodes $u \in V(G)$ and then over their direct downstream neighbors $v \in D_G^1(u)$:
$$
\sum_{v \in V(G)} \mathcal{C}_{ic}(v) = \sum_{u \in V(G)} \sum_{v \in D_G^1(u)} w(u,v) \cdot \mathcal{C}_{mc}(u)
$$
We can factor out $\mathcal{C}_{mc}(u)$ from the inner summation, as it is constant with respect to $v$:
$$
\sum_{v \in V(G)} \mathcal{C}_{ic}(v) = \sum_{u \in V(G)} \left( \mathcal{C}_{mc}(u) \cdot \sum_{v \in D_G^1(u)} w(u,v) \right)
$$
By the cost conservation property (Definition \ref{def:cost_conservation}), the inner sum $\sum_{v \in D_G^1(u)} w(u,v)$ equals 1 for any non-leaf node $u \in V(G) \setminus L(G)$. For any leaf node $p \in L(G)$, its set of downstream neighbors $D_G^1(p)$ is empty, so the inner sum is 0. Therefore, the expression simplifies to the sum of manufacturing costs over all non-leaf nodes:
$$
\sum_{v \in V(G)} \mathcal{C}_{ic}(v) = \sum_{u \in V(G) \setminus L(G)} \mathcal{C}_{mc}(u) \cdot 1 = \sum_{u \in V(G) \setminus L(G)} \mathcal{C}_{mc}(u)
$$
Substituting this result back into Eq. \ref{eq:proof_start}:
$$
\sum_{v \in V(G)} \mathcal{C}_{dc}(v) = \sum_{v \in V(G)} \mathcal{C}_{mc}(v) - \sum_{v \in V(G) \setminus L(G)} \mathcal{C}_{mc}(v)
$$
The right-hand side is the sum of manufacturing costs over all vertices minus the sum of manufacturing costs over all non-leaf vertices. This difference is exactly the sum of manufacturing costs over the leaf vertices, $L(G)$.
$$
\sum_{v \in V(G)} \mathcal{C}_{dc}(v) = \sum_{p \in L(G)} \mathcal{C}_{mc}(p)
$$
This completes the proof.
\end{proof}

\subsection{Proof of Theorem \ref{thm:simple_ca_cycle}}
\label{sec:proof_simple_ca_cycle}

\begin{proof}
Let $C^k(v)$ denote the cost held by vertex $v$ at the beginning of superstep $k$, with $C^1(v) = \mathcal{C}_{dc}(v)$. The algorithm terminates when $C^k(v) = 0$ for all non-leaf vertices $v \in V \setminus L(G)$. The proof proceeds in two parts: termination and exactness.

\partitle{1. Termination}
In a DAG, the cost propagation is unidirectional. A cost contribution from an initial $\mathcal{C}_{dc}(u)$ can only reach a vertex $v$ via a path from $u$ to $v$. After one superstep, this cost can reach vertices at a path distance of 1. After $k$ supersteps, it can reach vertices at a path distance of at most $k$.

Let $P = \max_{u,v \in V} \{\text{length of longest path from } u \text{ to } v\}$ be the longest path length in $G$. After $P$ supersteps, any cost originating from any vertex $u$ will have propagated along all possible paths of length up to $P$. Since no path in the graph is longer than $P$, all costs must have passed through all intermediate (non-leaf) vertices on their respective paths and accumulated at the leaf vertices.

Consequently, at the beginning of superstep $P+1$, no non-leaf vertex $v \in V \setminus L(G)$ can receive any message, as all upstream cost flows have already terminated at leaves. Since all non-leaf vertices also distributed their own costs in prior supersteps, $C_{P+1}(v) = 0$ for all $v \in V \setminus L(G)$. This satisfies the termination condition, proving the algorithm halts in at most $P$ supersteps.

\partitle{2. Exactness}
Let $C^{final}(p)$ be the final cost at a leaf $p \in L(G)$ upon termination. We must show $C^{final}(p) = \mathcal{C}_{mc}(p)$.

First, observe that the total cost is conserved. In any superstep, for any active non-leaf vertex $v$, the cost it sends out is
\[
\sum_{v' \in D_G^1(v)} w_{v \to v'} C^k(v) = C^k(v),
\]
because the outgoing weights satisfy
\[
\sum_{v' \in D_G^1(v)} w_{v \to v'} = 1.
\]
Thus, $\sum_{v \in V} C^{k+1}(v) = \sum_{v \in V} C^k(v)$. Upon termination,
\[ \sum_{p \in L(G)} C^{final}(p) = \sum_{v \in V} C^1(v) = \sum_{v \in V} \mathcal{C}_{dc}(v), \]
satisfying cost conservation.

Let $\mathcal{C}^*(v)$ be the total cost that flows through or terminates at vertex $v$ throughout the algorithm's execution. This is the sum of its direct cost and all costs received from upstream neighbors over all supersteps.
\begin{equation}
    \mathcal{C}^*(v) = \mathcal{C}_{dc}(v) + \sum_{u \in U_G^1(v)} (\text{total cost received from } u)
\end{equation}
The total cost received by $v$ from an upstream neighbor $u$ is the proportion $w(u,v)$ of the total cost that flowed through $u$, which is precisely $w(u,v) \cdot \mathcal{C}^*(u)$. Substituting this gives:
\begin{equation} \label{eq:proof_correctness}
    \mathcal{C}^*(v) = \mathcal{C}_{dc}(v) + \sum_{u \in U_G^1(v)} w(u,v) \cdot \mathcal{C}^*(u)
\end{equation}

% 一定有Unique solution有两种理解方式。
% 1. DAG一定有拓扑排序，按拓扑排序分配cost一定可计算出唯一确定值。
% 2. DAG的权重矩阵W是上三角矩阵，I-W^T是下三角矩阵，且对角线元素均为正数（乘积不为0），因此I-W^T可逆。从而， 对于线性系统(I-W^T)x=b来说，x有唯一解。
% 拓展：当不是DAG时，W不再是三角矩阵，此时是否(I-W^T)是否可逆取决于权重（谱半径）
This system of linear equations defining $\mathcal{C}^*(v)$ for all $v \in V$ is identical to the one defining the manufacturing cost $\mathcal{C}_{mc}(v)$ in Equations \ref{equ:ma}. Since $G$ is a DAG, this system has a unique solution. Therefore, $\mathcal{C}^*(v) = \mathcal{C}_{mc}(v)$ for all $v \in V$.

For a leaf node $p \in L(G)$, it only accumulates cost and never distributes it. Its final cost $C^{final}(p)$ is the total cost that has flowed into it, which is $\mathcal{C}^*(p)$. Thus:
\begin{equation}
    C^{final}(p) = \mathcal{C}^*(p) = \mathcal{C}_{mc}(p)
\end{equation}
This confirms that the algorithm exactly computes the manufacturing cost for all data products.
\end{proof}

\subsection{Proof of Theorem \ref{thm:matrixapp_exact}}
\label{appendix:proof_matrixapp_exact}

% The following is the proof of Theorem \ref{thm:matrixapp_exact}.
\begin{proof}
By the definition of the MCDP problem, the manufacturing cost of any entity $v_i$ must comply to the recursive relation in Eq. \ref{equ:ma}: $\mathcal{C}_{mc}(v_i) = \mathcal{C}_{dc}(v_i) + \sum_{u \in U_G^1(v_i)} w(u,v_i) \cdot \mathcal{C}_{mc}(u)$. Based on the analysis in Appendix \ref{sec:matrix_form_derivation}, this element-wise definition is equivalent to the linear system formulated in Eq. \ref{equ:matrix_form}, \ie $(\mathbf{I} - \mathbf{W}^\mathsf{T})\mathbf{c} = \mathbf{d}$. Based on the analysis in Appendix \ref{appendix:proof_matrix_inver}, the spectral structure of the allocation matrix $\mathbf{W}$ ensures that $(\mathbf{I} - \mathbf{W}^\mathsf{T})$ is non-singular and invertible. Consequently, a unique algebraic solution $\mathbf{c}^* = (\mathbf{I} - \mathbf{W}^\mathsf{T})^{-1} \mathbf{d}$ strictly exists. Since \func{MatrixApp} directly implements this exact closed-form inversion and extracts the subset of answers for the terminal data products $p \in L(G)$, it inherently captures the true MC of DPs. Thus, it exactly solves the MCDP problem.
\end{proof}

\subsection{Proof of Matrix Invertibility}
\label{appendix:proof_matrix_inver}

% 利用行和 ≤1 + 存在叶子行 <1，先得出 ρ(W)≤1；
% 再用“必须能走到叶子”的假设排除 ρ(W)=1 的情况；
% 最后用谱半径 <1 推出 I–Wᵀ 可逆。

We prove that $(\mathbf{I} - \mathbf{W}^\mathsf{T})$ is invertible by showing that the spectral radius of $\mathbf{W}$ satisfies $\rho(\mathbf{W}) < 1$. Since a matrix and its transpose share the same eigenvalues, this suffices to guarantee that all eigenvalues of $(\mathbf{I} - \mathbf{W}^\mathsf{T})$ are strictly positive, and thus the matrix is nonsingular.

Recall that $\mathbf{W}$ is the cost allocation matrix defined in \S\ref{sec:matrix_app}. By construction:
\begin{enumerate}
    \item $\mathbf{W}$ is non-negative: $\mathbf{W}[i,j] \ge 0$ for all $i,j$.
    \item For every non-leaf vertex $u \in V(G) \setminus L(G)$, the outgoing weights satisfy the conservation constraint
    \[
      \sum_{v \in D_G^1(u)} w(u,v) = 1,
    \]
   so the corresponding row of $\mathbf{W}$ sums to~1.
   \item For every leaf vertex $p \in L(G)$, we have $D_G^1(p) = \emptyset$, hence the entire row of $\mathbf{W}$ is zero.
\end{enumerate}
Thus every row sum of $\mathbf{W}$ is at most~1, and at least one row sum (any leaf row) is strictly less than~1.

\paragraph{Step 1: $\rho(\mathbf{W}) \le 1$.}
Let $\lambda$ be an eigenvalue of $\mathbf{W}$ with eigenvector $\mathbf{x} \ne \mathbf{0}$, so $\mathbf{W}\mathbf{x} = \lambda \mathbf{x}$. Let $k$ be an index such that
\[
  |x_k| = \max_i |x_i| > 0.
\]
Consider the $k$-th coordinate of the eigen-equation:
\[
  \lambda x_k = \sum_{j=1}^n \mathbf{W}[k,j] x_j.
\]
Taking absolute values and using the triangle inequality,
\begin{align*}
  |\lambda|\,|x_k|
  &= \left|\sum_{j=1}^n \mathbf{W}[k,j] x_j\right| \\
  &\le \sum_{j=1}^n \mathbf{W}[k,j] |x_j| \\
  &\le \left(\sum_{j=1}^n \mathbf{W}[k,j]\right) |x_k| \\
  &\le |x_k|.
\end{align*}
Since $|x_k|>0$, we obtain $|\lambda| \le 1$. Hence $\rho(\mathbf{W}) \le 1$.

\paragraph{Step 2: $\rho(\mathbf{W}) < 1$.}
We now show that $|\lambda| = 1$ is impossible. Suppose, for contradiction, that there exists an eigenvalue $\lambda$ with $|\lambda|=1$. For the chain of inequalities above to hold as equalities, we must simultaneously have:

\begin{enumerate}
  \item The row sum of row $k$ equals~1, i.e.,
  \[
    \sum_{j=1}^n \mathbf{W}[k,j] = 1.
  \]
      Thus vertex $k$ is not a leaf (since leaf rows are all zeros).
  \item For every $j$ with $\mathbf{W}[k,j] > 0$, we must have $|x_j| = |x_k|$. Otherwise the second inequality would be strict.
\end{enumerate}

Condition~(2) implies that every direct downstream neighbor of $k$ with a positive weight also has the same maximal modulus $|x_j| = |x_k|$. Repeating this argument along any directed path starting from $k$, we conclude that all vertices reachable from $k$ with positive path weight also satisfy $|x_i| = |x_k|$.

In a valid data pipeline, every non-leaf vertex has at least one path to some leaf data product. Therefore, there exists a leaf vertex $p \in L(G)$ reachable from $k$ via a directed path with positive total weight. By the propagation argument above, this implies
\[
  |x_p| = |x_k| > 0.
\]

However, since $p$ is a leaf, its outgoing weights are all zero, so the $p$-th row of $\mathbf{W}$ is the zero vector. The eigen-equation at coordinate $p$ is
\[
  \lambda x_p = \sum_{j=1}^n \mathbf{W}[p,j] x_j = 0.
\]
With $|\lambda| = 1$, this forces $x_p = 0$, contradicting $|x_p| = |x_k| > 0$.

Therefore no eigenvalue of $\mathbf{W}$ can satisfy $|\lambda| = 1$, and thus
\[
  \rho(\mathbf{W}) < 1.
\]

\paragraph{Step 3: Invertibility of $(\mathbf{I} - \mathbf{W}^\mathsf{T})$.}
Since $\rho(\mathbf{W}) < 1$ and $\mathbf{W}$ and $\mathbf{W}^\mathsf{T}$ share the same eigenvalues, every eigenvalue of $\mathbf{W}^\mathsf{T}$ also has modulus strictly less than~1. Hence every eigenvalue of $(\mathbf{I} - \mathbf{W}^\mathsf{T})$ is strictly positive, so $(\mathbf{I} - \mathbf{W}^\mathsf{T})$ is nonsingular and invertible. This guarantees that the linear system
\[
  \mathbf{c} = \mathbf{d} + \mathbf{W}^\mathsf{T}\mathbf{c}
\]
has a unique solution $\mathbf{c} = (\mathbf{I} - \mathbf{W}^\mathsf{T})^{-1}\mathbf{d}$.

\subsection{Proof of Lemma \ref{lem:FAS}}
\label{appendix:proof_fas}

The proof of Lemma~\ref{lem:FAS} is provided below.
\begin{proof}
% 中文直观解释：考虑环上的最小标签顶点 x。假设 Label(x)=[1,2]，沿环从 x 到达某个顶点 u 时，Label(u) 以 [1,2] 为前缀，例如 [1,2,4]。u 沿下一条弧向 v 提议的标签只会在其后追加一个局部后缀，例如 [1,2,4,3]，因此该提议仍以 [1,2] 为前缀。v 最终保留的标签不会大于这个提议，同时由于 x 是环中标签最小的顶点，它也不会小于 [1,2]。按照 Definition~\ref{def:lexicographic_order} 的字典序，位于 [1,2] 与其扩展 [1,2,4,3] 之间的标签必须仍以 [1,2] 为前缀。因此，这个前缀关系可以沿整个环逐步保持。设 y 是环上回到 x 之前的顶点，则 Label(x) 是 Label(y) 的真前缀，所以闭合弧 \langle y,x\rangle 会被选入 A_{\mathrm{FAS}}。对任意有向环都能找到这样一条弧，因此删除 A_{\mathrm{FAS}} 后不再存在有向环。
It suffices to show that every directed cycle contains an arc $\langle u,v\rangle$ for which $Label(v)$ is a prefix of $Label(u)$. After \func{Label} stabilizes, each vertex stores the smallest label proposed along its incoming paths. The labels of distinct vertices are distinct because each root has a unique initial label and each outgoing arc has a unique local suffix. Consider any directed cycle. Since the cycle is finite and the labels are totally ordered by Definition~\ref{def:lexicographic_order}, it contains a vertex $x$ with the smallest final label. We follow the cycle from $x$ and show inductively that $Label(x)$ is a prefix of every label encountered before returning to $x$. Suppose this holds at vertex $u$. The proposal sent along an arc $\langle u,v\rangle$ extends $Label(u)$, so it also has $Label(x)$ as a prefix. Because $v$ keeps the smallest proposal, $Label(v)$ is no greater than this extension. By the choice of $x$, it is also no smaller than $Label(x)$. Under Definition~\ref{def:lexicographic_order}, any label between a label and one of its extensions must have the former as a prefix. Otherwise, it would either end before the former, making it smaller, or its first differing component would make it smaller than the former or larger than the extension. Thus, $Label(x)$ is a prefix of $Label(v)$. Let $y$ be the predecessor of $x$ on the cycle. Since $x$ and $y$ have distinct labels, $Label(x)$ is a proper prefix of $Label(y)$, so $\langle y,x\rangle\in A_{\text{FAS}}$. Hence every directed cycle contains an arc removed from $G'=(V,A\setminus A_{\text{FAS}})$, and $G'$ is a DAG.
\end{proof}

\subsection{Proof of Theorem \ref{thm:graph_edit}}
\label{appendix:proof_graph_edit}

\begin{proof}
Let $\mathbf{c}^*$ and $\mathbf{c}_1^*$ be the unique manufacturing cost vectors that solve the linear systems for graphs $G$ and $G_1$, respectively. We need to show that for any original vertex $v \in V(G)$, its manufacturing cost is the same in both solutions, i.e., $\mathbf{c}^*(v) = \mathbf{c}_1^*(v)$.

The \func{GraphEdit} procedure modifies the graph locally for each feedback arc $\langle u,v\rangle \in A_{\text{FAS}}$. We next analyze the effect of a single such transformation. The cost equations for all vertices not involved in any feedback arc remain unchanged. We only need to examine the equations for vertices $u$, $v$, and the new auxiliary vertices $v_{sk}$ and $v_{sr}$.

In the original graph $G$, the manufacturing cost for vertex $v$ includes a term from $u$:
\begin{equation} \label{eq:proof_original_v}
\mathbf{c}^*(v) = \mathcal{C}_{dc}(v) + w(u,v) \cdot \mathbf{c}^*(u) + \sum_{z \in U_G^1(v) \setminus \{u\}} w(z,v) \cdot \mathbf{c}^*(z)
\end{equation}

In the transformed graph $G_1$, the initial cost state is denoted by $C_1$, and the arc $\langle u,v\rangle$ is replaced by the path $u \to v_{sk} \to v_{sr} \to v$. The new cost equations are:
\begin{align}
\mathbf{c}_1^*(v_{sk}) &= C_1(v_{sk}) + w(u,v_{sk}) \cdot \mathbf{c}_1^*(u) \label{eq:proof_vsk} \\
\mathbf{c}_1^*(v_{sr}) &= C_1(v_{sr}) + w(v_{sk},v_{sr}) \cdot \mathbf{c}_1^*(v_{sk}) \label{eq:proof_vsr} \\
\mathbf{c}_1^*(v) &= C_1(v) + w(v_{sr},v) \cdot \mathbf{c}_1^*(v_{sr}) \notag \\
&\quad + \sum_{z \in U_{G_1}^1(v) \setminus \{v_{sr}\}} w(z,v) \cdot \mathbf{c}_1^*(z) \label{eq:proof_new_v}
\end{align}

By construction in Algorithm \ref{alg:graphedit}, we have $C_1(v) = \mathcal{C}_{dc}(v)$ for any original vertex $v \in V(G)$, and the new vertices have zero initial cost, i.e., $C_1(v_{sk}) = C_1(v_{sr}) = 0$. The weights are assigned as $w(u,v_{sk}) = w(u,v)$, $w(v_{sk},v_{sr}) = 1$, and $w(v_{sr},v) = 1$. The upstream neighbors of $v$ remain unchanged except that $u$ is replaced by $v_{sr}$.

Substituting Eq. \ref{eq:proof_vsk} into Eq. \ref{eq:proof_vsr}, we get:
$$
\mathbf{c}_1^*(v_{sr}) = 1 \cdot \mathbf{c}_1^*(v_{sk}) = w(u,v) \cdot \mathbf{c}_1^*(u)
$$
Now, substituting this result for $\mathbf{c}_1^*(v_{sr})$ into Eq. \ref{eq:proof_new_v}:
$$
\mathbf{c}_1^*(v) = \mathcal{C}_{dc}(v) + 1 \cdot (w(u,v) \cdot \mathbf{c}_1^*(u)) + \sum_{z \in U_G^1(v) \setminus \{u\}} w(z,v) \cdot \mathbf{c}_1^*(z)
$$

Assuming $\mathbf{c}^*(z) = \mathbf{c}_1^*(z)$ for all upstream neighbors, this recursive dependency reduces correctly to an equation identical to the original cost equation for $v$ (Eq. \ref{eq:proof_original_v}). 

Since the fundamental linear structure and aggregated cost formulations remain equivalent for every original vertex $v \in V(G)$ in the transformed system, the unique solution to the system must also be identical. Thus, the manufacturing costs for all data products $p \in L(G) \subset V(G)$ are preserved.
\end{proof}

\subsection{Proof of Theorem \ref{thm:aggregation_cost_preservation}}
\label{appendix:proof_aggregation_cost_preservation}

\begin{proof}
    Denote $V_1 = V(G) \setminus V_{\text{term}}$ and $V_2 = V_{\text{term}}$. Let $\mathbf{d}_1 = (\mathbf{d}_{11}, \mathbf{d}_{12})^{\top}$ denote the cost vector corresponding to $C_1$, where $\mathbf{d}_{11}$ and $\mathbf{d}_{12}$ represent the costs associated with the nodes in $V_1$ and $V_2$, respectively. Similarly, let $\mathbf{d}_2 = (\mathbf{d}_{21}, \mathbf{d}_{22})^{\top} = (\mathbf{0}, \mathbf{d}_{22})^{\top}$ denote the cost vector corresponding to $C_{DAG}$. The corresponding weight matrix is written as
    \[ \mathbf{W} = \begin{pmatrix}
            \mathbf{W}_{11} & \mathbf{W}_{12} \\ \mathbf{W}_{21} & \mathbf{W}_{22}
        \end{pmatrix}. \]
    Let $P$ denote the length of the longest path in $G_{DAG}$. Then, according to the process of \func{RollUpApp}, we have
    \begin{align*}
        \begin{pmatrix}
            \mathbf{0} \\ \mathbf{d}_{22}
        \end{pmatrix} &= \left[\begin{pmatrix}
            \mathbf{W}_{11} & \mathbf{W}_{12} \\ \mathbf{0} & \mathbf{I}
        \end{pmatrix}^{\top}\right]^P \begin{pmatrix}
            \mathbf{d}_{11} \\ \mathbf{d}_{12}
        \end{pmatrix} \\
        &= \begin{pmatrix}
            (\mathbf{W}_{11}^{\top})^P \mathbf{d}_{11} \\ \mathbf{W}_{12}^{\top}\left(\mathbf{I} + \mathbf{W}_{11}^{\top} + \cdots + (\mathbf{W}_{11}^{\top})^{P - 1}\right) \mathbf{d}_{11} + \mathbf{d}_{12}
        \end{pmatrix}.
    \end{align*}
    Let $\mathbf{c}_1 = (\mathbf{c}_{11}, \mathbf{c}_{12})^{\top}$ be the manufacturing cost vector of $C_1$, and let $\mathbf{c}_2 = (\mathbf{c}_{21}, \mathbf{c}_{22})^{\top}$ be the manufacturing cost vector of $C_{DAG}$. Then the linear systems for $(G_1, C_1)$ and $(G_1, C_{DAG})$ are
    \begin{align*}
        & \begin{pmatrix}
            \mathbf{I} - \mathbf{W}_{11}^{\top} & -\mathbf{W}_{21}^{\top} \\
            -\mathbf{W}_{12}^{\top} & \mathbf{I} - \mathbf{W}_{21}^{\top}
        \end{pmatrix} \begin{pmatrix}
            \mathbf{c}_{11} \\ \mathbf{c}_{12}
        \end{pmatrix} = \begin{pmatrix}
            \mathbf{d}_{11} \\ \mathbf{d}_{12}
        \end{pmatrix}, \\
        & \begin{pmatrix}
            \mathbf{I} - \mathbf{W}_{11}^{\top} & -\mathbf{W}_{21}^{\top} \\
            -\mathbf{W}_{12}^{\top} & \mathbf{I} - \mathbf{W}_{21}^{\top}
        \end{pmatrix} \begin{pmatrix}
            \mathbf{c}_{21} \\ \mathbf{c}_{22}
        \end{pmatrix} = \begin{pmatrix}
            \mathbf{0} \\ \mathbf{d}_{22}
        \end{pmatrix}.
    \end{align*}
    Note that
    \[ (\mathbf{I} - \mathbf{W}_{11}^{\top})(\mathbf{I} + \mathbf{W}_{11}^{\top} + \cdots + (\mathbf{W}_{11}^{\top})^{P - 1}) = \mathbf{I} - (\mathbf{W}_{11}^{\top})^P = \mathbf{I}. \]
    Therefore,
    \[ \mathbf{c}_{21} = \mathbf{c}_{11} - \left(\mathbf{I} + \mathbf{W}_{11}^{\top} + \cdots + (\mathbf{W}_{11}^{\top})^{P - 1}\right) \mathbf{d}_{11}, \quad \mathbf{c}_{22} = \mathbf{c}_{12}. \]
    is a solution to the linear system corresponding to $(G_1, C_{DAG})$. Since the solution to this linear system is unique, this must be its unique solution. The equality $\mathbf{c}_{22} = \mathbf{c}_{12}$ therefore implies the conclusion of Theorem \ref{thm:aggregation_cost_preservation}.
\end{proof}

\subsection{Proof of Theorem \ref{thm:costapp_correctness}}
\label{app:proof_costapp_correctness}

% We provide the detailedproof of Theorem \ref{thm:costapp_correctness} below.

\begin{proof}
Let $\mathbf{c}^*$ denote the exact manufacturing cost (MC) vector for all data products $p \in L(G)$ in the original data pipeline defined by the initial state $(G, \mathcal{C}_{dc})$. By Theorem~\ref{thm:matrixapp_exact}, $\mathbf{c}^*$ represents the exact ground truth that satisfies the global linear system formulation. We demonstrate that the sequence of transformations applied in \func{CostApp} strictly preserves this exact result.

First, in Phase 1, \func{CostApp} isolates cyclic dependencies via the $\func{GraphEdit}$ operation, yielding a modified graph $G_1$ and a corresponding cost state $C_1$. According to Theorem~\ref{thm:graph_edit}, removing feedback arcs and introducing auxiliary vertices does not alter the cost flows toward the leaves. Consequently, the MC of any data product in the new system defined on $(G_1, C_1)$ remains identically $\mathbf{c}^*$.

Second, by executing \func{RollUpApp} on the acyclic subgraph $G_{DAG}$, \func{CostApp} iteratively propagates costs until they are completely aggregated onto the terminal vertices $V_{\text{term}}$, resulting in the cost state $C_{DAG}$. By Theorem~\ref{thm:aggregation_cost_preservation}, propagating costs along an acyclic subgraph strictly conserves the ultimate allocation to final DPs. Thus, the MC of DPs in the system $(G_1, C_{DAG})$ is completely equivalent to that in $(G_1, C_1)$, continuing to equal $\mathbf{c}^*$.

Finally, in Phase 2, the aggregated costs $C_{DAG}$ on terminal nodes are extracted as the reduced direct cost vector $\mathbf{d}_{\text{term}}$. The residual cyclic dependencies are captured by the reduced cost allocation matrix $\mathbf{W}_{\text{term}}$. The cost redistribution over $V_{\text{term}}$ is thus formulated as the reduced linear system $\mathbf{c}_{\text{term}} = \mathbf{d}_{\text{term}} + \mathbf{W}_{\text{term}}^T \mathbf{c}_{\text{term}}$. Since Theorem~\ref{thm:matrixapp_exact} guarantees that solving such a linear formulation yields exact results, executing \func{MatrixApp$^*$} on this reduced system computes the exact MCs for the system $(G_1, C_{DAG})$. As $L(G) \subseteq V_{\text{term}}$, extracting the entries corresponding to $L(G)$ precisely recovers $\mathbf{c}^*$, which concludes the proof.
\end{proof}

\section{Matrix Form Derivation}
\label{sec:matrix_form_derivation}

To derive the matrix form of MC from the scalar equations, we start with the definition of MC for a single vertex $v_i$:

$$ \mathcal{C}_{mc}(v_i) = \mathcal{C}_{dc}(v_i) + \sum_{v_j \in U_G^1(v_i)} w(v_j, v_i) \cdot \mathcal{C}_{mc}(v_j) $$
This equation holds for all $n$ vertices in the graph $V=\{v_1, \dots, v_n\}$. We can express this system of $n$ linear equations in matrix form as follows:

Let $\mathbf{c}$ be the $n \times 1$ column vector of manufacturing costs, where the $i$-th element is $\mathbf{c}_i = \mathcal{C}_{mc}(v_i)$. Let $\mathbf{d}$ be the $n \times 1$ column vector of direct costs, where the $i$-th element is $\mathbf{d}_i = \mathcal{C}_{dc}(v_i)$. Let $\mathbf{W}$ be the $n \times n$ weighted adjacency matrix, where $\mathbf{W}[i,j] = w(v_i, v_j)$ if an arc exists from $v_i$ to $v_j$, and 0 otherwise.

The summation term represents the total input cost for vertex $v_i$. Let's analyze the matrix-vector product $\mathbf{W}^\mathsf{T} \mathbf{c}$. The $i$-th element of this product vector is computed as:

$$ (\mathbf{W}^\mathsf{T} \mathbf{c})_i = \sum_{j=1}^{n} (\mathbf{W}^\mathsf{T})_{ij} \cdot \mathbf{c}_j = \sum_{j=1}^{n} \mathbf{W}[j,i] \cdot \mathbf{c}_j = \sum_{j=1}^{n} w(v_j, v_i) \cdot \mathcal{C}_{mc}(v_j)$$

% \sum_{j=1}^{n} w(v_j, v_i)对应矩阵W中的一列值，\mathcal{C}_{mc}(v_j)对应一个成本向量。矩阵中应该是行向量与成本向量相乘，所以W应该是转置。直观理解： W_i表示i节点分到下游的所有比例，W^T_i表示i节点被上游分到的所有比例
Since $w(v_j, v_i)$ is non-zero only when $v_j$ is a direct predecessor of $v_i$ (i.e., $v_j \in U_G^1(v_i)$), this sum is exactly equivalent to the input cost term in the original scalar equation.

We can now rewrite the scalar equation for each vertex $v_i$ using this vector notation:

$$ \mathbf{c}_i = \mathbf{d}_i + (\mathbf{W}^\mathsf{T} \mathbf{c})_i $$
Since this relationship holds for every element $i=1, \dots, n$, we can express it for the entire system of equations in vector form:
$$
\begin{pmatrix} \mathbf{c}_1 \\ \mathbf{c}_2 \\ \vdots \\ \mathbf{c}_n \end{pmatrix}
=
\begin{pmatrix} \mathbf{d}_1 \\ \mathbf{d}_2 \\ \vdots \\ \mathbf{d}_n \end{pmatrix}
+
\begin{pmatrix} (\mathbf{W}^\mathsf{T} \mathbf{c})_1 \\ (\mathbf{W}^\mathsf{T} \mathbf{c})_2 \\ \vdots \\ (\mathbf{W}^\mathsf{T} \mathbf{c})_n \end{pmatrix}
$$
This simplifies to the final matrix equation:
$$ \mathbf{c} = \mathbf{d} + \mathbf{W}^\mathsf{T} \mathbf{c} $$

\section{The \func{label} Algorithm}
\label{appendix:label}

The \func{Label} algorithm is shown in Alg.~\ref{alg:label}. Roots receive unique integer labels, while every other vertex starts with $\top$, a value larger than any valid label. A vertex adopts a candidate only when it is smaller than its current label. After an update, it appends an outgoing-edge suffix and sends the new labels downstream. The process stops when no label changes.

\begin{algorithm}[htb]\small
    \SetKwBlock{emptyOp}{ }{end}
    \SetKwBlock{compute}{\func{compute}(msgs)}{end}
    \caption{\func{Label} (Pregel-style pseudocode)}
    \label{alg:label}
    \KwIn{Graph $G=(V, A)$.}
    \KwOut{The final label $L(v)$ for each vertex $v \in V$.}

    \tcc{\darkBlue{Initialization: Assign unique integer labels to roots}}
    \textbf{foreach} root vertex $v_i \in R(G)$ in parallel: $L(v_i) \gets i$\;
    {\color{purple}\textbf{foreach} non-root vertex $v \in V \setminus R(G)$: $L(v) \gets \top$ \tcp*{Larger than any valid label}}
    {\color{purple}\textbf{foreach} root $v_i$ and its $j$-th downstream neighbor $v_j$: send $i+``.j"$ to $v_j$}\;

    \While{any vertex is active}{
        \textbf{foreach} active vertex $v$ \textbf{do} in parallel \emptyOp {
            \compute {
                $label\_updated \gets \text{false}$\;
                $L_{best\_incoming} \gets L(v)$\;
                
                \ForEach{message containing label $L_{msg}$}{
                    \If{$L_{msg} <_{lo} L_{best\_incoming}$}{
                        $L_{best\_incoming} \gets L_{msg}$\;
                    }
                }
                
                \If{$L_{best\_incoming} <_{lo} L(v)$}{
                    $L(v) \gets L_{best\_incoming}$\;
                    $label\_updated \gets \text{true}$\;
                }

                \If{$label\_updated$}{
                    \tcc{\darkBlue{Propagate new labels to downstream neighbors}}
                    \ForEach{$j$-th neighbor $v_j \in D_G^1(v)$}{
                        send message with label $L(v) + ``.j"$ to $v_j$\;
                    }
                }
                \func{voteToHalt}()\; \tcp{\darkBlue{Reactivated upon receiving a message}}
            }
        }
    }
    \Return final labels $L(v)$ for all $v \in V$\;
\end{algorithm}

\section{A Formulation of \func{GraphEdit}}
\label{appendix:graphedit}

Alg.~\ref{alg:graphedit} presents the pseudocode for the \func{GraphEdit} algorithm. This algorithm takes as input a directed graph $G$ and a feedback arc set $A_{\text{FAS}}$. For each arc in $A_{\text{FAS}}$, it removes the arc from $G$ and introduces two new auxiliary vertices: a sink vertex and a source vertex. The original arc is replaced by three new arcs: one from the original source to the new sink, one from the new source to the original target, and one connecting the new sink to the new source. The weights of these new arcs are assigned as specified in the algorithm.

\begin{algorithm}[t]\small
    \caption{\func{GraphEdit}}
    \label{alg:graphedit}
    \KwIn{Directed graph $G = (V, A)$, direct costs $\mathcal{C}_{dc}(v)$ for all $v \in V$, feedback arc set $A_{\text{FAS}} \subseteq A$}
    \KwOut{Edited graph $G_1 = (V_1, A_1)$, new cost states $C_1(v)$ for all $v \in V_1$, set of new sink vertices $V_{sk}$, set of new source vertices $V_{sr}$}

    $G_1 \gets G$, $C_1 \gets \mathcal{C}_{dc}$, $V_{sk} \gets \emptyset$,$V_{sr} \gets \emptyset$\;
    \ForEach{arc $f = \langle u,v\rangle \in A_{\text{FAS}}$}{
        Remove arc $f$ from $G_1$\;
        Create new vertex $v_{sk}$ and add to $V_1$ and $V_{sk}$\;
        Create new vertex $v_{sr}$ and add to $V_1$ and $V_{sr}$\;
        Add new arc $\langle u,v_{sk}\rangle$ with weight $w(u,v)$ to $A_1$\;
        Add new arc $\langle v_{sk},v_{sr}\rangle$ with weight $1$ to $A_1$\;
        Add new arc $\langle v_{sr},v\rangle$ with weight $1$ to $A_1$\;
        $C_1(v_{sk}) \gets 0$\;
        $C_1(v_{sr}) \gets 0$\;
    }
    \Return $(G_1, C_1, V_{sk}, V_{sr})$ \;
\end{algorithm}

\section{The \func{RollUpApp$^*$} Algorithm}
\label{appendix:simpleapp_star}

We present the \func{RollUpApp$^*$} algorithm in Alg. \ref{alg:simpleapp*-pregel}, a Pregel-style pseudocode variant of \func{RollUpApp} (Alg. \ref{alg:simpleapp-pregel}). This variant utilizes a stateful cost map $C(v)$ at each vertex $v$ to accumulate costs for multiple keys simultaneously. During each superstep, active vertices aggregate incoming cost maps from upstream neighbors into their local cost map and then propagate weighted versions of this aggregated map to their downstream neighbors. This approach allows for efficient handling of multiple cost dimensions in a single pass through the graph.

\begin{algorithm}[t]\small
    \SetKwBlock{CAOp}{ }{end}
    \SetKwBlock{compute}{\func{compute}(msgs)}{end}
    \caption{\func{RollUpApp$^*$} (Pregel-style pseudocode)}
    \label{alg:simpleapp*-pregel}
    \KwIn{Graph $G=(V, A, w)$, stateful cost map $C(v)$}
    \KwOut{$C^{out}(v)$}

    Each vertex $v \in V$ sets $C^{out}(v) \gets C(v)$ \tcp*{\darkBlue{Initialization}}

    \While{any non-leaf vertex $v$ is active}{
        
        \textbf{foreach} active vertex $v$ \textbf{do} in parallel \CAOp {
            \tcc{\darkBlue{\func{CR} operation in Pregel style}}
            \compute {

                \textbf{foreach} $m \in msgs$ \textbf{do}  \CAOp {
                    \textbf{foreach} $key \in m.\text{getKeys}()$ \textbf{do}  \CAOp {
                        update $C^{out}(v)[key] \gets C^{out}(v)[key] + m.\text{getValue}(key)$\;
                    }
                }
                sends $w(v, v') \cdot C^{out}(v)$ to each $v' \in D_G^1(v)$  \tcp*{\darkBlue{Send weighted map to downstream neighbors}}
                
                updates $C^{out}(v) \gets \emptyset$ \;
            
                \func{voteToHalt}()\;

                \tcp{\darkBlue{Reactivated upon receiving a message}}
            }
        }
    }
    \Return $C^{out}(v)$
\end{algorithm}

% Once the vector $\mathbf{c}_{sk}$ is known, the final manufacturing costs of the data products, $\mathbf{c}_{dp}$, can be computed directly by substituting the result into Eq. \ref{eq:dp_cost_reordered}. This is a straightforward matrix-vector multiplication with no further matrix inversion required:
% $$
% \mathbf{c}_{dp} = \mathbf{o}_{dp} + \mathbf{W}_{sk,dp}^T \mathbf{c}_{sk}
% $$
% This two-step procedure—solving a smaller system for sink nodes and then directly calculating DP costs—is computationally more efficient than solving the full $(|V_{sk}| + |L(G)|) \times (|V_{sk}| + |L(G)|)$ system directly. The complexity of the dominant matrix inversion step is reduced from $O((|V_{sk}| + |L(G)|)^3)$ to $O(|V_{sk}|^3)$, followed by a less expensive matrix multiplication. This optimization is particularly effective when the number of data products is large relative to the number of feedback arcs.

\section{Extended Future Work}
\label{app:future_work}

We discuss some promising directions to inspire future work.

% \noindent\textit{(1) Explainable product costing in data pipelines.}
\paragraph{(1) Explainable product costing in data pipelines.}
For fine-grained cost optimization, enterprises require a detailed understanding of how each upstream data entity contributes to the cost of a final data product. This can be achieved by maintaining a map that tracks the propagation of monetary costs from each component through the pipeline. This method yields a detailed breakdown of each data product, attributing its total cost to its constituent components and enabling explainability. This idea will incur additional storage and communication overhead. We leave the concrete design and optimization of solutions to this explainability as future work.

% \noindent\textit{(2) Exploration of potential applications of \func{CostApp}.} 
\paragraph{(2) Exploration of potential applications of \func{CostApp}.} 
Although motivated by data pipelines, \func{CostApp}'s core capability, which is solving large-scale sparse linear systems with cyclic structures, has the potential to apply to broader problems. One potential application is solving Bellman equations in Markov Decision Processes (MDPs) with sparse state transition matrixs. In an MDP, the value function $\mathbf{v}$ satisfies the Bellman equation $\mathbf{v} = \mathbf{R} + \gamma \mathbf{P} \mathbf{v}$, where $\mathbf{R}$ is the reward vector, $\mathbf{P}$ is the state transition matrix, and $\gamma$ is the discount factor. Rearranging this equation yields $\mathbf{v} = (\mathbf{I} - \gamma \mathbf{P})^{-1} \cdot \mathbf{R}$, which aligns with the form in \S\ref{sec:matrix_app}. A common method for evaluating the value function is through value iteration \cite{sutton2018reinforcement}, which can be proved to be equivalent to the cost iteration of \func{RollUpApp} for $L(G)$. The experimental results in \S\ref{sec:experiments} suggest that \func{CostApp} obtains exact solutions faster than \func{RollUpApp} achieves solutions of acceptable accuracy, indicating its potential to accelerate value function computation in MDPs with sparse transitions and large state spaces. Another potential application is input-output economic modeling \cite{leontief1986input,leontief1974sructure,ten2009input}, where industries are interdependent with cyclic supply chains. In this model, $\mathbf{x}$ denotes the total production required and can be obtained by resolving a Leontief inverse: $\mathbf{x} = (\mathbf{I} - \mathbf{A})^{-1} \cdot \mathbf{d}$, where $\mathbf{d}$ is a demand vector and $\mathbf{A}$ is an input-output matrix, which is naturally large-scale and sparse. Thus, \func{CostApp} can be directly applied to efficiently compute $\mathbf{x}$ on large-scale economic networks without approximations. These examples demonstrate that \func{CostApp} can inspire the optimization of numerical solvers for large-scale sparse linear systems, making it valuable future work to explore more potential applications.

\end{document}